\documentclass[accepted]{uai2023} 

\usepackage[american]{babel}
\usepackage{float}
\usepackage{graphicx}
\usepackage{enumitem}
\usepackage{amsmath}
\usepackage{amsthm}
\usepackage{amsfonts}
\usepackage{mathtools}
\usepackage{caption}
\usepackage{subcaption}

\usepackage{natbib} 
\usepackage{mathtools} 
\usepackage{booktabs} 
\usepackage{tikz} 

\usepackage{xr}
\makeatletter
\newcommand*{\addFileDependency}[1]{
  \typeout{(#1)}
  \@addtofilelist{#1}
  \IfFileExists{#1}{}{\typeout{No file #1.}}
}
\makeatother
\newcommand*{\myexternaldocument}[1]{%
    \externaldocument{#1}%
    \addFileDependency{#1.tex}%
    \addFileDependency{#1.aux}%
}
\myexternaldocument{uai2023-supplement}

\newtheorem{theorem}{Theorem}[section]
\newtheorem{corollary}{Corollary}[theorem]
\newtheorem{lemma}[theorem]{Lemma}
\newtheorem{prop}{Proposition}
\newtheorem{assumption}{Assumption}

\newtheorem{definition}[theorem]{Definition}
\usepackage{siunitx}

\allowdisplaybreaks

\title{ Short-term Temporal Dependency Detection under Heterogeneous Event Dynamic with Hawkes Processes }

\author[1,*]{\href{mailto:<albertchenyu@gmail.com>?Subject=Your UAI 2023 paper}{Yu Chen}{}}
\author[1,*]{Fengpei Li}
\author[1]{Anderson Schneider}
\author[1]{Yuriy Nevmyvaka}
\author[2]{Asohan Amarasingham}
\author[3]{Henry Lam}

\affil[1]{%
    Machine Learning Research, Morgan Stanley, New York, NY
}
\affil[2]{%
    Department of Mathematics, The City College of New York, New York, NY
}
\affil[3]{%
    Department of Industrial Engineering \& Operations Research, Columbia University, New York, NY
}
\affil[*]{%
    Authors have equal contribution
}

\begin{document}
\raggedbottom

\maketitle

\begin{abstract}
\vspace{-0.2in}
Many \textit{event sequence} data exhibit mutually exciting or inhibiting patterns. Reliable detection of such temporal dependency is crucial for scientific investigation. The \textit{de facto} model is the Multivariate Hawkes Process (MHP), whose impact function naturally encodes a causal structure in Granger causality. However, the vast majority of existing methods use direct or nonlinear transform of \textit{standard} MHP intensity with constant baseline, inconsistent with real-world data. Under irregular and unknown heterogeneous intensity, capturing temporal dependency is hard as one struggles to distinguish the effect of mutual interaction from that of intensity fluctuation. In this paper, we address the short-term temporal dependency detection issue. We show the maximum likelihood estimation (MLE) for cross-impact from MHP has an error that can not be eliminated but may be reduced by order of magnitude, using heterogeneous intensity not of the target HP but of the interacting HP. Then we proposed a robust and computationally-efficient method modified from MLE that does not rely on the prior estimation of the heterogeneous intensity and is thus applicable in a data-limited regime (e.g., few-shot, no repeated observations). Extensive experiments on various datasets show that our method outperforms existing ones by notable margins, with highlighted novel applications in neuroscience.
\end{abstract}


\section{Introduction}\label{sec:intro}

A substantial amount of timestamp data come as a sequence of irregular and asynchronous events. They are recorded in continuous time and observed in domains such as computational biology (e.g., action potential/neuron spike trains \citep{kass2001spike,pillow2008spatio}, genomic events \citep{reynaud2010adaptive}), quantitative finance (e.g., limit order book modeling for high-frequency trading \citep{bacry2015hawkes,bowsher2007modelling}, credit risk modeling \citep{errais2010affine}), social network analysis (e.g., social media user activity \citep{farajtabar2015coevolve,zhou2013learning}), e-healthcare (\citep{ wang2018supervised}) or seismology (e.g., earthquake aftershock \citep{ogata1988statistical}). Besides asynchronicity, the sequence data often exhibit mutual interaction patterns where the occurrence of one event can excite or inhibit that of another. For example, the news-driven trading in behavioral finance studied the mutual excitation between investor sentiment shocks and the negative price jumps \citep{yang2018applications}, while in the cortical network, inhibitory connectivity in firing-rate between neurons and synapse controls memory maintenance \citep{mongillo2018inhibitory}. These interaction patterns have been called: \textit{temporal dependency} \citep{zuo2020transformer}, \textit{cross-correlation} \citep{zhang2020self}, \textit{coupling effect} \citep{pillow2008spatio} or \textit{Granger causality} \citep{xu2016learning}. As noted in \citep{eichler2017graphical}, although stand-alone notions of Granger causality can not establish causal-effect links, the detection of temporal dependency remains useful for prediction or scientific investigation. 

Temporal point process (TPP)  \citep{cox1980point} is a powerful tool for modeling event sequence and Multivariate Hawkes process (MHP) \citep{hawkes1971point}, as a special type of TPP, has been widely used as the \textit{de facto} tool for capturing the temporal dependency among events (see above, e.g.,\citep{bacry2015hawkes,farajtabar2015coevolve,wang2022hawkes,zuo2020transformer}). An MHP models the occurrence probability using a history-dependent conditional \textit{intensity} and its impact function is particularly well-suited to detect the mutual excitatory effect. The inhibitory effect can also be incorporated for MHP, but some nonlinear link function is required to map the MHP intensity into $\mathbb R^+$ (e.g., notably a clipping function $x^+=\max(x,0)$ in \citep{10.3150/13-BEJ562} or sigmoid function in \citep{zhouj2021efficient}).

Despite the expressiveness of \textit{impact functions} (also called \textit{coupling filter}, \textit{trigger kernel}, \textit{influence function}, see \citep{pillow2008spatio, zhouj2021efficient, zhou2013learning}), the background component in MHP intensity is assumed to be time-invariant. Possibly due to the extra modeling difficulty entailed, virtually all existing studies on MHP, implicitly or explicitly, use direct or nonlinear transform of standard MHP intensity with constant baseline, including the modern DL-based methods (e.g., Transformer HP \citep{zuo2020transformer} HP in infinite relational model or Dirichlet mixture model \citep{blundell2012modelling,xu2017dirichlet}, sigmoid nonlinear MHP with P\'olya-Gamma variable augmentation \citep{zhouj2021efficient}, self-attentive HP and recurrent neural network \citep{zhang2020self}). Notable exceptions with heterogeneity include \citep{mei2017neural} with neurally self-modulating HP with LSTM and \citep{zhou2021nonlinear} where a state switching latent process is proposed (yet still assuming constant background within each state) and \citep{hawkes2018hawkes} where the heterogeneous background is briefly discussed as a generalization of MHP to represent "exogenous economic activity". 

However, real-world event dynamics are often decisively \textit{temporally heterogeneous}. For example, Twitter has information bursts spurred by exogenous events (e.g., breaking news or sports games)\citep{wang2022hawkes}, the firing of neurons is commonly driven by varying visual stimuli \citep{siegle2021survey}, and trading activity has a diurnal variation (e.g., more trades occur around market open/close than around noon \citep{bowsher2007modelling}). Under unknown heterogeneous dynamics, temporal dependency detection is difficult as one struggles to distinguish the effect of mutual interaction from that of background intensity fluctuation (e.g., did the arrival of orders for stock A stimulate that for stock B, or did they both simply experience a spike in trading activity?).

In this paper, we show that the maximum likelihood estimation (MLE) from standard MHP for short-term temporal dependency detection has non-negligible errors under heterogeneous event dynamics. However, this error decreases by an order of magnitude (in terms of impact window or kernel width) if the heterogeneous background between the \textit{target} HP (recipients of the impact) and \textit{source} HP (initiators of the impact) is \textit{uncorrelated} (or \textit{orthogonal} in the Hilbert space sense, $L_2[0,T]$ or $C[0,T]$, where $T$ is observation horizon). Thus, loosely speaking, MHP can still estimate short-term cross-impact reasonably well, until the heterogeneous intensity between the target HP and source HP shares a common/correlated varying background. Building on this insight, we proposed a robust and computationally-efficient method modified from MLE, which utilizes a nonparametric estimate of heterogeneous intensity not of the target HP but of the source HP. By focusing on the background intensity of the source, we not only reduce the inference difficulty due to the coupling between the target HP background and impact function but also the error by regressing the common varying background out of the target HP intensity. 

The contribution of this paper can be summarized as:
\begin{itemize}
    \item To the best of our knowledge, our work is first to formally report and analyze the error of MLE from MHP in short-term temporal dependency detection due to heterogeneous background. We investigate the relation between the error and the correlation (or inner product) of backgrounds among interacting HPs, upon which we propose a novel method to reduce the error.

    \item Through extensive experiments, we show our method exhibits superior performance and is robust, cost-efficient, applicable in a data-limited regime (e.g., lack of repeated observations), and suitable for inference.
    
    \item Finally, we apply our method on mouse visual cortices data and discover the distant interactions between neurons on a fine timescale (e.g., within 50 ms) in both top-down and bottom-up pathways, showcasing its novelty and direct applicability in neuroscience.

\end{itemize}

\section{Related Work}
\textbf{Hawkes process.} Many efforts have been devoted to detecting the temporal dependency among point processes, e.g. \citep{chwialkowski2014kernel,gunawardana2011model}. Among point processes, Hawkes processes stand out as the \textit{de facto} tool for modeling complex temporal dependency in event sequences. The paper \citep{eichler2017graphical} established the link between Granger causality and impact function in MHP and many methods are proposed to learn the temporal dependency in MHP, via group sparsity, \citep{xu2016learning}, nonparametric learning using Euler-Lagrange equation \citep{zhou2013learningtrigger}, isotonic nonlinear link function \citep{wang2016isotonic}, online learning \citep{yang2017online} and modern DL-based methods (see intro,  \citep{zuo2020transformer,blundell2012modelling,xu2017dirichlet,zhang2020self}.  However, these methods use direct or nonlinear transform of \textit{standard} MHP time-invariant base intensity, overlooking the heterogeneity in event dynamics. Notably, \citep{mei2017neural} implicitly allows for heterogeneity. Latent variable augmentation is proposed in \citep{zhouj2021efficient,zhou2021nonlinear,zhou2022efficient,zhou2020efficient} to incorporate the time-varying background, but the modeling of heterogeneity typically relies on piecewise constants. Moreover, most methods are data-intensive (e.g., as reported in \citep{yang2017online}, methods as \citep{zhou2013learningtrigger} require more than $10^5 d$ arrival data to obtain good results on $d\leq 5$ event streams) and computationally-extensive (e.g., MCMC, EM algorithm or complex neural architecture) which is unsuitable for inference in data-limited regimes. Indeed, often in practice, only short/unrepeated sequences are available \citep{salehi2019learning}, which not only amplifies the
risk of overfitting but also makes estimation of heterogeneous background infeasible.

\begin{figure}[ht]
\centering
\begin{subfigure}{0.47\textwidth}
\centering
\includegraphics[width=0.99\linewidth]{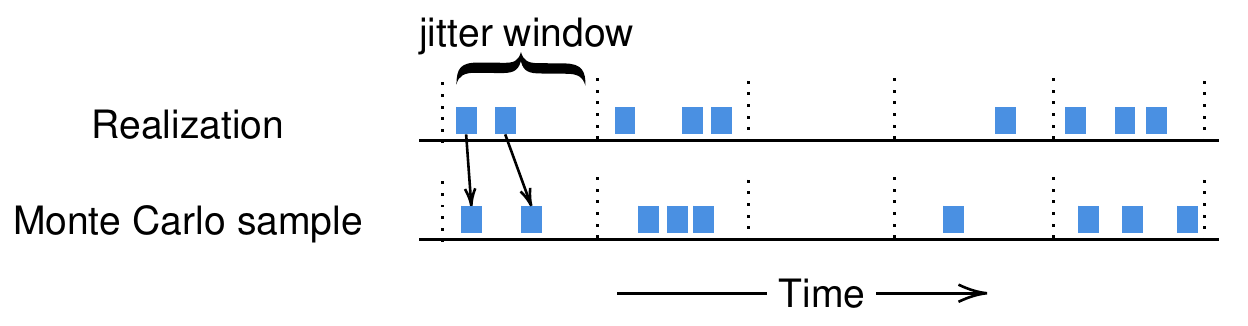}
\end{subfigure}
\vspace{-0.17in}
\caption{\small{Construction of Monte Carlo samples from the null in condition inference based CCG. Blue dots are timestamps}.
}
\vspace{-0.2in}
\label{fig:jitter_demo}
\end{figure}

\textbf{Conditional inference based cross-correlogram (CCG).} Heterogeneous dynamic is ubiquitous in neuroscience \citep{farajtabar2015coevolve}. Due to the limitations of TPP and MHP in this regime, a popular method in neuroscience for detecting temporal dependency is cross-correlogram (CCG) via conditional inference. Conditional hypothesis testing with a carefully designed null hypothesis can bypass the background heterogeneity issue. Particularly, given realizations of two point processes, CCG assesses the correlation of timing in the temporal dependence between events, conditioning on the unobserved fluctuating background activity. As shown in \citep{amarasingham2012conditional}, the method relies on conditional inference and bootstrap, where Monte Carlo samples from the null are generated by: shifting the timestamps within each \textit{jitter window} (prior knowledge on the timescale of interactions) by a random amount, small enough to preserve the local background intensity, but big enough to break the fine time interaction pattern (see Figure \ref{fig:jitter_demo}). However, this method requires prior knowledge of timescales and only works when the timescale of background activity frequency (jitter window width) is larger than that of the interaction effect (See Figure \ref{fig:sim_demo} and discussion below. Additional details in Appendix 
\ref{appendix:ccg}, \ref{appendix:small_sigma_I}).
Also, the outcome of the hypothesis test does not measure the strength of the coupling effect directly.



\section{Analysis and Methods }
\subsection{Basic Concepts}
A temporal point process is a stochastic process whose realization consists of a list of discrete event timestamps $\{t_n\}_{n\in\mathbb N}\subseteq \mathbb R^+$, which can be equivalently represented by a counting process $\{N(t), 0\leq t\leq T\}$ \citep{daley2008introduction}. Formally, given a probability triple $(\Omega, \{\mathcal H_t\}_{0\leq t\leq T},\mathbb P)$, $N(t):=N((0,t],\omega) $ is a realization (i.e., $\omega\in\Omega$) of counting measure $N$ for the number of points in $(0,t]$ and $\mathcal H_t$ is the $\sigma$-algebra generated from $N(B)$ for Borel subsets $B\subseteq(0,t]$ ( or $(-\infty,t]$, we do not distinguish them here). The intensity of the point process is
$\lambda(t):=\lim_{\delta\to 0}\frac{1}{\delta} \mathbb P(N(t+\delta)-N(t)>0|\mathcal H_t)$. It can be shown (see \citep{ogata1978asymptotic}) for $\mathcal H_t$-progressively measurable $\lambda(t),f(t)$ with left continuous (thus predictable) $f(t)$ that $\mathbb E[dN(t)|\mathcal H_t]=\lambda(t)dt$ and
\begin{align}\label{ogte}
    \mathbb E\int_0^T f(t)dN(t) = &\mathbb E\int_0^T f(t)\mathbb E[dN(t)|\mathcal H_t]\nonumber\\
   =&\mathbb E\int_0^T f(t)\lambda(t)dt.
\end{align}

For the multivariate Hawkes process, the density has the form
\begin{equation}\label{standardmhp}
    \lambda_j(t) = \alpha_j +\sum_{i=1}^d \int_0^t h_{i\to j}(t-s)dN_i(s)
\end{equation}
for $1\leq i,j\leq d$, where $d$ is the dimension (number of event streams), $\alpha_j$ is the \textit{baseline} intensity for process $N_j$ and $h_{i\to j}$ is the impact function from $N_i$ to $N_j$. Standard MHP models mutual excitatory behavior and requires $h_{i\to j}\geq 0$ to avoid negative intensity which is meaningless. However, one can simply set $\lambda \leftarrow \max(\lambda,0)$ \citep{10.3150/13-BEJ562} to extend MHP for modeling mutual inhibitory behavior.

\subsection{Heterogeneous event dynamic}

The standard MHP assumes the baseline intensity $\alpha$ to be a constant \eqref{standardmhp}, which is incongruous with the heterogeneous event dynamic frequently observed in real-world scenarios. To accommodate heterogeneity, instead of using \eqref{standardmhp} as building blocks to construct a complex structure, we directly proposed a generalized MHP intensity for $1\leq i,j \leq d$:
\begin{equation}\label{gMHP}
    \lambda_j(t) = \alpha_j +f_j(t)+\sum_{i=1}^d \int_0^t h_{i\to j}(t-s)dN_i(s)
\end{equation}
where $f_j(t)$ is the fluctuation in the background intensity. For now, we do not restrict whether $f_j$ is stochastic or deterministic, but simply assume it is $\mathcal H_t$-adapted. For identifiability between $\alpha$ and $f$, we assume $\int_0^T\mathbb E[f(t)]dt=0$ (or more generally $\int_0^P f(t)dt=0$ if it is deterministic and perodic with peroid $P$ or $\mathbb E[f]=0$ if $f(t)$ is stationary).

The main approaches for learning MHP falls under two directions: maximum likelihood-based (MLE) approaches \citep{ogata1978asymptotic,zhou2013learning,yang2017online} and moment-matching flavored approaches based on higher-order statistics \citep{da2014hawkes}. Due to the unknown statistical property of $f$, the moment-based methods are not applicable for \eqref{gMHP}. To investigate the applicability of the MLE approach for \eqref{gMHP}, we study a representative model for subsequent discussion. However, we emphasize that our proposed method applies generally to models from \eqref{gMHP}.

\subsection{Representative Model} \label{sec:representatie_model}
Consider two point processes $N_i$, $N_j$ as shown in Figure \ref{fig:couple_process_diagram}. 
The intensity functions are,
\begin{equation}
\begin{aligned}
\lambda_{j}(t)
&= \alpha_j
+ f_{j}(t) 
+ \int_0^t h_{i\to j}(t-s) dN_{i}(s) \\
%
\lambda_{i}(t) 
&= \alpha_i
+ f_{i}(t)
\end{aligned}
\label{eq:true_model}
\end{equation}

\label{subsec:jitter_model}
\begin{figure}[ht]
\centering
\begin{subfigure}{0.48\textwidth}
\centering
\includegraphics[width=0.7\linewidth]{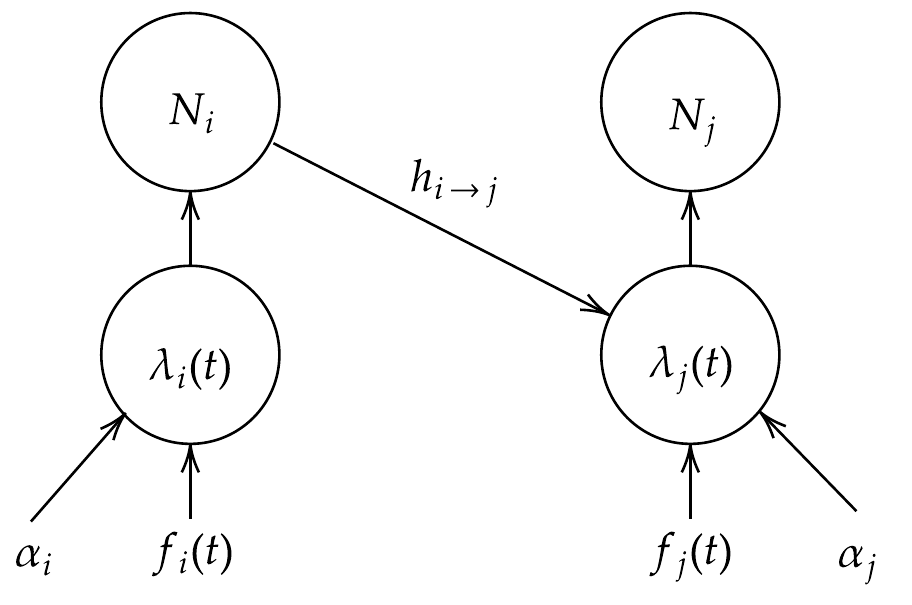}
\end{subfigure}
\vspace{-0.18in}
\caption{\small Illustrative MHP with a heterogeneous background. Two events stream $N_i, N_j$ with intensities $\lambda_i(t), \lambda_j(t)$, baseline $\alpha_i+f_{i}(t), \alpha_j+f_{j}(t)$ and the one-way impact function $h_{i\to j}$. }
\label{fig:couple_process_diagram}
\vspace{-0.15in}
\end{figure}

where $h_{i\to j}$ is the impact function and $f_{i}(t), f_{j}(t)$ are unknown fluctuations. There are various methods of learning 
the form $h_{i\to j}$ with data-driven and nonparametric techniques (\citep{zhou2013learningtrigger,xu2016learning,yang2017online}. To facilitate the discussion of MLE, we assume the form of impact has been learned within a 1-D parametric family $h_{i\to j}(\cdot)\in\{\theta\cdot \mathbf{1}_{[0,\sigma_h]}(\cdot)\}_{\theta\in\Theta}$ which is widely applied in neuroscience (here, $\mathbf{1}_{[0,\sigma_h]}(t) = 1$ if $0\leq t \leq \sigma_h$ and 0 otherwise). We set the ground truth impact to be $h_{i\to j}=c\cdot \mathbf{1}_{[0,\sigma_h]}$ for a given $c>0$. Moreover, we focus on the recovery of impact function (i.e., estimation of $c$) and treat other parameters as \textit{nuisance} parameters, as in \textit{profile likelihood}\citep{murphy2000profile}.

In MHP \eqref{standardmhp}, not considering $f_j$, one  parameterizes $\lambda_j$ as
\begin{equation}\label{like}
    \lambda_{\boldsymbol\theta }(t) = \theta_1 + \theta_2\int_0^t\mathbf{1}_{[0,\sigma_h]}(t-s)dN_i(s),
\end{equation}
which is misspecified and maximizes the log-likelihood:
\begin{align*}
   \hat{\boldsymbol{\theta}}=&\operatorname*{argmin}_{\boldsymbol\theta }\ell({\boldsymbol\theta }; \mathcal H_T)\nonumber\\
   =& -\int_0^T \lambda_{\boldsymbol\theta }(t)dt+\int_0^T \log\lambda_{\boldsymbol\theta }(t)dN_j(t),
\end{align*}
see, e.g., \citep{ogata1978asymptotic}. In the misspecified model, one would expect $\hat{\boldsymbol{\theta}}$ converges to $\boldsymbol\theta_{KL} $, the minimizer in KL-divergence information criterion \citep{white1982maximum}:
\begin{equation*}
    \boldsymbol\theta_{KL} =\operatorname*{argmin}_{\boldsymbol\theta }\Lambda(\boldsymbol\theta):=\mathbb E\ell({\boldsymbol\theta }),
\end{equation*}
under suitable regularity conditions, including $\mu$-strong convexity and $L$-Lipschitz gradient of $\Lambda$. We want to quantify the error between $[\boldsymbol\theta_{KL}]_2$ and $c$. We list technical conditions in Appendix \ref{appendix:proofs}, along with proofs for the following results.

\begin{prop}\label{prop1}
    Under regularity conditions specified in Appendix, for deterministic $f_i$ and $f_j$ in \eqref{gMHP}, the error satisfies
    \begin{equation*}
        |[\boldsymbol\theta_{KL}]_2-c|=\Theta\bigg(\bigg|\int_0^T\frac{ f_i(t)f_j(t)}{\alpha_j+c} dt\cdot \sigma_h+ o(\sigma_h)\bigg|\bigg).
    \end{equation*}
\end{prop}
\begin{prop}\label{prop2}
    Under the same condition as in Proposition \ref{prop1}, if $f_i$ and $f_j$ are stationary, the error satisfies 
    \begin{equation*}
        |[\boldsymbol\theta_{KL}]_2-c|=\Theta\bigg(\bigg|\frac{\text{Cov}(f_i,f_j)}{ \alpha_j +c}\sigma_h + o(\sigma_h)\bigg|\bigg).
    \end{equation*}
\end{prop}
where the big-$\Theta$ notation stands for growth function with the same rate in upper and lower bound, i.e. $f(x)=\Theta(g(x))$ if there exists $0\leq m\leq M$ s.t. $mg(x)\leq f(x)\leq Mg(x), \forall x$.

Proposition \ref{prop1} and \ref{prop2} suggest that, under heterogeneous event dynamics, the error in estimating impact function scales linear with $\sigma_h$, with the coefficient determined by the "inner product" between $f_i$ and $f_j$. In fact, if we define $\langle f_i,f_j\rangle = \mathbb E\int_0^T f_i(t)f_j(t)dt$, then we can unify (and generalize to a non-stationary case) the result in Proposition \ref{prop1} and \ref{prop2}. We see that, for short-term temporal dependency detection (i.e., $\sigma_h \to 0$), the ratio between estimation error and interaction timescale $\sigma_h$ is non-vanishing and non-negligible unless the two HPs have \textit{uncorrelated} background (i.e., $\langle f_i,f_j\rangle =0$).

How could one reduce the order of this error term? The most natural way is to observe or estimate $f_j$ directly. Indeed, given access to $f_j$, MLE is no longer misspecified. However, as discussed in \citep{zhou2020efficient}, the "exogenous component" (i.e., the baseline intensity) and the "endogenous" component (i.e., impact function) are "coupled" in the likelihood that hampers inference. In \citep{zhou2020efficient}, a \textit{branching} structure is used to decouple these two components in HP, which does not apply to MLE because when the same, typically limited data are used to estimate both $f_j$ and $h_{i\to j}$, the results are generally non-reliable (indeed, a naive use of MLE for fitting both would result in delta measures around the event timestamps for $N_j$). However, since the correlation between $f_i$ and $f_j$ results in a large error, one conjectures whether estimation of $f_i$, or entities highly correlated with $f_i$, could help regress out the common varying intensity out of $f_j$. Indeed, we have the following:
\begin{prop}\label{prop3}
    Under the same condition as in Proposition \ref{prop2}, if we let $r:=\max\{\|f_i-\mathbb Ef_i\|_\infty,\|f_j-\mathbb Ef_j\|_\infty\}$, if we have access to $g = \frac{ f_i-\mathbb E[f_i]}{\sqrt{\text{Var}(f_i)}}$ (i.e., normalized basis for $f_i$) in the likelihood \eqref{like} so that one parameterizes
    \begin{equation}\label{like2}
    \lambda_{\boldsymbol\theta }(t) = \theta_0 + \theta_1 g+\theta_2\int_0^t\mathbf{1}_{[0,\sigma_h]}(t-s)dN_i(s),
\end{equation}
then 
\begin{align*}
    [\boldsymbol\theta_{KL}]_1  =& \mathbb E[ (f_j-\mathbb E[f_j])g] +o(r^2+\sigma_h), \nonumber\\
    [\boldsymbol\theta_{KL}]_2  = &o(r^2+\sigma_h).
\end{align*}
\end{prop}

Although we can not directly observe $f_i$, Proposition \ref{prop3} suggests that using $f_i$ as a basis may reduce the error. Moreover, the form of $[\boldsymbol\theta_{KL}]_1 \approx \langle f_j,g\rangle$ also suggests using a "project $f_j$ on $f_i$" as basis to modify the MLE.

\subsection {Proposed Method}
Inspired by the analysis above, we now propose our modification for estimating impact. In particular, we minimize the following expression modified from likelihood function $\tilde{\ell}$:,
\begin{gather}
\min_{h_{i\to j}, \beta_j, \beta_w, \sigma_w } 
\left\{
-\sum_{s \in N_j } \log \tilde{\lambda}_{j}(s)
+ \int_0^T \tilde{\lambda}_{j}(s) \mathrm{d} s  \right\} 
\label{eq:target_likelihood} \\
\tilde{\lambda}_{j}(t) := 
\Big( \beta_j
+ \beta_w \; \overline{\textbf{s}_{i}}(t) 
+ \int_0^t h_{i\to j}(t-s) dN_{i}(s) 
\Big)_+
\label{eq:regression_intensity} \\
\overline{\textbf{s}_{i}}(t)
= \int_0^T W(t-s;\sigma_w) dN_i(s)
\label{eq:mean_coarsen_regressor}
\end{gather}
where $\overline{\textbf{s}_{i}}$ can be regarded as the coarsened point process smoothed by a Gaussian kernel 
$W(\tau;\sigma_w)=\frac{1}{\sqrt{2\pi\sigma_w^2}} \exp(-\frac{\tau^2}{2\sigma_w^2})$ with scale $\sigma_w$, serving as a substitute basis for $f_i$. We also specify an algorithm that can be implemented in continuous time, which does not require one to discretize the time points \citep{eden2008continuous, foufoula1986continuous}, so that the memory requirement is proportional to the number of time points instead of the number of time bins. The optimization algorithm is detailed in Appendix \ref{appendix:jitter_optimization_algo}. Empirical and theoretical analysis of the estimator will be discussed in section \ref{subsec:simulation_study}.
\vspace{-0.1in}




\subsection{Other use cases of the method} \label{subsec:use_cases}
\vspace{-0.15in}

Before experiments, we present some generality in the application of the method, with details left in Appendix.

\paragraph{Hypothesis testing (in Appendix \ref{appendix:hypothesis_testing}) } 
We compare our model with conditional inference via CCG and standard MHP, in hypothesis testing. Both our model and CCG have proper uniform p-value distribution under the null \citep[Theorem 10.14]{wasserman2004all}, where the standard MHP fails. Moreover, our method is also more powerful/sensitive at detecting weak signals with small sample sizes, see Figure \ref{fig:sim_demo}. Figure \ref{fig:sim_demo} shows a simulation example of fine timescale interaction between two point processes.
Synthetic data is generated by HP with one process inhibiting the other and common fluctuating background in Figure \ref{fig:sim_demo}A.
Figure \ref{fig:sim_demo}B is the result of the conditional inference via cross-correlogram (CCG).
The curve is mostly in the negative region indicating some inhibitory influence, yet the majority part of the curve stays within the acceptance band (i.e., not statistically significant).
Figure \ref{fig:sim_demo}CD are the result of the standard MHP vs our method, where impact function is represented as lag period. As shown, our method accurately detects the inhibitory relation and the estimated error is close to the true function (red curve), with the improvement compared to CCG in the statistical power and 
standard MHP in terms of error.
A similar observation in real data will be shown in Figure \ref{fig:neural_demo}.
A more detailed comparison between these models is in Appendix \ref{appendix:hypothesis_testing}.

\begin{figure}[ht]
\centering
\begin{subfigure}{0.48\textwidth}
\centering
\includegraphics[width=0.99\linewidth]{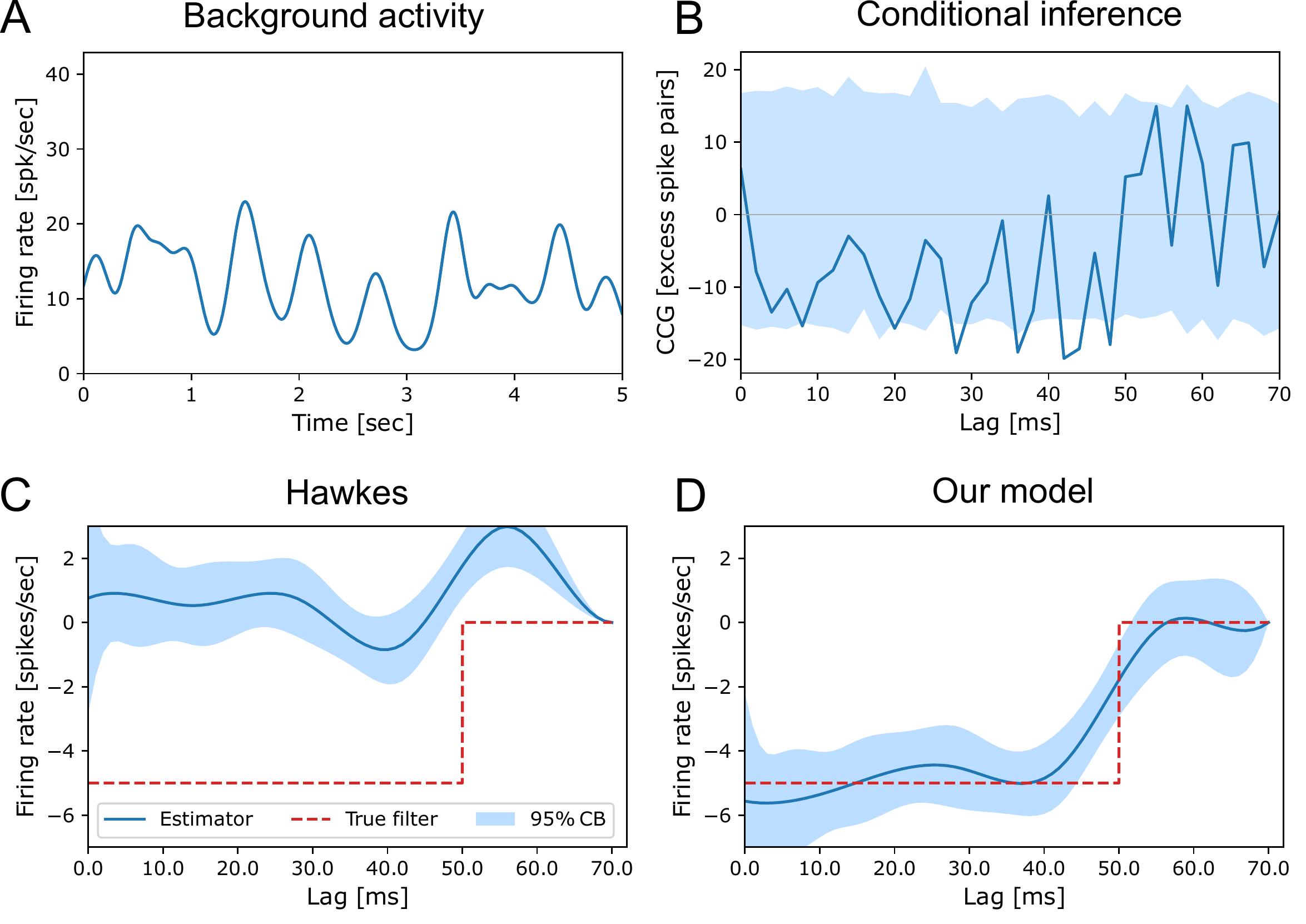}
\end{subfigure}
\caption{\small Impact function estimation with background fluctuation in simulation.
\textbf{A}: Shared background intensity.
\textbf{B}: CCG-based conditional inference. The 95\% acceptance band is constructed using Monte Carlo samples from the null distribution.
\textbf{C}: Standard MHP. The red curve is the ground truth.
\textbf{D}: Our model.
The band in C and D is also 95\% pointwise CI.
}
\vspace{-0.2in}
\label{fig:sim_demo}
\end{figure}

\paragraph{Non-parametric fitting for the impact function (in Appendix \ref{appendix:nonparametric})} 
Our method does not have constraints on modeling the impact function, which can be easily extended to non-parametric fitting.
One option is the general additive model using splines \citep[ch. 5]{pillow2008spatio, hastie2009elements}.
By leveraging the integral trick (Appendix \ref{appendix:jitter_optimization_algo}), time points do not need to be discretized and computational cost is small.

\paragraph{Bayesian inference (in Appendix \ref{appendix:bayesian}) } 
The method can be adopted for Bayesian inference where the uncertainty of the smoothing kernel width $\sigma_w$ is evaluated using a sampling-based inference algorithm.
The simulation shows that incorporating the uncertainty of  $\sigma_w$ does not affect the estimation of the temporal dependency significantly.
\vspace{-0.2in}



\section{Experiments }
In this section, we empirically verify the method through multiple simulation studies,
then apply the new tool to the neuroscience dataset where we discover the network of interacting neurons on a fine timescale.
For simulations, continuous-time point processes are generated using Lewis' thinning algorithm \citep{lewis1979simulation, ogata1981lewis}.
The gradient descent-based optimization algorithm is in Appendix section \ref{appendix:jitter_optimization_algo}.
Our code is available
\url{https://github.com/AlbertYuChen/point_process_coupling_public}.

\subsection{Simulation study } \label{subsec:simulation_study}

\subsubsection{Toy Example with background fluctuation } \label{subsec:background_artifacts}
In this synthetic dataset, the dynamic baselines have known form so that their correlation or the "inner product" between the source and target processes, as discussed in Section \ref{sec:representatie_model}, can be calculated in closed-form.
The background activities are
$f_i(t) = A \sin(2\pi (t - \phi_{\mathrm{rnd}})), 
f_j(t) = A \sin(2\pi (t -  \phi_{\mathrm{rnd}} - \phi_{\mathrm{lag}}))$,
where $A$ is the amplitude,
$T$ is the length of the trial.
We sample $\phi_{\mathrm{rnd}}\sim \mathrm{Unif(0, 1)}$ and set it to vary from trial to trial so the same background is never repeatedly observed. Here
$\phi_{\mathrm{lag}}$ controls the correlation between $f_i,f_j$, which we quantify using the \textit{normalized} dot product
$\langle f_i, f_j \rangle
:= \frac{1}{T A^2} \int_0^T f_i(s) f_j(s) \mathrm{d}s$.
When $\phi_{\mathrm{lag}}= 0$ and $0.5$, the dot product achieves the largest positive and negative value respectively;
when $\phi_{\mathrm{lag}}= 0.25$, the dot product is zero.

For the problem we are considering, short-term temporal dependency detection with dynamic background, there really is no "state-of-the-art" model as we are not mainly interested in predicting future observations, but we aim at gaining insight into the relationship between features and responses for scientific discovery, which is a more challenging task \citep{fan2020statistical}.
Although many recent point process models, such as \citep{mei2017neural, zhang2020self, zuo2020transformer}, are designed for prediction task, 
one popular representative deep learning-based model by \cite{mei2017neural} using recurrent neural networks is included as the baseline model.
The performance of three models are compared: standard MHP, our model, and Neural Hawkes \citep{mei2017neural}. 
Some other deep learning models are not considered due to the convoluted black-box structure.
For example in \citep{zhang2020self}, the intensity function is
\begin{equation*}
\begin{aligned}
\lambda_i(t)=&\text{softplus}(\mu_{u,i+1}+ \\ 
& (\eta_{u,i+1}-\mu_{u,i+1})\exp(-\gamma_{u,i+1}(t-t_i))),
\end{aligned}
\end{equation*}
where the variables $\mu,\eta,\gamma$ are all functions of latent variables obtained
through attention network. Another example is \citep{zuo2020transformer}, where the intensity function is
\begin{equation*}
    \lambda_k(t)=f_k(\alpha_k\frac{t-t_j}{t_j}+\boldsymbol w_k^T \boldsymbol h(t_j)+b_k),
\end{equation*}
where $t_j$ is the last event (not necessarily type k) and $h$ is the latent variable that carries more history information extracted from transformers.
Just by observing the intensity form above, one realizes that these models, designed for the event sequence prediction, are very difficult to draw inference on the coupling effect.
The method in \citep{mei2017neural} is the simplest framework we found where one can split out the coupling effect with minimum modification of the model.

The impact function is the square window impact function with a given width, so only the amplitude needs to be estimated.
Neural Hawkes takes intervals of the superimposed point processes one by one in sequence. The impact function from source to target is modeled as
\begin{gather}
\boldsymbol{c}(t) = 
\bar{\boldsymbol{c}}_{i+1} + 
(\boldsymbol{c}_{i+1} - \bar{\boldsymbol{c}}_{i+1}) 
\mathbb{I}_{[0,\sigma_h]}(t-t_i^{\mathrm{source}}), \\
\boldsymbol{h}(t) = \boldsymbol{o}_i \odot \mathrm{tanh}(\boldsymbol{c}(t)), \\
\lambda_{\mathrm{target}}
= \left(\boldsymbol{W}_{\mathrm{target}}^T \boldsymbol{h} \right)_+,
\end{gather}
which is slightly modified for the context (original kernel in \citep{mei2017neural} is exponential).
The impact function is extracted from the model (the original model does not directly offer an estimated parameter) as
$h_{\mathrm{source}\to \mathrm{target}}(t)
= \boldsymbol{W}_{\mathrm{target}}^T
\left[ \boldsymbol{o}_i \odot \mathrm{tanh}(
(\boldsymbol{c}_{i+1} - \bar{\boldsymbol{c}}_{i+1}) 
\mathbb{I}_{[0,\sigma_h]}(t) \right]
$ which could capture a time point's impact on the intensity.
Instead of modeling multiple points in the history at once as in the standard MHP, Neural Hawkes considers non-linear mapping, which only receives one last interval, while the history effect is carried over $\boldsymbol{c}_{i+1}$, $\bar{\boldsymbol{c}}_{i+1}$, and $\boldsymbol{o}_{i}$ through a recurrent neural network. The result is shown in Figure \ref{fig:bias_comparison} while
details are left in Appendix \ref{subsec:sinusoid}.

\begin{figure}[ht]
\centering
\includegraphics[width=0.7\linewidth]{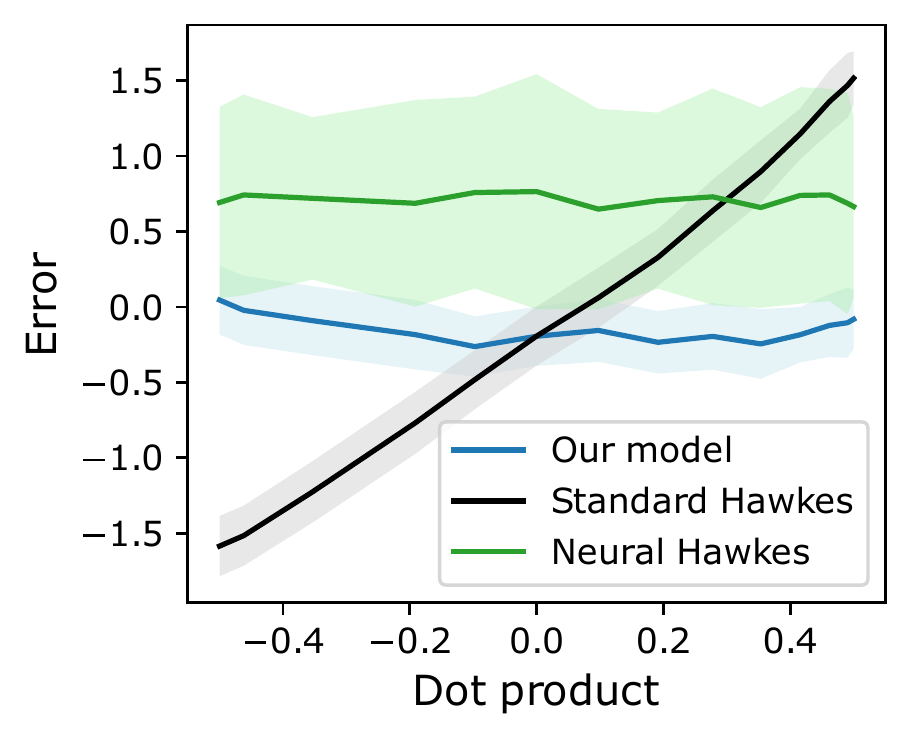}
\vspace{-0.15in}
\caption{\small A comparison of impact function estimation between standard MHP, Neural Hawkes, and our model under dynamic background.
The confidence band is created from 100 simulations.
}
\vspace{-0.12in}
\label{fig:bias_comparison}
\end{figure}

As shown, the bias of standard MHP is nearly linearly correlated with the dot product, as suggested by theoretical analysis. The error of Neural Hawkes is less susceptible to this correlation, which corroborates the ability of a recurrent structure to capture the interaction effect despite dynamic background. However, the error and variance of the impact estimation from Neural Hawkes are visibly non-negligible. This is likely due to the fact that neural network models typically need large datasets for training.
In contrast, our model performs satisfactorily in this example.

\subsubsection{Background kernel smoothing } \label{subsec:background_kernel_smoothing}
The kernel-smoothed basis in \eqref{eq:mean_coarsen_regressor} plays a key role in our method.
This section studies the relationship between the kernel width and the error of the estimator. In special cases, we are able to approximate the behavior of the estimator with an analytical formula.
Following the model framework in \eqref{eq:true_model}, assuming the background activity is generated similar to the \textit{linear Cox process} \citep{diggle1985kernel} or the \textit{cluster process} \citep[Definition 6.3.I.]{daley2003introduction}:
\begin{equation}
f_{i}=f_{j}:= 
\sum_i \phi_{\sigma_I}(t - t^c_i)
\label{eq:linear_cox_main}
\end{equation}
where $\phi_{\sigma_I}(\cdot)$ is some positive and even function, i.e., $\phi_{\sigma_I}(\cdot) >0$ and $ \phi_{\sigma_I}(\tau) = \phi_{\sigma_I}(-\tau)$. Here
$t^c_i$ are the centers of the windows generated by a Poisson process with intensity $\rho$.
$f_{i}$ is second-order stationary with a \textit{reduced covariance density} defined as follows (also see Appendix \ref{appendix:theoretical_derivations}).
\begin{equation}
{\small
\begin{aligned}
\breve{c}_{\Lambda}(u) :=& 
\mathbb{E}[f_{i}(x) f_{i}(x + u)] 
- \mathbb{E}[f_{i}(x)] \mathbb{E}[f_{i}(x + u)] \\
=& \rho [\phi_{\sigma_I} \ast \phi_{\sigma_I}](u) \\
\breve{c}_{N}(u) :=& 
\mathbb{E}\left[ 
\frac{d N_{i}( x) d N_{i}(x + u) }{(\mathrm{d}t)^2} \right]
 - \mathbb{E}\left[\frac{d N_{i}(x)}{\mathrm{d}t}\right] 
 \mathbb{E}\left[\frac{d N_{i}(x+u)}{\mathrm{d}t}\right]  \\
=& \rho \cdot [\phi_{\sigma_I} \ast \phi_{\sigma_I}](u)
+ (\rho+\alpha_i) \delta(u)
\end{aligned}
}
\end{equation}
which describes the smoothness of background activity, and
$\alpha_i$ is the constant in \eqref{eq:true_model}.
If adjacent points with lag $u$ have larger covariance $\breve{c}_{\Lambda}(u)$, the background would be smoother.
The impact functions are
$h_{i\to j}(t) = \alpha_{i\to j} h(t)$, with amplitude to be fitted, for example
$h(t) = \mathbb{I}_{[0,\sigma_h]}(t)$.
Then the error in model \eqref{eq:target_likelihood} may be approximated as,
\begin{equation}
\mathrm{error}(\hat\alpha_{i\to j}) \approx 
\frac{ 
\langle W, W \rangle_{\breve{c}_{N}} 
    \langle h, \mathbf{1} \rangle_{\breve{c}_{\Lambda}} 
- \langle h, W \rangle_{\breve{c}_{N}} 
    \langle W, \mathbf{1} \rangle_{\breve{c}_{\Lambda}}
}{
\langle W, W \rangle_{\breve{c}_{N}} 
\langle h, h^- \rangle_{\breve{c}_{N}}
- \langle W, \mathbf{1} \rangle_{\breve{c}_{\Lambda}}^2
}
\label{eq:bias_formula}
\end{equation}
$\mathbf{1}$ is the constant and
$h^-(\tau)=h(-\tau)$.
The special inner product here are defined as
$\langle g_1, g_2 \rangle_{\breve{c}} 
:= \int [g_1 \ast g_2](s) \breve{c}(s) \mathrm{d}s$ with $\ast$ denoting the convolution.
The derivation of the analytical formula is in Appendix \ref{appendix:theoretical_derivations}. Simulation and analytical results are presented in Figure \ref{fig:kernel_smoothing_eg}.

\begin{figure}[ht]
\centering
\includegraphics[width=0.9\linewidth]{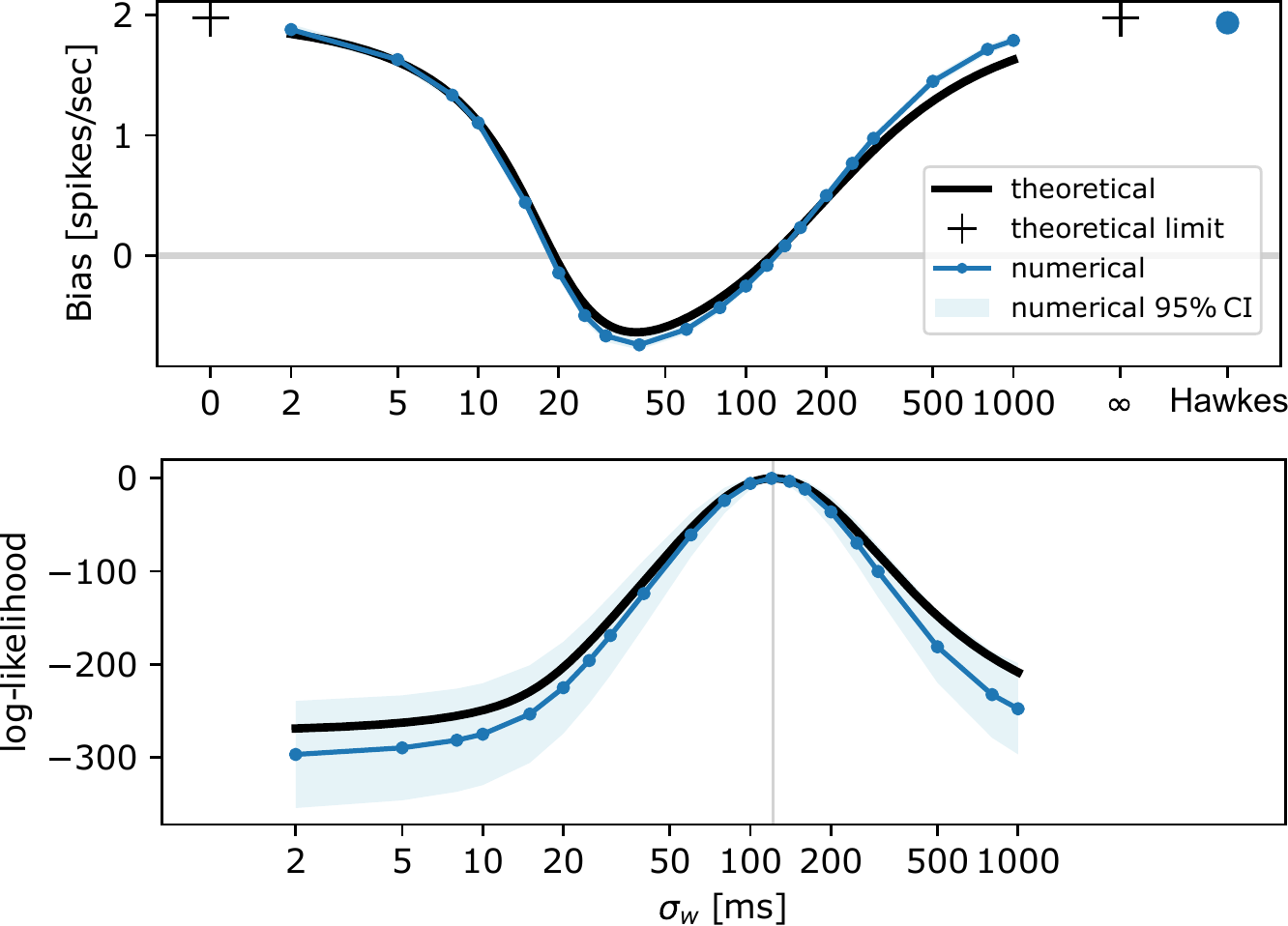}
\vspace{-0.1in}
\caption{\small Error and likelihood of the estimator as functions of background smoothing kernel width $\sigma_w$ in \eqref{eq:mean_coarsen_regressor}. Numerical and theoretical results as in \eqref{eq:bias_formula} are shown in blue and dark respectively. The error of standard MHP is the blue dot on the right.
}
\label{fig:kernel_smoothing_eg}
\end{figure}

The error and log-likelihood are plotted as functions of the smoothing kernel width $\sigma_w$ in \eqref{eq:mean_coarsen_regressor}.
The MLE, indicated by the vertical line in Figure \ref{fig:kernel_smoothing_eg}, achieves a small error that agrees with the example in section \ref{subsec:background_artifacts}.
Interestingly, when the kernel width is too small or too large, including the theoretical limits by taking $\sigma_w \to 0$ or $\sigma_w \to\infty$, the model fails under heterogeneity. In this case, the error is close to that of standard MHP.
Details are in Appendix \ref{appendix:sim_linear_cox}.

\begin{figure}[ht]
\centering
\includegraphics[width=0.9\linewidth]{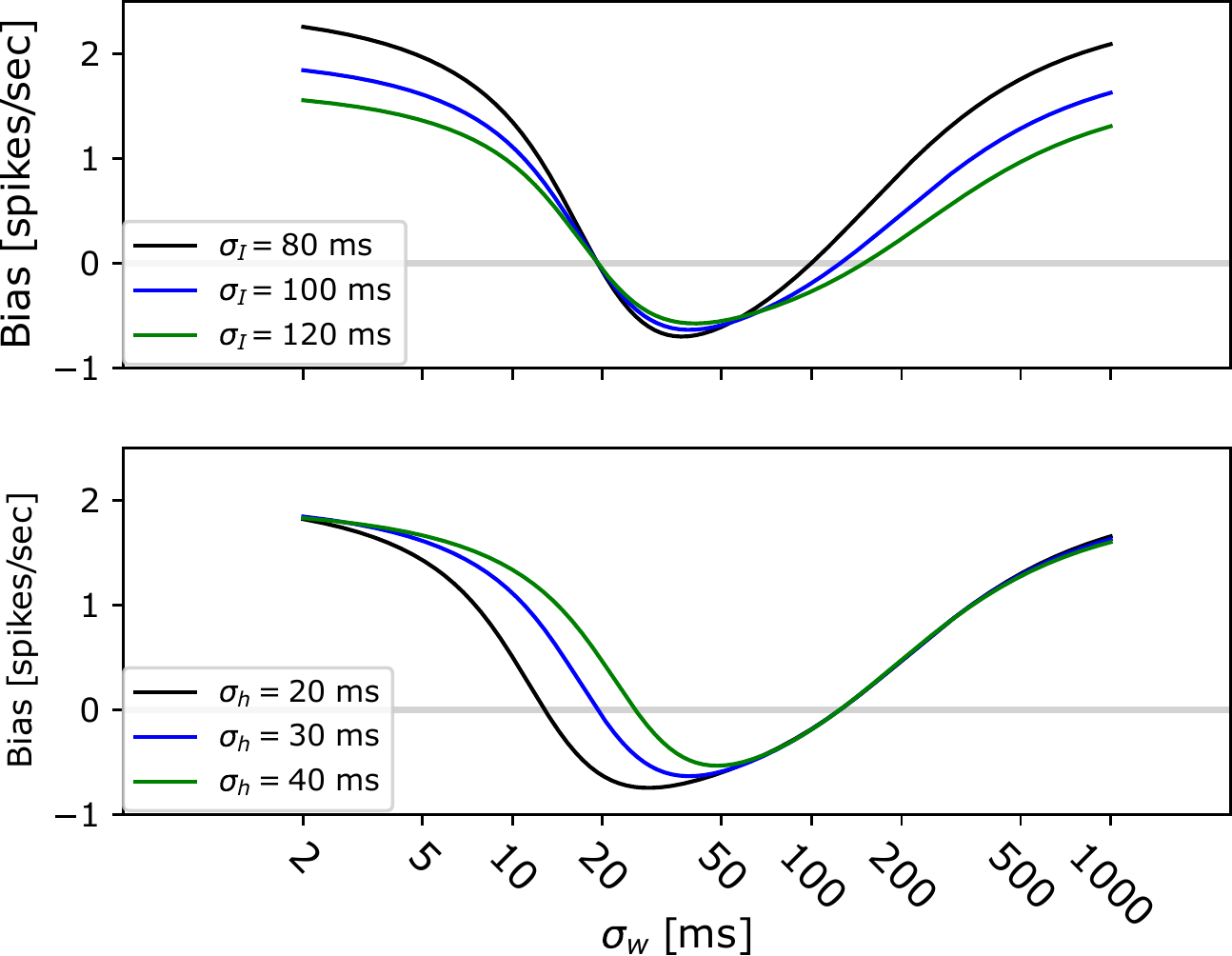}
\vspace{-0.1in}
\caption{Error function with different background timescales $\sigma_I$ (top) or coupling effect timescales $\sigma_h$ (bottom).
Numerical results match the theoretical results well, so only theoretical results are presented according to \eqref{eq:bias_formula}.
}
\vspace{-0.12in}
\label{fig:bias_tuning}
\end{figure}

In Figure \ref{fig:kernel_smoothing_eg},
when $\sigma_w$ is between 20 ms and 120 ms, the error can be negative.
The error as a function of the background smoothing kernel has two roots. The roots are related to the timescale of the coupling effect $\sigma_h$ and the timescale of the background $\sigma_I$ as in \eqref{eq:linear_cox_main}.
In Figure \ref{fig:bias_tuning},
if $\sigma_I$ increases, the root on the right, corresponding to the MLE, will move toward the right, as the background smoothing kernel $W$ captures the fluctuation of the background.
If $\sigma_h$ increases, the root on the left will move toward the right.
This can be intuitively interpreted by \eqref{eq:bias_formula}.
Let $W_h$ be the kernel with $\sigma_w \approx \sigma_h$, then
$\langle W_h, W_h \rangle_{\breve{c}_{N}}
    \langle h, \mathbf{1} \rangle_{\breve{c}_{\Lambda}}
\approx \langle h, W_h \rangle_{\breve{c}_{N}}
    \langle W_h, \mathbf{1} \rangle_{\breve{c}_{\Lambda}}$.
So $\sigma_w = \sigma_h$ is close to the root of \eqref{eq:bias_formula}.
Changing the amplitude of the impact function $\alpha_{i\to j}$ in a certain range does not influence the bias curve.
More details are in Appendix \ref{appendix:timescales}.

\subsubsection{Two-way cross connections and self-connections }
The model in Figure \ref{fig:couple_process_diagram} only shows one cross-connection $i\to j$.
This simulation scenario includes the most general two-way MHP cross/self connections between processes ($i\to j$ and $j\to i$), and self-connections ($i\to i$ and $j\to j$). The comparison between our model and standard MHP is in Table \ref{tab:full_model_results}. Details of the experiment are in Appendix \ref{appendix:full_connection}.
Our model considerably outperforms the standard MHP in estimating cross-impact connections.
However, both models perform poorly on self-connection estimation, as they are considered nuisance parameters in our method.

\begin{table}[H]
\centering
{\scriptsize
\begin{tabular}{cccclcc}
\cline{1-3} \cline{5-7}
 & \multicolumn{2}{c}{Our model} &  &  & \multicolumn{2}{c}{Standard Hawkes} \\ \cline{1-3} \cline{5-7}
 & $i$ & $j$ &  &  & $i$ & $j$ \\ \cline{1-3} \cline{5-7} 
$i$ & $1.70(0.18)$ & $\textbf{0.21}(0.14)$  & 
    & $i$ & $2.39(0.18)$ & $2.39(0.19)$  \\ \cline{1-3} \cline{5-7} 
$j$ & $\textbf{0.22}(0.15)$ & $1.66(0.18)$ &  
    & $j$ & $2.40(0.19)$ & $2.39(0.18)$ \\ \cline{1-3} \cline{5-7} 
\end{tabular}
}
\vspace{-0.1in}
\caption{Comparison between our model and standard MHP model in full connection task.
Rows are source nodes, columns are target nodes.
Each cell shows the mean absolute error with standard deviation.
Unit in spikes/sec.
}
\vspace{-0.2in}
\label{tab:full_model_results}
\end{table}


\subsubsection{Multivariate Hawkes model}
It is natural to extend our bivariate regression-type method to a multivariate regression-type model.
The coupling effect in multivariate processes can be regarded as a form of graph structure recovery in graphical models, where each point process is considered a node.
From this perspective, e.g., \citep[sec. 19.4.4]{meinshausen2006high, murphy2012machine}, multivariate regression extends the bivariate case by studying \textit{pairwise}
conditional relations for all possible pairs.
More specifically, given a pair of random variables $X,Y$, let $Z$ represent the totality of all other random variables excluding $X,Y$. The multivariate regression infers if a \textit{bivariate} relation $X\perp Y|Z$ holds, also known as the \textit{global Markov property} \citep{koller2009probabilistic}.
A similar concept in standard MHP can be found in \citep{eichler2017graphical}.
In our MHP setting, this is equivalent to estimating the impact functions between $N_i$ and $N_j$ given the observations of all other processes and so that their effect enters as the dynamic background. Notice that the standard MHP cannot model this extension because even if the baseline intensity of each point process is constant, the totality of random effect from all other nodes excluding two nodes will not necessarily give a constant baseline to the nodes under consideration.
Consider the intensity function in the multivariate point process,
\begin{equation}
\begin{aligned}
\lambda_{j}(t)
&= \alpha_j
+ \int_0^t h_{i\to j}(t-s) dN_{i}(s) \\
&\underbrace{+ f_{j}(t) 
+ \sum_{r\neq i, j} \int_0^t h_{r\to j}(t-s) dN_{r}(s) 
}_{\tilde{f}_j(t) }
\end{aligned}
\label{eq:multivariate}
\end{equation}
where $f_j$ together with input from other processes are treated as a new background $\tilde{f}_j(t)$. This perspective exactly reduces the MHP to model \eqref{eq:regression_intensity}.

\begin{table}[H]
\centering
{\footnotesize
\begin{tabular}{c|cc}
\hline
 & Bias (std) [spikes/sec] & RMSE (std) [spikes/sec] \\ \hline
Hawkes & 1.52 (0.040)  & 1.54 (0.41) \\
Ours & \textbf{0.028} (0.040)  & \textbf{0.25} (0.33) \\ \hline
\end{tabular}
}
\caption{Comparison between the performance between the standard Hawkes model and our model. }
\label{tab:multivariate}
\end{table}
The performance of the model is evaluated using simulation dataset, which involves 6 processes and all processes are driven by fluctuating background. The coupling effects between nodes can be positive, negative or zero.
Table \ref{tab:multivariate} shows that our method outperforms the standard Hawkes model in multivariate processes scenario.
Details of the experiment is in Appendix \ref{appendix:multivariate_regression_sim}.


\subsubsection{Other simulation scenarios} \label{subsec:other_simulations}
Other properties of the model and empirical verifications are briefly summarized in this section due to the page limit.


\textbf{Varying-timescale background.} See Appendix \ref{appendix:varying_sigma_I}.
We violate the settings in section \ref{subsec:background_kernel_smoothing} by relaxing the fixed background timescale $\sigma_I$ in \eqref{eq:linear_cox_main} to randomly changing timescale to test the robustness of the model. 

\textbf{Fast-changing background.} See Appendix \ref{appendix:small_sigma_I}.
In extreme cases, the background activity $f_{i}$ can have fast-changing activities.
In this situation, the conditional inference-based method will be limited by its formalization of the null hypothesis, which implicitly assumes the timescale of the coupling effect is smaller than the timescale of the background. We show our model is still able to accurately estimate the cross-impact effect while the conditional inference-based method fails. 

\textbf{Asymptotic Normality.} See Appendix \ref{appendix:normality}.
Similar to profile likelihood, the approximate normality of the estimator is observed in simulations. The property may be convenient for model inference, details are also in Appendix \ref{appendix:hypothesis_testing}. 

\textbf{Selection of impact function length.} See Appendix \ref{appendix:unmatch_filter_length}. In practice, the timescale of the interaction effect is typically unknown. When users are not confident with the prior knowledge of the timescale of the coupling effect, our methods can be adapted to use a shorter impact function or non-parametric fitting first, as shown in Appendix \ref{appendix:nonparametric}.




\subsection{Neuropixels data }

Spiking neural activities usually come with non-stationary background signals due to external stimuli or inter-area interactions.
With the recent advance in high-density electrophysiological recording technology, such as Neuropixels, hundreds of neurons from multiple brain regions can be recorded simultaneously. This offers opportunities to further investigate the interactions between brain areas \citep{siegle2021survey, chen2022population}.
However, the point-to-point coupling effect on fine timescale across regions is not well studied.

Here, we apply our method to the hierarchical mouse visual system across 5 brain areas: V1, LM, RL, AL, and AM in ascending order with V1 as the primary visual cortex processing simple visual features, and AM as the high-order cortex handling sophisticated signals \citep{harris2019hierarchical, siegle2021survey} (Figure \ref{fig:neuropixels}).
We aim to fit the coupling effect across brain regions and discover the excitatory or inhibitory interactions on a fine timescale.
Details are in Appendix \ref{appendix:jitter_app}.

Figure \ref{fig:neural_demo}A demonstrates the averaged activities of 3 brain regions with large correlation as a clue for potential background fluctuation. Results are very similar to the simulation in Figure \ref{fig:sim_demo},  CCG in Figure \ref{fig:neural_demo}B shows some negative but not statistically significant effects.
Our method is more sensitive in detecting the effect between 0 and 50 ms lag.
In contrast, due to the background artifacts, the standard Hawkes model detects non-significant or slightly positive coupling effect.

\begin{figure}[H]
\centering
\begin{subfigure}{0.48\textwidth}
\centering
\includegraphics[width=0.99\linewidth]{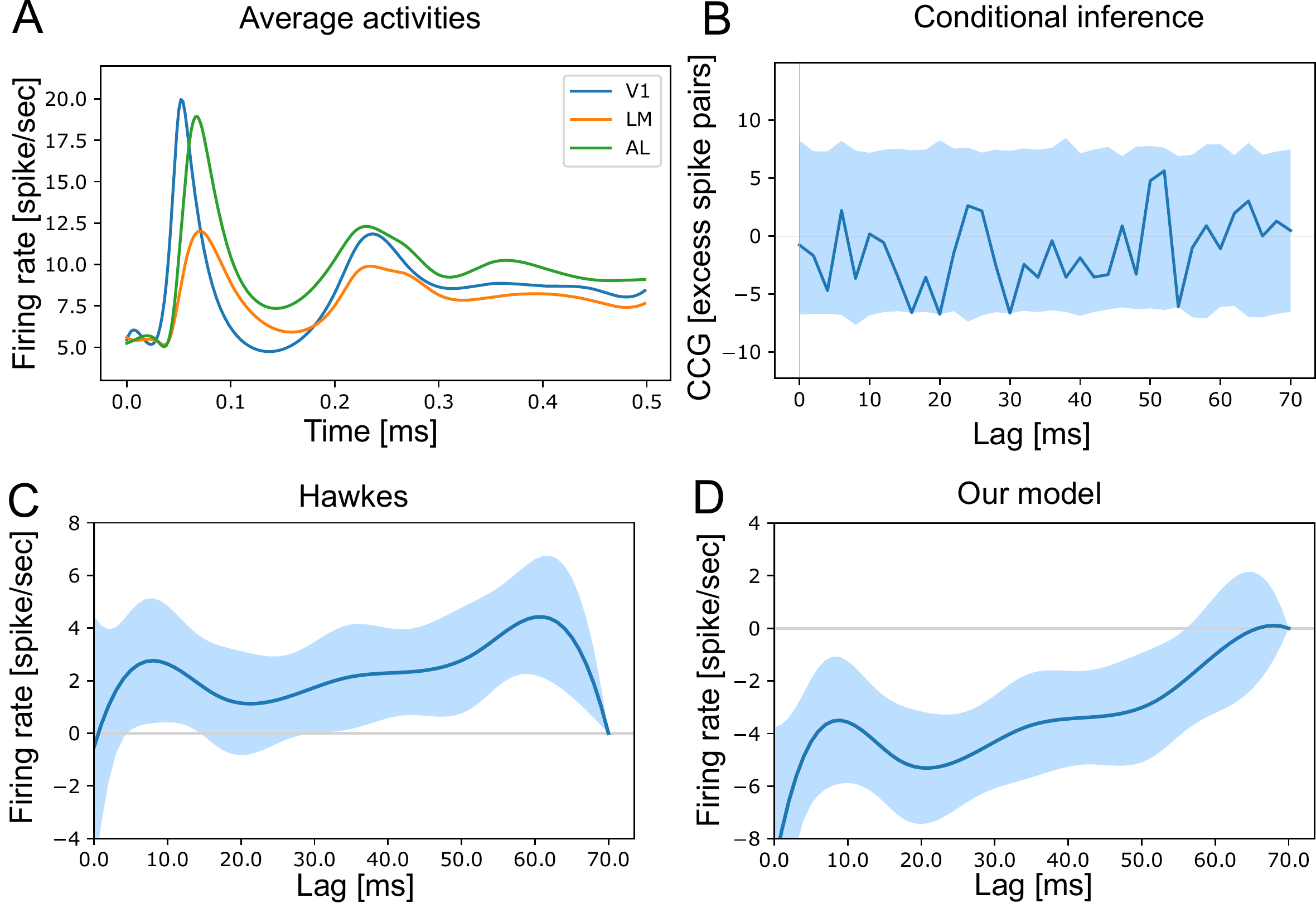}
\end{subfigure}
\vspace{-0.12in}
\caption{Neuropixels data. 
\textbf{A}: Activities of three brain areas showing correlated backgrounds.
B,C,D are results of a pair of neurons.
\textbf{B}: CCG.
\textbf{C}: standard MHP.
\textbf{D}: Our method.
}
\label{fig:neural_demo}
\vspace{-0.2in}
\end{figure}

\begin{figure}[ht]
\centering
\hspace*{-0.1in}
\includegraphics[width=1.05\linewidth]{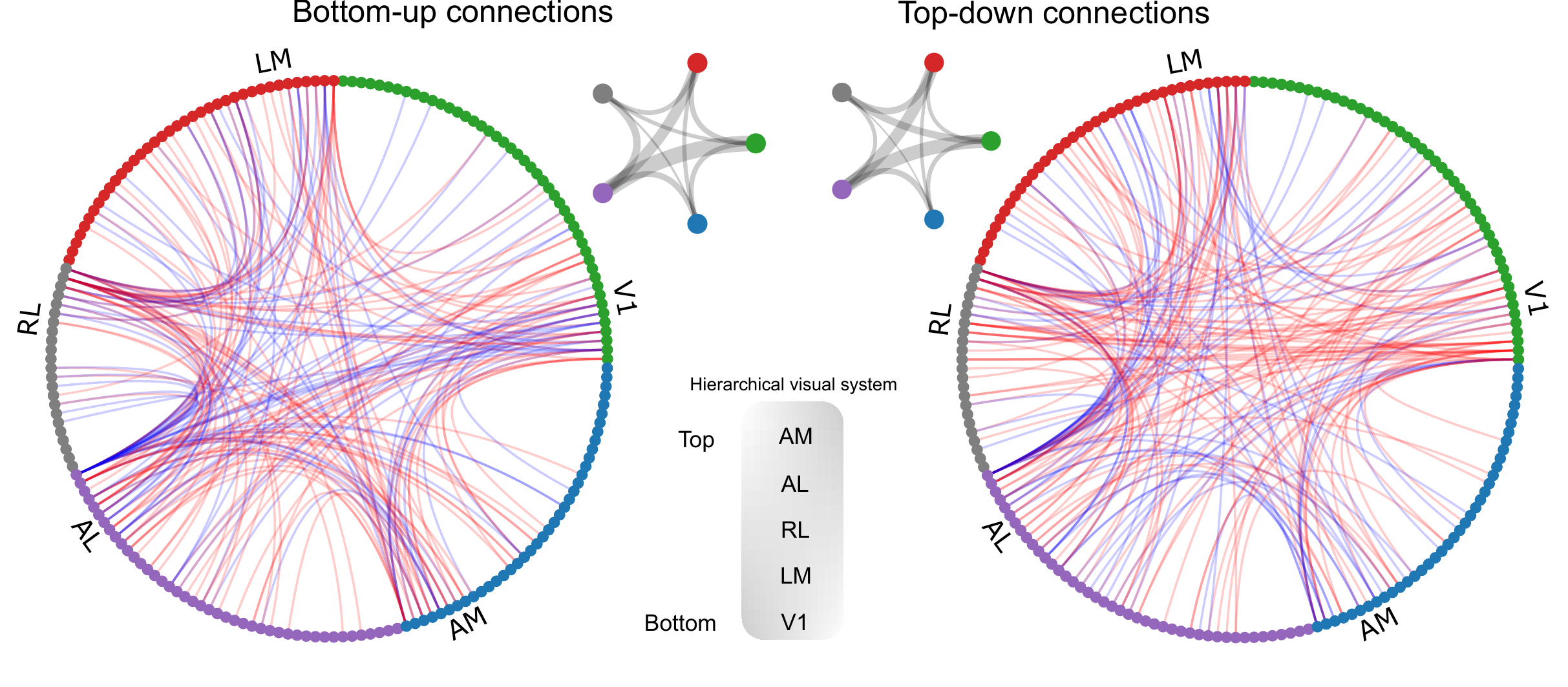}
\vspace{-0.2in}
\caption{Network of coupling neurons in mouse visual system.
The analysis involves 5 brain areas, V1, LM, RL, AL, and AM in hierarchically increasing order, where V1 has the lowest order \citep{siegle2021survey}.
Excitatory (positive) and inhibitory (negative) impact functions are shown in red and blue edges.
For better visualization, it only shows 20\% randomly selected edges.
The coupling filter, a directed edge, connecting a lower-order region to a higher-order region, for example from V1 to RL, is categorized into the bottom-up graph on the left; the graph on the right shows the top-down connections \citep{siegle2021survey,harris2019hierarchical}.
The small graphs at the corner count the total number of edges between areas.
}
\vspace{-0.15in}
\label{fig:neuropixels}
\end{figure}

Figure \ref{fig:neuropixels} shows the discovered neuronal network of 190 neurons.
Multiple significant impact functions are selected with Bonferroni correction at level 0.01. 
766 directed edges are split into bottom-up connections and top-down connections \citep{siegle2021survey,harris2019hierarchical}.
The impact function is fitted using 50 ms square window determined by exploring CCG and non-parametric fitting (see examples in Appendix  \ref{appendix:jitter_app}).
Our main findings using MHP extension are: (a) Most edges concentrate at a few neurons. (b)
The active senders or receivers are consistent across top-down and bottom-up networks. The real data has no ground truth so we cannot directly evaluate the performance of this multivariate extension. However, we do note that the findings directly corroborate previous neuroscience studies \citep{harris2019hierarchical,glickfeld2017higher} based on anatomical analysis, whereas our findings are entirely data-driven. The findings are also complementary to \citep{jia2020multi,siegle2021survey} using traditional CCG method (in section 2, our method outperforms CCG in both computation and performance).
\vspace{-0.1in}


\begin{figure}[H]
\centering
\includegraphics[width=0.8\linewidth]{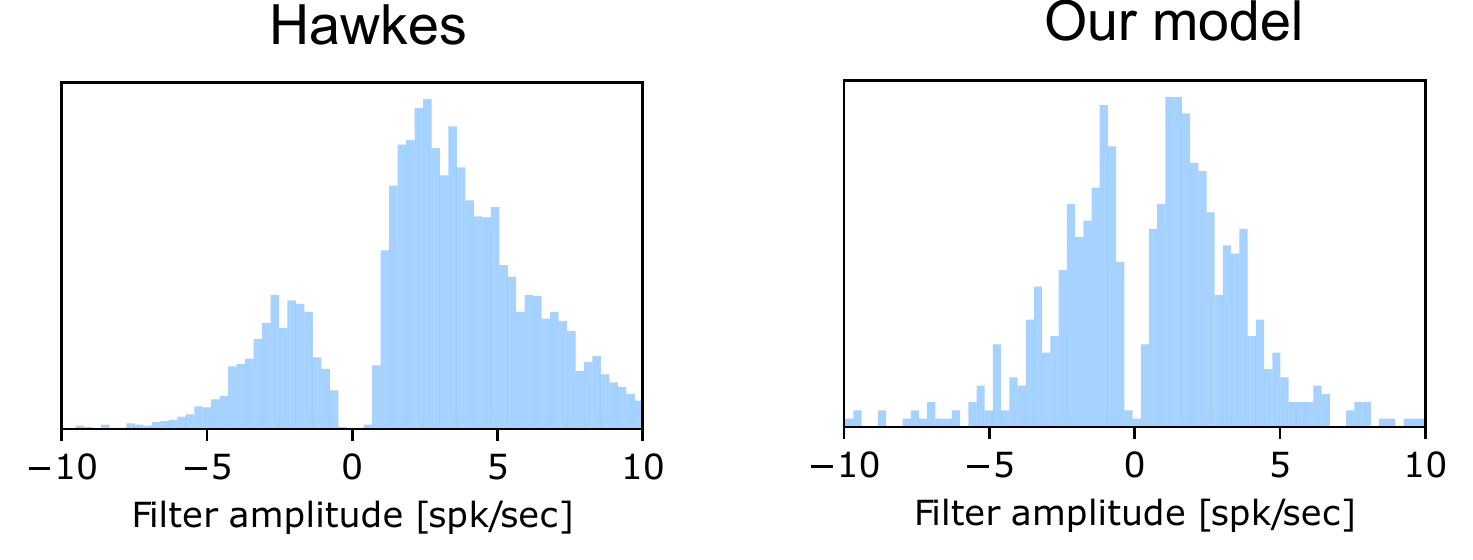}
\vspace{-0.1in}
\caption{A comparison of impact function histograms between the standard Hawkes model and our model.
}
\vspace{-0.12in}
\label{fig:h_histogram}
\end{figure}
Figure \ref{fig:h_histogram} compares the distributions of the impact functions between the standard Hawkes model and our model.
Because of the background artifacts shown in Figure \ref{fig:neural_demo}, it is suspected that the standard Hawkes model may falsely detect more positive relations.
Goodness-of-fit and more details of the experiment can be found in Appendix \ref{appendix:jitter_app}. As shown, all the above discovery is greatly facilitated by our method.

\section{Conclusion }
\vspace{-0.1in}
We report and analyze the error of MLE from MHP in short-term temporal dependency detection due to heterogeneous background, which is common but largely overlooked. We developed a flexible, robust, and computationally efficient model to address this problem in an attempt to generalize the use case for MHP in practice.
Finally, we applied the new tool to the neuroscience dataset and discovered a neuronal network across visual cortices in the mouse visual system.




\begin{acknowledgements} 
    Yu Chen was supported by NIMH grant RO1 MH064537 when the initial part of the work was done at Carnegie Mellon University.
\end{acknowledgements}

\bibliography{Bibliography}

\begin{thebibliography}{71}
\providecommand{\natexlab}[1]{#1}
\providecommand{\url}[1]{\texttt{#1}}
\expandafter\ifx\csname urlstyle\endcsname\relax
  \providecommand{\doi}[1]{doi: #1}\else
  \providecommand{\doi}{doi: \begingroup \urlstyle{rm}\Url}\fi

\bibitem[Amarasingham et~al.(2012)Amarasingham, Harrison, Hatsopoulos, and
  Geman]{amarasingham2012conditional}
Asohan Amarasingham, Matthew~T Harrison, Nicholas~G Hatsopoulos, and Stuart
  Geman.
\newblock Conditional modeling and the jitter method of spike resampling.
\newblock \emph{Journal of Neurophysiology}, 107\penalty0 (2):\penalty0
  517--531, 2012.

\bibitem[Bacry et~al.(2015)Bacry, Mastromatteo, and Muzy]{bacry2015hawkes}
Emmanuel Bacry, Iacopo Mastromatteo, and Jean-Fran{\c{c}}ois Muzy.
\newblock Hawkes processes in finance.
\newblock \emph{Market Microstructure and Liquidity}, 1\penalty0 (01):\penalty0
  1550005, 2015.

\bibitem[Bartlett(1963)]{bartlett1963statistical}
Maurice~S Bartlett.
\newblock Statistical estimation of density functions.
\newblock \emph{Sankhy{\=a}: The Indian Journal of Statistics, Series A}, pages
  245--254, 1963.

\bibitem[Bartlett(1964)]{bartlett1964spectral}
Maurice~S Bartlett.
\newblock The spectral analysis of two-dimensional point processes.
\newblock \emph{Biometrika}, 51\penalty0 (3/4):\penalty0 299--311, 1964.

\bibitem[Bhatti and Bracken(2006)]{bhatti2006calculation}
M~Bhatti and P~Bracken.
\newblock The calculation of integrals involving b-splines by means of
  recursion relations.
\newblock \emph{Applied mathematics and computation}, 172\penalty0
  (1):\penalty0 91--100, 2006.

\bibitem[Blundell et~al.(2012)Blundell, Beck, and
  Heller]{blundell2012modelling}
Charles Blundell, Jeff Beck, and Katherine~A Heller.
\newblock Modelling reciprocating relationships with hawkes processes.
\newblock \emph{Advances in neural information processing systems}, 25, 2012.

\bibitem[Bowsher(2007)]{bowsher2007modelling}
Clive~G Bowsher.
\newblock Modelling security market events in continuous time: Intensity based,
  multivariate point process models.
\newblock \emph{Journal of Econometrics}, 141\penalty0 (2):\penalty0 876--912,
  2007.

\bibitem[Br{\'e}maud et~al.(2005)Br{\'e}maud, Massouli{\'e}, and
  Ridolfi]{bremaud2005power}
Pierre Br{\'e}maud, Laurent Massouli{\'e}, and Andrea Ridolfi.
\newblock Power spectra of random spike fields and related processes.
\newblock \emph{Advances in applied probability}, 37\penalty0 (4):\penalty0
  1116--1146, 2005.

\bibitem[Brillinger(1972)]{brillinger1972spectral}
David~R Brillinger.
\newblock The spectral analysis of stationary interval functions.
\newblock In \emph{Vol. 1 Theory of Statistics}, pages 483--514. University of
  California Press, 1972.

\bibitem[Brillinger(1974)]{brillinger1974cross}
David~R Brillinger.
\newblock Cross-spectral analysis of processes with stationary increments
  including the stationary $\mathrm{G}/\mathrm{G}/\infty$ queue.
\newblock \emph{The Annals of Probability}, pages 815--827, 1974.

\bibitem[Brown et~al.(2002)Brown, Barbieri, Ventura, Kass, and
  Frank]{brown2002time}
Emery~N Brown, Riccardo Barbieri, Val{\'e}rie Ventura, Robert~E Kass, and
  Loren~M Frank.
\newblock The time-rescaling theorem and its application to neural spike train
  data analysis.
\newblock \emph{Neural computation}, 14\penalty0 (2):\penalty0 325--346, 2002.

\bibitem[Chen et~al.(2017)Chen, Witten, and Shojaie]{chen2017nearly}
Shizhe Chen, Daniela Witten, and Ali Shojaie.
\newblock Nearly assumptionless screening for the mutually-exciting
  multivariate hawkes process.
\newblock \emph{Electronic journal of statistics}, 11\penalty0 (1):\penalty0
  1207, 2017.

\bibitem[Chen et~al.(2022)Chen, Douglas, Medina, Olarinre, Siegle, and
  Kass]{chen2022population}
Yu~Chen, Hannah Douglas, Bryan~J Medina, Motolani Olarinre, Joshua~H Siegle,
  and Robert~E Kass.
\newblock Population burst propagation across interacting areas of the brain.
\newblock \emph{Journal of Neurophysiology}, 128\penalty0 (6):\penalty0
  1578--1592, 2022.

\bibitem[Chwialkowski and Gretton(2014)]{chwialkowski2014kernel}
Kacper Chwialkowski and Arthur Gretton.
\newblock A kernel independence test for random processes.
\newblock In \emph{International Conference on Machine Learning}, pages
  1422--1430. PMLR, 2014.

\bibitem[Cowling et~al.(1996)Cowling, Hall, and Phillips]{cowling1996bootstrap}
Ann Cowling, Peter Hall, and Michael~J Phillips.
\newblock Bootstrap confidence regions for the intensity of a poisson point
  process.
\newblock \emph{Journal of the American Statistical Association}, 91\penalty0
  (436):\penalty0 1516--1524, 1996.

\bibitem[Cox and Isham(1980)]{cox1980point}
David~Roxbee Cox and Valerie Isham.
\newblock \emph{Point processes}, volume~12.
\newblock CRC Press, 1980.

\bibitem[Da~Fonseca and Zaatour(2014)]{da2014hawkes}
Jos{\'e} Da~Fonseca and Riadh Zaatour.
\newblock Hawkes process: Fast calibration, application to trade clustering,
  and diffusive limit.
\newblock \emph{Journal of Futures Markets}, 34\penalty0 (6):\penalty0
  548--579, 2014.

\bibitem[Daley and Vere-Jones(2003)]{daley2003introduction}
Daryl~J Daley and David Vere-Jones.
\newblock \emph{An introduction to the theory of point processes: volume I:
  elementary theory and methods}.
\newblock Springer, 2003.

\bibitem[Daley and Vere-Jones(2008)]{daley2008introduction}
Daryl~J Daley and David Vere-Jones.
\newblock \emph{An introduction to the theory of point processes: volume II:
  general theory and structure}.
\newblock Springer, 2008.

\bibitem[Diggle(1985)]{diggle1985kernel}
Peter Diggle.
\newblock A kernel method for smoothing point process data.
\newblock \emph{Journal of the Royal Statistical Society: Series C (Applied
  Statistics)}, 34\penalty0 (2):\penalty0 138--147, 1985.

\bibitem[Eden and Brown(2008)]{eden2008continuous}
Uri~T Eden and Emery~N Brown.
\newblock Continuous-time filters for state estimation from point process
  models of neural data.
\newblock \emph{Statistica Sinica}, 18\penalty0 (4):\penalty0 1293, 2008.

\bibitem[Eichler et~al.(2017)Eichler, Dahlhaus, and
  Dueck]{eichler2017graphical}
Michael Eichler, Rainer Dahlhaus, and Johannes Dueck.
\newblock Graphical modeling for multivariate hawkes processes with
  nonparametric link functions.
\newblock \emph{Journal of Time Series Analysis}, 38\penalty0 (2):\penalty0
  225--242, 2017.

\bibitem[Errais et~al.(2010)Errais, Giesecke, and Goldberg]{errais2010affine}
Eymen Errais, Kay Giesecke, and Lisa~R Goldberg.
\newblock Affine point processes and portfolio credit risk.
\newblock \emph{SIAM Journal on Financial Mathematics}, 1\penalty0
  (1):\penalty0 642--665, 2010.

\bibitem[Fan et~al.(2020)Fan, Li, Zhang, and Zou]{fan2020statistical}
Jianqing Fan, Runze Li, Cun-Hui Zhang, and Hui Zou.
\newblock \emph{Statistical foundations of data science}.
\newblock CRC press, 2020.

\bibitem[Farajtabar et~al.(2015)Farajtabar, Wang, Gomez~Rodriguez, Li, Zha, and
  Song]{farajtabar2015coevolve}
Mehrdad Farajtabar, Yichen Wang, Manuel Gomez~Rodriguez, Shuang Li, Hongyuan
  Zha, and Le~Song.
\newblock Coevolve: A joint point process model for information diffusion and
  network co-evolution.
\newblock \emph{Advances in Neural Information Processing Systems}, 28, 2015.

\bibitem[Foufoula-Georgiou and Lettenmaier(1986)]{foufoula1986continuous}
Efi Foufoula-Georgiou and Dennis~P Lettenmaier.
\newblock Continuous-time versus discrete-time point process models for
  rainfall occurrence series.
\newblock \emph{Water Resources Research}, 22\penalty0 (4):\penalty0 531--542,
  1986.

\bibitem[Glickfeld and Olsen(2017)]{glickfeld2017higher}
Lindsey~L Glickfeld and Shawn~R Olsen.
\newblock Higher-order areas of the mouse visual cortex.
\newblock \emph{Annual review of vision science}, 3:\penalty0 251--273, 2017.

\bibitem[Gunawardana et~al.(2011)Gunawardana, Meek, and
  Xu]{gunawardana2011model}
Asela Gunawardana, Christopher Meek, and Puyang Xu.
\newblock A model for temporal dependencies in event streams.
\newblock \emph{Advances in neural information processing systems}, 24, 2011.

\bibitem[Hansen et~al.(2015)Hansen, Reynaud-Bouret, and
  Rivoirard]{10.3150/13-BEJ562}
Niels~Richard Hansen, Patricia Reynaud-Bouret, and Vincent Rivoirard.
\newblock {Lasso and probabilistic inequalities for multivariate point
  processes}.
\newblock \emph{Bernoulli}, 21\penalty0 (1):\penalty0 83 -- 143, 2015.
\newblock \doi{10.3150/13-BEJ562}.
\newblock URL \url{https://doi.org/10.3150/13-BEJ562}.

\bibitem[Harris et~al.(2019)Harris, Mihalas, Hirokawa, Whitesell, Choi,
  Bernard, Bohn, Caldejon, Casal, Cho, et~al.]{harris2019hierarchical}
Julie~A Harris, Stefan Mihalas, Karla~E Hirokawa, Jennifer~D Whitesell, Hannah
  Choi, Amy Bernard, Phillip Bohn, Shiella Caldejon, Linzy Casal, Andrew Cho,
  et~al.
\newblock Hierarchical organization of cortical and thalamic connectivity.
\newblock \emph{Nature}, 575\penalty0 (7781):\penalty0 195--202, 2019.

\bibitem[Haslinger et~al.(2010)Haslinger, Pipa, and
  Brown]{haslinger2010discrete}
Robert Haslinger, Gordon Pipa, and Emery Brown.
\newblock Discrete time rescaling theorem: determining goodness of fit for
  discrete time statistical models of neural spiking.
\newblock \emph{Neural computation}, 22\penalty0 (10):\penalty0 2477--2506,
  2010.

\bibitem[Hastie et~al.(2009)Hastie, Tibshirani, and
  Friedman]{hastie2009elements}
Trevor Hastie, Robert Tibshirani, and Jerome Friedman.
\newblock \emph{The elements of statistical learning: data mining, inference,
  and prediction}.
\newblock Springer Science \& Business Media, 2009.

\bibitem[Hawkes(1971{\natexlab{a}})]{hawkes1971point}
Alan~G Hawkes.
\newblock Point spectra of some mutually exciting point processes.
\newblock \emph{Journal of the Royal Statistical Society: Series B
  (Methodological)}, 33\penalty0 (3):\penalty0 438--443, 1971{\natexlab{a}}.

\bibitem[Hawkes(1971{\natexlab{b}})]{hawkes1971spectra}
Alan~G Hawkes.
\newblock Spectra of some self-exciting and mutually exciting point processes.
\newblock \emph{Biometrika}, 58\penalty0 (1):\penalty0 83--90,
  1971{\natexlab{b}}.

\bibitem[Hawkes(2018)]{hawkes2018hawkes}
Alan~G Hawkes.
\newblock Hawkes processes and their applications to finance: a review.
\newblock \emph{Quantitative Finance}, 18\penalty0 (2):\penalty0 193--198,
  2018.

\bibitem[Jia et~al.(2020)Jia, Siegle, Durand, Heller, Ramirez, and
  Olsen]{jia2020multi}
Xiaoxuan Jia, Joshua~H Siegle, S{\'e}verine Durand, Greggory Heller, Tamina
  Ramirez, and Shawn~R Olsen.
\newblock Multi-area functional modules mediate feedforward and recurrent
  processing in visual cortical hierarchy.
\newblock \emph{bioRxiv}, 2020.

\bibitem[Kass and Ventura(2001)]{kass2001spike}
Robert~E Kass and Val{\'e}rie Ventura.
\newblock A spike-train probability model.
\newblock \emph{Neural computation}, 13\penalty0 (8):\penalty0 1713--1720,
  2001.

\bibitem[Koller and Friedman(2009)]{koller2009probabilistic}
Daphne Koller and Nir Friedman.
\newblock \emph{Probabilistic graphical models: principles and techniques}.
\newblock MIT press, 2009.

\bibitem[Kutoyants(1998)]{kutoyants1998statistical}
Yu~A Kutoyants.
\newblock \emph{Statistical inference for spatial Poisson processes}, volume
  134.
\newblock Springer Science \& Business Media, 1998.

\bibitem[Lewis and Shedler(1979)]{lewis1979simulation}
PA~W Lewis and Gerald~S Shedler.
\newblock Simulation of nonhomogeneous poisson processes by thinning.
\newblock \emph{Naval research logistics quarterly}, 26\penalty0 (3):\penalty0
  403--413, 1979.

\bibitem[Lewis(1970)]{lewis1970remarks}
PAW Lewis.
\newblock Remarks on the theory, computation and application of the spectral
  analysis of series of events.
\newblock \emph{Journal of Sound and Vibration}, 12\penalty0 (3):\penalty0
  353--375, 1970.

\bibitem[Mei and Eisner(2017)]{mei2017neural}
Hongyuan Mei and Jason~M Eisner.
\newblock The neural hawkes process: A neurally self-modulating multivariate
  point process.
\newblock \emph{Advances in neural information processing systems}, 30, 2017.

\bibitem[Meinshausen and B{\"u}hlmann(2006)]{meinshausen2006high}
Nicolai Meinshausen and Peter B{\"u}hlmann.
\newblock High-dimensional graphs and variable selection with the lasso.
\newblock 2006.

\bibitem[Mongillo et~al.(2018)Mongillo, Rumpel, and
  Loewenstein]{mongillo2018inhibitory}
Gianluigi Mongillo, Simon Rumpel, and Yonatan Loewenstein.
\newblock Inhibitory connectivity defines the realm of excitatory plasticity.
\newblock \emph{Nature neuroscience}, 2018.

\bibitem[Mugglestone and Renshaw(1996)]{mugglestone1996practical}
Moira~A Mugglestone and Eric Renshaw.
\newblock A practical guide to the spectral analysis of spatial point
  processes.
\newblock \emph{Computational Statistics \& Data Analysis}, 21\penalty0
  (1):\penalty0 43--65, 1996.

\bibitem[Murphy(2012)]{murphy2012machine}
Kevin~P Murphy.
\newblock \emph{Machine learning: a probabilistic perspective}.
\newblock MIT press, 2012.

\bibitem[Murphy and Van~der Vaart(2000)]{murphy2000profile}
Susan~A Murphy and Aad~W Van~der Vaart.
\newblock On profile likelihood.
\newblock \emph{Journal of the American Statistical Association}, 95\penalty0
  (450):\penalty0 449--465, 2000.

\bibitem[Ogata(1978)]{ogata1978asymptotic}
Yoshiko Ogata.
\newblock The asymptotic behaviour of maximum likelihood estimators for
  stationary point processes.
\newblock \emph{Annals of the Institute of Statistical Mathematics}, 1978.

\bibitem[Ogata(1981)]{ogata1981lewis}
Yosihiko Ogata.
\newblock On lewis' simulation method for point processes.
\newblock \emph{IEEE transactions on information theory}, 1981.

\bibitem[Ogata(1988)]{ogata1988statistical}
Yosihiko Ogata.
\newblock Statistical models for earthquake occurrences and residual analysis
  for point processes.
\newblock \emph{Journal of the American Statistical association}, 1988.

\bibitem[Pillow et~al.(2008)Pillow, Shlens, Paninski, Sher, Litke,
  Chichilnisky, and Simoncelli]{pillow2008spatio}
Jonathan~W Pillow, Jonathon Shlens, Liam Paninski, Alexander Sher, Alan~M
  Litke, EJ~Chichilnisky, and Eero~P Simoncelli.
\newblock Spatio-temporal correlations and visual signalling in a complete
  neuronal population.
\newblock \emph{Nature}, 454\penalty0 (7207):\penalty0 995--999, 2008.

\bibitem[Reynaud-Bouret and Schbath(2010)]{reynaud2010adaptive}
Patricia Reynaud-Bouret and Sophie Schbath.
\newblock Adaptive estimation for hawkes processes; application to genome
  analysis.
\newblock 2010.

\bibitem[Salehi et~al.(2019)Salehi, Trouleau, Grossglauser, and
  Thiran]{salehi2019learning}
Farnood Salehi, William Trouleau, Matthias Grossglauser, and Patrick Thiran.
\newblock Learning hawkes processes from a handful of events.
\newblock \emph{Advances in Neural Information Processing Systems}, 32, 2019.

\bibitem[Siegle et~al.(2021)Siegle, Jia, Durand, Gale, Bennett, Graddis,
  Heller, Ramirez, Choi, Luviano, et~al.]{siegle2021survey}
Joshua~H Siegle, Xiaoxuan Jia, S{\'e}verine Durand, Sam Gale, Corbett Bennett,
  Nile Graddis, Greggory Heller, Tamina~K Ramirez, Hannah Choi, Jennifer~A
  Luviano, et~al.
\newblock Survey of spiking in the mouse visual system reveals functional
  hierarchy.
\newblock \emph{Nature}, 592\penalty0 (7852):\penalty0 86--92, 2021.

\bibitem[Wang and Zhang(2022)]{wang2022hawkes}
Liwen Wang and Lin Zhang.
\newblock Hawkes processes for understanding heterogeneity in information
  propagation on twitter.
\newblock \emph{Frontiers in Physics}, page 970, 2022.

\bibitem[Wang et~al.(2018)Wang, Zhang, He, and Zha]{wang2018supervised}
Lu~Wang, Wei Zhang, Xiaofeng He, and Hongyuan Zha.
\newblock Supervised reinforcement learning with recurrent neural network for
  dynamic treatment recommendation.
\newblock In \emph{Proceedings of the 24th ACM SIGKDD international conference
  on knowledge discovery \& data mining}, pages 2447--2456, 2018.

\bibitem[Wang et~al.(2016)Wang, Xie, Du, and Song]{wang2016isotonic}
Yichen Wang, Bo~Xie, Nan Du, and Le~Song.
\newblock Isotonic hawkes processes.
\newblock In \emph{International conference on machine learning}, pages
  2226--2234. PMLR, 2016.

\bibitem[Wasserman(2004)]{wasserman2004all}
Larry Wasserman.
\newblock \emph{All of statistics: a concise course in statistical inference},
  volume~26.
\newblock Springer, 2004.

\bibitem[White(1982)]{white1982maximum}
Halbert White.
\newblock Maximum likelihood estimation of misspecified models.
\newblock \emph{Econometrica: Journal of the econometric society}, pages 1--25,
  1982.

\bibitem[Xu and Zha(2017)]{xu2017dirichlet}
Hongteng Xu and Hongyuan Zha.
\newblock A dirichlet mixture model of hawkes processes for event sequence
  clustering.
\newblock \emph{Advances in neural information processing systems}, 30, 2017.

\bibitem[Xu et~al.(2016)Xu, Farajtabar, and Zha]{xu2016learning}
Hongteng Xu, Mehrdad Farajtabar, and Hongyuan Zha.
\newblock Learning granger causality for hawkes processes.
\newblock In \emph{International conference on machine learning}, pages
  1717--1726. PMLR, 2016.

\bibitem[Yang et~al.(2018)Yang, Liu, Chen, and Hawkes]{yang2018applications}
Steve~Y Yang, Anqi Liu, Jing Chen, and Alan Hawkes.
\newblock Applications of a multivariate hawkes process to joint modeling of
  sentiment and market return events.
\newblock \emph{Quantitative finance}, 18\penalty0 (2):\penalty0 295--310,
  2018.

\bibitem[Yang et~al.(2017)Yang, Etesami, He, and Kiyavash]{yang2017online}
Yingxiang Yang, Jalal Etesami, Niao He, and Negar Kiyavash.
\newblock Online learning for multivariate hawkes processes.
\newblock \emph{Advances in Neural Information Processing Systems}, 30, 2017.

\bibitem[Zhang et~al.(2020)Zhang, Lipani, Kirnap, and Yilmaz]{zhang2020self}
Qiang Zhang, Aldo Lipani, Omer Kirnap, and Emine Yilmaz.
\newblock Self-attentive hawkes process.
\newblock In \emph{International conference on machine learning}, pages
  11183--11193. PMLR, 2020.

\bibitem[Zhou et~al.(2020)Zhou, Li, Fan, Wang, Sowmya, and
  Chen]{zhou2020efficient}
Feng Zhou, Zhidong Li, Xuhui Fan, Yang Wang, Arcot Sowmya, and Fang Chen.
\newblock Efficient inference for nonparametric hawkes processes using
  auxiliary latent variables.
\newblock \emph{Journal of Machine Learning Research}, 2020.

\bibitem[Zhou et~al.(2021{\natexlab{a}})Zhou, Kong, Zhang, Feng, and
  Zhu]{zhou2021nonlinear}
Feng Zhou, Quyu Kong, Yixuan Zhang, Cheng Feng, and Jun Zhu.
\newblock Nonlinear hawkes processes in time-varying system.
\newblock \emph{arXiv preprint arXiv:2106.04844}, 2021{\natexlab{a}}.

\bibitem[Zhou et~al.(2021{\natexlab{b}})Zhou, Zhang, and
  Zhu]{zhouj2021efficient}
Feng Zhou, Yixuan Zhang, and Jun Zhu.
\newblock Efficient inference of flexible interaction in spiking-neuron
  networks.
\newblock In \emph{International Conference on Learning Representations},
  2021{\natexlab{b}}.
\newblock URL \url{https://openreview.net/forum?id=aGfU_xziEX8}.

\bibitem[Zhou et~al.(2022)Zhou, Kong, Deng, Kan, Zhang, Feng, and
  Zhu]{zhou2022efficient}
Feng Zhou, Quyu Kong, Zhijie Deng, Jichao Kan, Yixuan Zhang, Cheng Feng, and
  Jun Zhu.
\newblock Efficient inference for dynamic flexible interactions of neural
  populations.
\newblock \emph{Journal of Machine Learning Research}, 2022.

\bibitem[Zhou et~al.(2013{\natexlab{a}})Zhou, Zha, and Song]{zhou2013learning}
Ke~Zhou, Hongyuan Zha, and Le~Song.
\newblock Learning social infectivity in sparse low-rank networks using
  multi-dimensional hawkes processes.
\newblock In \emph{Artificial Intelligence and Statistics}, pages 641--649.
  PMLR, 2013{\natexlab{a}}.

\bibitem[Zhou et~al.(2013{\natexlab{b}})Zhou, Zha, and
  Song]{zhou2013learningtrigger}
Ke~Zhou, Hongyuan Zha, and Le~Song.
\newblock Learning triggering kernels for multi-dimensional hawkes processes.
\newblock In \emph{International conference on machine learning}, pages
  1301--1309. PMLR, 2013{\natexlab{b}}.

\bibitem[Zuo et~al.(2020)Zuo, Jiang, Li, Zhao, and Zha]{zuo2020transformer}
Simiao Zuo, Haoming Jiang, Zichong Li, Tuo Zhao, and Hongyuan Zha.
\newblock Transformer hawkes process.
\newblock In \emph{International conference on machine learning}, pages
  11692--11702. PMLR, 2020.

\end{thebibliography}

\newpage
\title{Short-term Temporal Dependency Detection under Heterogeneous Event Dynamic with Hawkes Processes \\
(Supplementary Material)}

%
%



\onecolumn 
\maketitle

\appendix

In section \ref{appendix:jitter_optimization_algo}, we provide details of optimization algorithms, and condition inference based cross-correlogram. 
Proofs of the propositions in the main text can be found in section \ref{appendix:proofs}.
Section \ref{appendix:simulatio_empirical} lists all simulation scenarios for the model verification.
Section \ref{appendix:use_cases} presents some use cases and versatility of our new tool.
Section \ref{appendix:theoretical_derivations} derives the approximated analytical formulae of the model's properties in a special situation.
Finally, more details about the neuroscience experiments are in section \ref{appendix:jitter_app}.

\section{Algorithms} \label{appendix:jitter_optimization_algo}

\subsection{Updating rules}
\begin{gather*}
\min_{h_{i\to j}, \beta_j, \beta_w, \sigma_w } 
\left\{
-\sum_{s \in N_j } \log \tilde{\lambda}_{j}(s)
+ \int_0^T \tilde{\lambda}_{j}(s) \mathrm{d} s  \right\} \\
\tilde{\lambda}_{j}(t) := 
\left( \beta_j
+ \beta_w \; \overline{\textbf{s}_{i}}(t) 
+ \int_0^t h_{i\to j}(t-\tau) \mathrm{d} N_{i}(\tau) \right)_+ \\
\overline{\textbf{s}_{i}}(t)
= \int_0^T W(t-s) \mathrm{d} N_i(s)
= \sum_{t_m \in N_i} W(t - t_m)
\end{gather*}
Let $\phi_w, \phi_h$ be the bases defined as,
\begin{align} \label{eq:regression_bases}
\phi_w(t) := \int W(t-s) \mathrm{d} N_i(s),\qquad
\phi_h(t) := \int h_{i\to j}(t-s) \mathrm{d} N_i(s)
\end{align}
The intensity can be rewritten in a linear form,
\begin{equation}\label{eq:intensity_linear}
\tilde{\lambda}_{j}(t) =
\beta_j \cdot 1
+ \beta_w \phi_w(t)
+ \beta_h \phi_h(t)
= \Psi(t) \boldsymbol\beta
\end{equation}
$\Psi(t)$ represents all bases, $\boldsymbol\beta$ is a vector of the coefficients.
If the impact function is fitted using non-parametric method, such as general additive model or splines
\begin{equation*}
h_{i\to j}(s) = \beta_{h,1} B_1(s) + ...+ \beta_{h,k} B_k(s)
\end{equation*}
where $B_1,...,B_k$ are spline bases for the impact function.
Define
\begin{equation*}
\phi_{h,1}(t) := \int B_1(t-s) \mathrm{d} N_i(s),...,\;
\phi_{h,k}(t) := \int B_k(t-s) \mathrm{d} N_i(s)
\end{equation*}
The intensity still maintains the linear form:
\begin{equation*}
\tilde{\lambda}_{j}(t) =
\beta_j \cdot 1
+ \beta_w \phi_w(t)
+ \beta_{h,1} \phi_{h,1}(t)
+...+ \beta_{h,k} \phi_{h,k}(t)
\end{equation*}

The target equation of the model can be optimized using gradient descent.
For a fixed $\sigma_w$, the target is convex and the optimization is efficient using Newton's method.
The first-order and second-order derivatives of the target equations are,
\begin{align*}
\frac{\partial \tilde{\ell}}{\partial \boldsymbol\beta } 
=& -\int_0^T \frac{\Psi(s)}{\tilde\lambda_j(s)}\mathrm{d}N_j(s)
+ \int_0^T \Psi(s) \mathrm{d}s \\
\frac{\partial^2 \tilde{\ell}}{\partial \boldsymbol\beta \partial \boldsymbol\beta^T}
=& \int_0^T \frac{\Psi(s) \Psi(s)^T}{\tilde\lambda_j(s)^2} \mathrm{d}N_j(s).
\end{align*}

The update of $\sigma_w$ can be done as a separate step using a gradient as follows.
If $W$ is a Gaussian kernel function, then
\begin{equation}
\begin{aligned}
&\frac{\partial \tilde{\ell}}{\partial \sigma_w}
= -\sum_{t_n \in N_j} \frac{\beta_w }{\tilde \lambda(t_n)}
\frac{\partial }{\partial \sigma_w} 
 \overline{\textbf{s}_{i}}(t_n) 
 + \beta_w \frac{\partial }{\partial \sigma_w} 
    \int_0^T \overline{\textbf{s}_{i}}(u)  \mathrm{d}u \\
=& -\sum_{t_n \in N_j} \frac{\beta_w }{\tilde \lambda(t_n)}
\sum_{t_m\in N_i} \frac{\partial }{\partial \sigma_w}  W(t_n - t_m) 
+ \beta_w \sum_{t_m\in N_i} \frac{\partial }{\partial \sigma_w} \int_0^T W(u-t_m)\mathrm{d}u \\
=& -\sum_{t_n \in N_j} \frac{\beta_w }{\tilde \lambda(t_n)}
\sum_{t_m\in N_i} \frac{\partial }{\partial \sigma_w}  W(t_n - t_m)
    + \beta_w \sum_{t_m\in N_i} 
    \left( W(T-t_m) \frac{\partial }{\partial \sigma_w}  W(T-t_m)
    - W(-t_m) \frac{\partial }{\partial \sigma_w} W(-t_m) \right)
\end{aligned}
\end{equation}
where 
\begin{equation}
\frac{\partial }{\partial \sigma_w}  W(x)
= \left(
-\frac{1}{\sqrt{2\pi}\sigma_w^2 }  
+ \frac{x^2}{\sqrt{2\pi} \sigma_w^4}
\right)  \exp\left\{ -\frac{x^2}{2 \sigma_w^2} \right\}
\end{equation}

$\sigma_w$ can be optimized using grid-search during warmup.
In Appendix \ref{appendix:bayesian}, we discuss the sampling-based method,
which shows incorporating the uncertainty of $\sigma_w$ or fixing it at the optimal does not make a significant difference.

\subsection{Integral trick}
One computational advantage of the proposed model in main \eqref{eq:intensity_linear} is that the integral $\int \Psi(s)\mathrm{d}s$ can be calculated in the closed form if the bases are designed carefully.
In contrast, models with intensity in logarithmic scale
$\log \lambda(t) = \Psi(t)\boldsymbol\beta$ do not enjoy this computation convenience.
For example, the derivative of the modified negative log-likelihood function becomes,
\begin{align*}
\frac{\partial \tilde{\ell}}{\partial \boldsymbol\beta } 
=& -\int_0^T \Psi(s)\mathrm{d}N_j(s)
+ \int_0^T \Psi(s) e^{\Psi(s) \boldsymbol\beta } \mathrm{d}s.
\end{align*}
Usually, it is not tractable to calculate the integral in the second term,
so it is approximated by discretizing the continuous functions.
This is the reason our model does not involve discretization or require specifying the time resolution.
Another benefit of using a continuous-time model is that the number of data points is small, which is proportional to the number of spikes instead of the number of time bins. For example, if the bin width is 1 ms, then for one 1-second long trial, it needs to store 1000 data points. If the trial has 20 spikes, the continuous-time model only needs to keep 20 data points. The memory space is 50 times smaller.

If the regression bases have form Eq \eqref{eq:regression_bases} with kernel, then 
\begin{equation*}
\int_0^T \Psi(t) \mathrm{d}t
= \int_0^T \int_0^T K(t-s) N_i(\mathrm{d}s)
= N_i(T) \int_{\mathrm{R}} K(s)\mathrm{d}s.
\end{equation*}
If $K$ is a Normal window function or a square window function, the above integral is simple. The boundary effect can be removed in the integral by only considering a few time points close to 0 or T.
Next, we show how to calculate such an integral if $K$ is B-spline, which is widely used in non-parametric curve fitting. An example can be found in Appendix \ref{appendix:nonparametric}.

The B-splines are defined using Cox-de Boor recursion equations. $t_i$ are knots (with repeated padding). $p$ is the degree of the spline polynomial.
When $p=3$, these are the cubic splines.
\begin{align*}
B_{i,0}(x) =& \mathbb{I}_{[t_i, t_{i+1})}(x) \\
B_{i,p}(x) =& \frac{x - t_i}{t_{i+p} - t_i} B_{i,p-1}(x) 
         + \frac{t_{i+p+1} - x}{t_{i+p+1} - t_{i+1}} B_{i+1,p-1}.
\end{align*}
Knot padding is important to create proper splines. If $p=3$ and the distinct knot locations are $(0,1,2)$, the input knots should be $(0,0,0,0,1,2,2,2,2)$. The knots need extra $p$ repeated knots of the two ends. If there are $K$ distinct knots, then there are $K+2p$ input knots. The total number of basis is $K+p-1$.

\begin{lemma} \label{lemma:bspline_integral}
For the B-spline curve defined above, the integral of the curve has closed-form as follows,
\begin{align}
\int_{-\infty}^{\infty} B_{i,p}(s) \mathrm{d}s
= \frac{t_{i+p+1} - t_{i}}{p+1}.
\end{align}
\end{lemma}
\begin{proof}
The support of each basis spans over $p+1$ knot-intervals (including the padded knots on the ends),
\begin{align*}
\text{supp}(B_{i,p}) =& [t_i, t_{i+p+1})
\end{align*}
\begin{align*}
\frac{\mathrm{d}}{\mathrm{d}x} B_{i,p}(x) 
= \frac{p}{t_{i+p} - t_i} B_{i,p-1}(x) 
   - \frac{p}{t_{i+p+1} - t_{i+1}} B_{i+1,p-1}(x).
\end{align*}
The support of the derivative is almost the same as the basis except for a few 0 derivative points. 
\begin{align*}
\text{supp}(\frac{\mathrm{d}}{\mathrm{d}x} B_{i,p}) \subseteq & [t_i, t_{i+p+1})
\end{align*}
We reform the derivative properties to get the integral \citep{bhatti2006calculation}.
\begin{equation*}
\frac{\mathrm{d}}{\mathrm{d}x} \sum_{i=0}^{\infty} c_i B_{i,p+1}(x)
= \sum_{i=0}^{\infty} (p+1) \frac{c_i - c_{i-1}}{t_{i+p+1} - t_{i}}  B_{i,p}(x)
\end{equation*}
$c_i$ are some arbitrary coefficients. Next we set $c_0,..., c_{i-1} = 0$, $c_i, c_{i+1},... = 1$.
\begin{equation*}
\frac{\mathrm{d}}{\mathrm{d}x} \sum_{j=i}^{\infty} c_j B_{j,p+1}(x)
= \frac{\mathrm{d}}{\mathrm{d}x} \sum_{j=i}^{i+p} B_{j,p+1}(x)
= \frac{p+1}{t_{i+p+1} - t_{i}} B_{i,p}(x)
\end{equation*}
The first equation simplifies the sum due to the supports of bases.
Then take the integral on both side,
\begin{align*}
\int_{-\infty}^x B_{i,p}(s) \mathrm{d}s
= \int_{t_i}^x B_{i,p}(s) \mathrm{d}s
= \frac{t_{i+p+1} - t_{i}}{p+1} \sum_{j=i}^{\infty} B_{j,p+1}(s)
= \frac{t_{i+p+1} - t_{i}}{p+1} \sum_{j=i}^{i+p} B_{j,p+1}(x)
\end{align*}
The area under the curve of a basis is,
\begin{align*}
\int_{-\infty}^{\infty} B_{i,p}(s) \mathrm{d}s
= \int_{t_i}^{t_{i+p+1}} B_{i,p}(s) \mathrm{d}s
= \frac{t_{i+p+1} - t_{i}}{p+1} \sum_{j=i}^{i+p} B_{j,p+1}(t_{i+p+1})
\end{align*}
Consider the summation term,
\begin{align*}
\sum_{j=i}^{i+p} B_{j,p+1}(t_{i+p+1})
=& B_{i,p+1}(t_{i+p+1}) + B_{i+1,p+1}(t_{i+p+1}) + ... + B_{i+p,p+1}(t_{i+p+1}) \\
=& \Big(\frac{t_{i+p+1} - t_i}{t_{i+p+1} - t_i} B_{i,p}(t_{i+p+1}) 
        +\frac{t_{i+p+2} - t_{i+p+1}}{t_{i+p+2} - t_{i+1}} B_{i+1,p}(t_{i+p+1}) \Big) \\
+& \Big(\frac{t_{i+p+1} - t_{i+1}}{t_{i+p+2} - t_{i+1}} B_{i+1,p}(t_{i+p+1}) 
        +\frac{t_{i+p+3} - t_{i+p+1}}{t_{i+p+3} - t_{i+2}} B_{i+2,p}(t_{i+p+1}) \Big) + ... \\
+& \Big(\frac{t_{i+p+1} - t_{i+p}}{t_{i+2p+1} - t_{i+p}} B_{i+p,p}(t_{i+p+1}) 
        +\frac{t_{i+2p+2} - t_{i+p+1}}{t_{i+2p+2} - t_{i+p+1}} B_{i+p+1,p}(t_{i+p+1}) \Big) \\
=& B_{i,p}(t_{i+p+1}) + B_{i+1,p}(t_{i+p+1}) +... + B_{i+p+1,p}(t_{i+p+1}) \\
=& B_{i,p-1}(t_{i+p+1}) + B_{i+1,p}(t_{i+p+1}) +...+ B_{i+p+2,p-1}(t_{i+p+1}) \\
=& B_{i,0}(t_{i+p+1}) + B_{i+1,0}(t_{i+p+1}) +... + B_{i+2p+1,0}(t_{i+p+1}) = 1 \\
\end{align*}
So the conclusion holds.
\end{proof}

Another popular example with closed-form integral is the exponential coupling function, which also enjoys the integral trick.
\begin{equation}
h_{i\to j}(\tau) = e^{-\gamma_{i\to j}(\tau)} \mathbb{I}(\tau\geq 0)
\end{equation}
The modified negative log-likelihood is,
\begin{equation}
\begin{aligned}
\tilde{\ell} =& - \sum_{t_n\in N_j} \log \left( \beta_j 
    + \overline{\textbf{s}_{i}}(t_n)
    + \alpha_{i\to j} \sum_{t_m\in N_i, t_m<t_n} 
        e^{-\gamma_{i\to j}(t_n-t_m)} \right) \\
&+ \beta_j T
+ \int_0^T \overline{\textbf{s}_{i}}(s) \mathrm{d}s
+ \frac{\alpha_{i\to j} }{ \gamma_{i\to j} }
    \sum_{t_m \in N_i} \left(1- e^{-\gamma_{i\to j}(T-t_m)} \right)
\end{aligned}
\end{equation}
where $\alpha_{i\to j}$ is the coefficient of the exponential basis.
The derivatives of the target equation over $\beta_j, \beta_w, \alpha_{i\to j}$ are similar to other linear models.
The derivative over the timescale parameter $\gamma_{i\to j}$ is
\begin{equation}
\begin{aligned}
\frac{\partial \tilde{\ell} }{ \partial \gamma_{i\to j} }
=& \sum_{t_n\in N_j} \frac{ \alpha_{i\to j} }{ \tilde \lambda (t_n) }
    \sum_{t_m\in N_i, t_m<t_n} (t_n-t_m) e^{-\gamma_{i\to j}(t_n-t_m)} \\
& - \frac{\alpha_{i\to j} }{ \gamma_{i\to j}^2 } 
    \sum_{t_m\in N_i}(1-e^{-\gamma_{i\to j}(T-t_m)})
+ \frac{\alpha_{i\to j} }{ \gamma_{i\to j} } 
    \sum_{t_m\in N_i} (T-t_m) e^{-\gamma_{i\to j}(T-t_m)}
\end{aligned}
\end{equation}

\subsection{Details of conditional inference based cross-correlogram (CCG) } \label{appendix:ccg}

\begin{figure}[ht]
\centering
\begin{subfigure}{0.6\textwidth}
\centering
\includegraphics[width=0.99\linewidth]{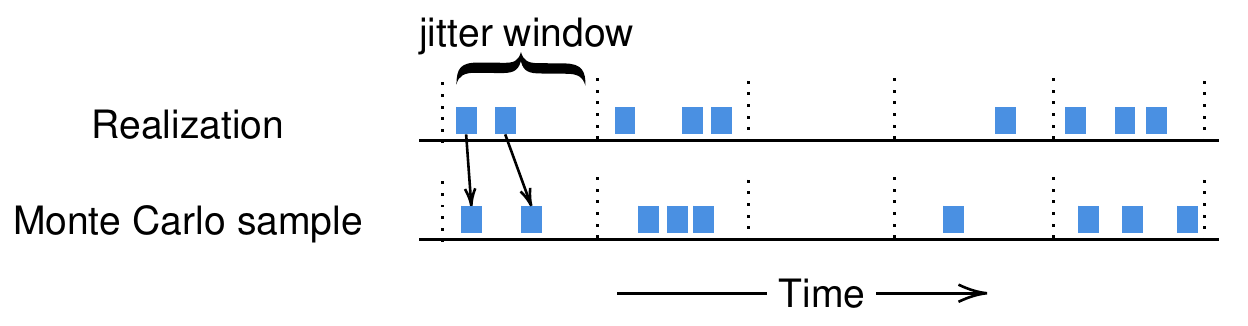}
\end{subfigure}
\caption{Construction of Monte Carlo samples from the null in condition inference based CCG. Blue dots are timestamps.
}
\end{figure}

The calculation of CCG needs to discretize the point processes first. The time bin is usually set with a small size so that each bin contains at most one point.
Consider two time-binned processes $X_i$ and $X_j$ with respect to the count processes $N_i$ and $N_j$.
The correlation (or cross-correlation) of timing at a certain lag $\tau \in \mathbb{N}$ ($X_i$ leads $X_j$) with fluctuating background activity is assessed in the following procedure:

\begin{enumerate}
\item Calculate the test statistic CCG at $\tau$,
\begin{equation}
\mathrm{CCG}(\tau) 
= \sum_n X_i(n-\tau) X_j(n)
\end{equation}

\item Divide the timeline into equal-width jitter windows $\Delta$ as shown in the figure above.

\item One Monte Carlo sample from the null distribution is generated by uniformly allocating the time points within each jitter window, as shown in the figure.
Such a process can be applied to either $N_i$ or $N_j$, or both. 
Empirically, these options do not make a big difference.
Then time-bin the "jittered" samples $\tilde{X}_i$, $\tilde{X}_j$.

\item Calculate the CCG of the Monte Carlo sample,
\begin{equation}
\mathrm{CCG}^{\mathrm{MC}}(\tau) 
= \sum_n \tilde{X}_i(n-\tau) \tilde{X}_j(n)
\end{equation}

\item Repeat step 3 and 4 multiple times $N_{\mathrm{MC}}$ to acquire the null distribution of CCG.
The p-value of the test is
\begin{equation}
\mathrm{p-value}
= \frac{N_{\tau} + 1}{N_{\mathrm{MC}} + 1}
\end{equation}
where $N_{\tau}$ is the number of Monte Carlo samples
$\mathrm{CCG}(\tau) < \mathrm{CCG}^{\mathrm{MC}}(\tau)$ (or the other way if $\mathrm{CCG}(\tau)$ is on the left tail).
Similarly, the acceptance band can be constructed from the null distribution.

\end{enumerate}

$\Delta$ is chosen with prior knowledge of the data, which is roughly the timescale of the background activity, such that jittering the samples within window $\Delta$ maintains the intensity of $N_i, N_j$, but it can "break" the fine time structures without the association of timing.
As the time points are jittered with only a small amount, the null samples are assumed to have the same intensity, so the method smartly bypasses the estimation of the dynamic background.
For better visualization, the mean of the null is usually subtracted from the test statistic so the CCG curve and the acceptance band can center around zero such as main Figure \ref{fig:sim_demo} and \ref{fig:neural_demo}.
In main Figure \ref{fig:sim_demo} and \ref{fig:neural_demo}, the jitter window is $\Delta=120$ ms, the discretization time bin width is 2 ms, the null distribution uses 1000 Monte Carlo samples.

\section{Proofs} \label{appendix:proofs}

We list the regularity conditions needed for our statements here. These technical conditions first include follows from Assumptions $A,B,C$ in \cite{ogata1978asymptotic}, which makes sure the consistency of $\hat{\boldsymbol{\theta}}$ converges to $\boldsymbol\theta_{KL}= \operatorname*{argmin}_{\boldsymbol\theta }\Lambda(\boldsymbol\theta):=\mathbb E\ell({\boldsymbol\theta })$ for the misspecified model. Then, we make the additional assumption:


\begin{assumption}\label{ass}
$\Lambda(\cdot)$ is $\mu$-strongly convex and has $L$-Lipschitz gradient.
\end{assumption}

With this assumption, we can guarantee that for $\boldsymbol\theta_1$ and $\boldsymbol\theta_2$, we have 
\begin{equation*}
    \mu\|\boldsymbol\theta_1-\boldsymbol\theta_2\|\leq\|\nabla\Lambda(\boldsymbol\theta_1)-\nabla\Lambda(\boldsymbol\theta_2)\|\leq L\|\boldsymbol\theta_1-\boldsymbol\theta_2\|
\end{equation*}
which gives $\|\nabla\Lambda(\boldsymbol\theta_1)-\nabla\Lambda(\boldsymbol\theta_2)\|=\Theta(\|\boldsymbol\theta_1-\boldsymbol\theta_2\|)$.

We state a Lemma. Before we state the Lemma, let as simplify the notation using 2,1 vs $i,j$ and absorb the baseline and heterogeneity into one function $f$ so that we have \begin{equation}\label{absorb}
\begin{aligned}
\lambda_{1}(t)
&=f_{1}(t) 
+ c\int_0^t\mathbf{1}_{[0,\sigma_h]}(t-s)dN_2(s) \\
%
\lambda_{2}(t) 
&= f_{2}(t)
\end{aligned}
\end{equation}
This will simplify the notation in the Lemma below. We always explicitly make clear which version of notation we are using before stating results.
\begin{lemma}\label{p_on_b}
Fixed $a$ and $b$, assume $f_2$ is continuous everywhere with bounded gradient, then
\begin{align}\label{approx2}
    \mathbb E\bigg[\frac{1}{a + b\int_0^t\mathbf{1}_{[0,\sigma_h]}(t-s)dN_2(s)}\bigg]=&\frac{1-f_2(t) \sigma_h}{a}+\frac{f_2(t)}{a+b}\sigma_h+o(\sigma_h) \nonumber\\
    =& \frac{1}{a}-\frac{b}{a(a+b)} f_2(t)\sigma_h+o(\sigma_h)
\end{align}
\end{lemma}
\begin{proof}
    Simply note that $\int_0^t\mathbf{1}_{[0,\sigma_h]}(t-s)dN_2(s) $ is a Poisson distribution with parameter (also mean) $\int_{t-\sigma_h}^tf_2(s)ds=f_2(t)\sigma_h+o(\sigma_h)$. The rest follows direct calculation using Poisson p.m.f.
\end{proof}

We also state another Lemma which will be useful later, based on taylor expansion of ratio function: 
\begin{lemma}\label{ratioe}
    Given finitely supported R.V. $X$ and $Y$ with mean $\mu_x$ and $\mu_y$, if we define $r_x = \|X-\mu_x\|_\infty$ and $r_y = \|Y-\mu_y\|_\infty$ and $r=\max(r_x,r_y)$. Suppose $\frac{X}{Y}$ is always strictly bounded away from 0 by some fixed positive constant, then
\begin{align*}
    \frac{X}{Y} =\frac{\mu_x}{\mu_y} -\frac{\mu_x}{\mu_y^2}(Y-\mu_y) +\frac{1}{\mu_y}(X-\mu_x)+ o(r)
\end{align*}
or higher order approximation:
\begin{align*}
    \frac{X}{Y} =\frac{\mu_x}{\mu_y} -\frac{\mu_x}{\mu_y^2}(Y-\mu_y) +\frac{1}{\mu_y}(X-\mu_x)+ \frac{\mu_x}{\mu_y^3}(Y-\mu_y)^2 -\frac{1}{\mu_y^2}(X-\mu_x)(Y-\mu_y) + o(r^2).
\end{align*}
Thus, 
\begin{equation*}
    \mathbb E\frac{X}{Y} = \frac{\mu_x}{\mu_y} +\frac{\text{Var}(Y)\mu_x}{\mu_y^3} -\frac{\text{Cov}(X,Y)}{\mu_y^2} + o(r^2).
\end{equation*}
\end{lemma}
\begin{proof}
    Omitted.
\end{proof}

Thus, if we use the above high order approximation:
\begin{equation*}
    R(X,Y) := \frac{\mu_x}{\mu_y} +\frac{\text{Var}(Y)\mu_x}{\mu_y^3} -\frac{\text{Cov}(X,Y)}{\mu_y^2},
\end{equation*}
we have $\mathbb E[\frac{X}{Y}]=R(X,Y)+o(r^2)$.

\subsection{Proof of Proposition \ref{prop1}.}
We first prove Proposition \ref{prop1}. In this proof we use original intensity notation \eqref{eq:true_model}.

\begin{proof}[Proof of Proposition \ref{prop1}]
Using \eqref{ogte} (or see \cite{ogata1978asymptotic}), we can show that 
\begin{align*}
    \Lambda(\theta) = &\mathbb E\bigg[-\int_0^P \lambda_\theta(t)dt+\int_0^P \log\lambda_\theta(t)dN_1(t)\bigg] \nonumber\\
    =&\mathbb E\bigg[-\int_0^P\lambda_1(t)\bigg( \frac{\lambda_\theta(t)}{\lambda_1(t)}-\log\lambda_\theta(t)\bigg)dt\bigg]
\end{align*}
for 
\begin{align}\label{special}
    \lambda_1(t) =& \alpha_1 + f_1(t)+ c\int_{0}^t\mathbf{1}_{[0,\sigma_h]}(t-s) dN_2(s).
\end{align}
where we set $\int_0^T f_1(t)dt=0$ to avoid identifiability issue with $\alpha_1$ and parametrize  
\begin{equation*}
    \lambda_\theta(t) = \theta_1 + \theta_2\int_0^t\mathbf{1}_{[0,\sigma_h ]}(t-s)dN_2(s),
\end{equation*}
for MLE. Using Lemma \ref{p_on_b} and $\int_0^T f_1(t)dt=0$, one can show
\begin{align*}
    \frac{\partial\Lambda}{\partial\theta_2}_{|\boldsymbol\theta=(\alpha_1,c)} = &\mathbb E\bigg[-\int_0^T\frac{\partial \lambda_\theta}{\partial \theta_2}\bigg( 1-\frac{\lambda_1^*(t)}{\lambda_\theta(t)}\bigg)dt\bigg]\nonumber\\
    =&-\int_0^T\mathbb E\bigg[\frac{f_1(t)\int_0^t\mathbf{1}_{[0,\sigma_h]}(t-s)dN_2(s)}{\alpha_1 + c\int_0^t\mathbf{1}_{[0,\sigma_h]}(t-s)dN_2(s)}\bigg]dt\nonumber\\
    =& -\int_0^T\frac{f_1(t)}{c}-\frac{f_1(t)\mu}{c}\mathbb E\bigg[\frac{1}{\alpha_1 + c\int_0^t\mathbf{1}_{[0,\sigma_h]}(t-s)dN_2(s)}\bigg]dt\nonumber\\
    =&  -\int_0^T\frac{ f_1(t)f_2(t)\sigma_h}{\alpha_1+c} dt+ o(\sigma_h) dt
\end{align*}

Now, note $\boldsymbol\theta_{KL}$ corresponds to the point where $\frac{\partial\Lambda}{\partial\theta_2}=0$. The rest follows from Assumption \ref{ass}.

\end{proof}

\subsection{Preliminary work for proofs of Proposition \ref{prop2} and \ref{prop3}.}
Hence forth we use the alternative intensity notation \eqref{absorb}. We lay some ground work before proving Proposition \ref{prop2} and \ref{prop3}. Again we restate the notation \eqref{absorb}, but also parametrize the density as $\theta$ (the impact function parameter) and $\boldsymbol \eta$ (all else is nuisance parameter):

\begin{align}\label{special}
    \lambda_1(t) =& f_1(t)+ c\int_{0}^t\mathbf{1}_{[0,\sigma_h]}(t-s) dN_2(s)\nonumber\\
    \lambda_2(t) =& f_2(t)\nonumber\\
    \lambda_{\theta,\boldsymbol \eta} (t) =& \sum_{i=1}^{M}\eta_i g_i(t) +\theta \int_{0}^t\mathbf{1}_{[0,\sigma_h]}(t-s) dN_2(s)
\end{align}
and re-define 
\begin{equation*}
   \ell(\theta,\boldsymbol\eta; \mathcal H_T)= -\int_0^T \lambda_{\theta,\boldsymbol\eta}(t)dt+\int_0^T \log\lambda_{\theta,\boldsymbol\eta}(t)dN_1(t)
\end{equation*}
and consequently $\Lambda(\theta,\boldsymbol\eta) = \mathbb E[\ell(\theta,\boldsymbol\eta)]$. However, we characterize a concept, on population level, similar to profile likelihood
\begin{equation*}
    \Lambda_p(\theta) = \sup_{\boldsymbol\eta} \mathbb E\bigg[-\int_0^T \lambda_{\theta,\boldsymbol\eta}(t)dt+\int_0^T \log\lambda_{\theta, \boldsymbol\eta}(t)dN_1(t)\bigg],
\end{equation*}
So, we can define 
\begin{equation*}
    \boldsymbol\eta(\theta) = \operatorname*{argmax}_{\boldsymbol \eta} \mathbb E\bigg[-\int_0^T \lambda_{\theta,\boldsymbol\eta}(t)dt+\int_0^T \log\lambda_{\theta, \boldsymbol\eta}(t)dN_1(t)\bigg]
\end{equation*}
so that $\Lambda (\theta,\boldsymbol\eta(\theta)) = \Lambda_p(\theta)$. As a result of Assumption \ref{ass}, $\Lambda_p$ is also $\mu$-strongly convex with $L$-Lipschitz gradient. Now, we first characterize $\boldsymbol\eta(c)$ which corresponds to the equatiosn $\frac{\partial\Lambda(c,\boldsymbol\eta)}{\partial \eta_i}=0$ for all $i$. To analyze this term, we write:
\begin{align*}
    &\frac{\partial\Lambda(c,\boldsymbol\eta)}{\partial \eta_i} \nonumber\\
    = &-\mathbb E\bigg[\int_0^T\frac{\partial\lambda_{\theta,\boldsymbol\eta}}{\partial \eta_i}(t)\bigg(1-\frac{\lambda_1(t)}{\lambda_{\theta,\boldsymbol\eta}(t)}\bigg)dt\bigg]\nonumber\\
    =& -\mathbb E\bigg[\int_0^Tg_i(t)\bigg(\frac{\sum_{i=1}^M [\boldsymbol\eta(c)]_ig_i(t)-f_1(t)}{\sum_{i=1}^M [\boldsymbol\eta(c)]_ig_i(t)+c\int_0^T\textbf 1_{[0,\sigma_h]}(t-s)dN_2(s)}\bigg)dt\bigg]\nonumber\\
    =& -\mathbb E\bigg[\int_0^T\frac{g_i(t)\Big(\sum_{i=1}^M [\boldsymbol\eta(c)]_ig_i(t)-f_1(t)\Big)}{\sum_{i=1}^M [\boldsymbol\eta(c)]_ig_i(t)} -\frac{c\Big(\sum_{i=1}^M [\boldsymbol\eta(c)]_ig_i(t)-f_1(t)\Big)f_2(t)\sigma_h}{\sum_{i=1}^M [\boldsymbol\eta(c)]_ig_i(t)\Big(c+\sum_{i=1}^M [\boldsymbol\eta(c)]_ig_i(t)\Big)}+o(\sigma_h)dt\bigg]
\end{align*}

Set $ \frac{\partial\Lambda(c,\boldsymbol\eta)}{\partial \eta_i}=0$ for $1\leq i\leq M$, we can solve for $\boldsymbol\eta(c)$, we can show, using Assumption \ref{ass}, if we had a $\tilde{\boldsymbol\eta}_c$ that satisfies
\begin{equation*}
    0 = \mathbb E\bigg[\int_0^Tg_i(t)(1-\frac{f_1(t)}{\sum_{i=1}^M [\tilde{\boldsymbol\eta}_c]_ig_i(t)})dt\bigg]
\end{equation*}
for all $i$ or, if all process are stationary:
\begin{equation}\label{horder0}
    0 = \mathbb E\bigg[g_i\bigg(1-\frac{f_1}{\sum_{i=1}^M [\tilde{\boldsymbol\eta}_c]_ig_i}\bigg)\bigg ],
\end{equation}
for all $i$, then
\begin{equation}\label{horder1}
    \|\tilde{\boldsymbol\eta}_c-\boldsymbol\eta(c)\|=O(\sigma_h).
\end{equation}

Now, we can investigate estimation for $\theta$:
\begin{align*}
    &\frac{\partial\Lambda_p(\theta)}{\partial \theta}_{|\theta=c}\nonumber\\
    =& \frac{\partial\Lambda(\theta,\boldsymbol\eta)}{\partial \theta}_{|\theta=c,\boldsymbol\eta=\boldsymbol\eta(c)}  \nonumber\\
    =& -\mathbb E\bigg[\int_0^T\frac{\partial\lambda_{\theta,\boldsymbol\eta}}{\partial \theta}(t)\bigg(1-\frac{\lambda_1(t)}{\lambda_{\theta,\boldsymbol\eta}(t)}\bigg)dt\bigg]\nonumber\\
    =&-\mathbb E\bigg[\int_0^T\bigg(\int_0^T\textbf 1_{[0,\sigma_h]}(t-s)dN_2(s)\bigg)\bigg(\frac{\sum_{i=1}^M [\boldsymbol\eta(c)]_ig_i(t)-f_1(t)}{\sum_{i=1}^M [\boldsymbol\eta(c)]_ig_i(t)+c\int_0^T\textbf 1_{[0,\sigma_h]}(t-s)dN_2(s)}\bigg)dt\bigg]\nonumber\\
    =&-\mathbb E\bigg[\int_0^T\frac{\sum_{i=1}^M [\boldsymbol\eta(c)]_ig_i(t)-f_1(t)}{c} -\frac{\bigg(\sum_{i=1}^M [\boldsymbol\eta(c)]_ig_i(t)-f_1(t)\bigg)\bigg(\sum_{i=1}^M [\boldsymbol\eta(c)]_ig_i(t)\bigg)}{c\bigg(\sum_{i=1}^M [\boldsymbol\eta(c)]_ig_i(t)+c\int_0^T\textbf 1_{[0,\sigma_h]}(t-s)dN_2(s)\bigg)}dt\bigg]\nonumber\\
    =&-\mathbb E\bigg[\int_0^T\frac{\sum_{i=1}^M [\boldsymbol\eta(c)]_ig_i(t)-f_1(t)}{c} -\frac{\bigg(\sum_{i=1}^M [\boldsymbol\eta(c)]_ig_i(t)-f_1(t)\bigg)\bigg(\sum_{i=1}^M [\boldsymbol\eta(c)]_ig_i(t)\bigg)}{c\sum_{i=1}^M [\boldsymbol\eta(c)]_ig_i(t)}\nonumber\\
    &+\frac{c^2\bigg(\sum_{i=1}^M [\boldsymbol\eta(c)]_ig_i(t)-f_1(t)\bigg)\bigg(\sum_{i=1}^M [\boldsymbol\eta(c)]_ig_i(t)\bigg) f_2(t) \sigma_h }{\bigg(c\sum_{i=1}^M [\boldsymbol\eta(c)]_ig_i(t)\bigg)\bigg(c\sum_{i=1}^M [\boldsymbol\eta(c)]_ig_i(t)+c^2\bigg)}+o(\sigma_h) dt\bigg]\nonumber\\
    =&-\mathbb E\bigg[\int_0^T \frac{\Big(\sum_{i=1}^M [\boldsymbol\eta(c)]_ig_i(t)-f_1(t)\Big) f_2(t)}{\sum_{i=1}^M [\boldsymbol\eta(c)]_ig_i(t)+c} \sigma_h +o(\sigma_h) dt\bigg].
\end{align*}

As we can see, this gradient is off-zero by \begin{equation}\label{horder2}
    -\mathbb E\bigg[\int_0^T \frac{\Big(\sum_{i=1}^M [\tilde{\boldsymbol\eta}_c]_ig_i(t)-f_1(t)\Big) f_2(t)}{\sum_{i=1}^M [\tilde{\boldsymbol\eta}_c]_ig_i(t)+c}dt\bigg]\sigma_h + o(\sigma_h)
\end{equation}
so we expect the same order of error between $\hat\theta$ and $c$. Now we are ready to prove Proposition \ref{prop2} and \ref{prop3}.

\subsection{Proof of Proposition \ref{prop2}.}
We use the alternative intensity notation \eqref{special}.
\begin{proof}[Proof of Proposition \ref{prop2}]
    When using the naive Hawkes, we only have $g(t)=1$ the constant function, then we can solve for $\tilde\eta_c$, which is simply solving for the $\tilde\eta_c$
\begin{equation*}
 0 = \mathbb E\bigg[\frac{\tilde\eta_c-f_1}{\tilde\eta_c}\bigg]
\end{equation*}
which is simply the baseline intensity mean $\tilde\eta_c = \mathbb E[f_1].$

The gradient is thus off by zero by
\begin{equation*}
    \mathbb E\bigg[\frac{(f_1-\mathbb Ef_1)f_2}{\mathbb E f_1 +c}\bigg]\sigma_h + o(\sigma_h) = \frac{\text{Cov}(f_1,f_2)}{\mathbb E f_1 +c}\sigma_h + o(\sigma_h)
\end{equation*}

The rest follows from Assumption \ref{ass} as in proof of Proposition \ref{prop1}. 
\end{proof}

\subsection{Proof of Proposition \ref{prop3}.}
We use the alternative intensity notation \eqref{special}. In general case, we first assume all $f_1(t)$, $f_2(t)$ and $g_i(t)$ for $i>1$ with $g_1=1$ are stationary and square integrable (so we drop the dependence on $t$ and make this a subspace Hilbert space). For basis normalization, let us fix $g_1 =1$ and all other $\|g_i\|_2 := \mathbb E[g_i^2] =1$. More precisely, let $ G\subset L^2 (\Omega,\mathcal F,\mathbb P) $ be a centered Hilbert Space, we assume all $\{g_i\}_{i>1}\subseteq G$, as well as $f_1-\mathbb Ef_1 \in G$ and  $f_2-\mathbb Ef_2 \in G$. Furthermore, we assume the basis $g_i$ are uncorrelated (orthogonal basis): $\langle g_i,g_j\rangle:= \mathbb E[g_ig_j]$.

\begin{proof}[Proof of Proposition \ref{prop3}]

Recall Lemma \ref{ratioe}, we first verify 
\begin{align*}
    [\tilde{\boldsymbol \eta}_c]_1 =& \mathbb E[f_1]\nonumber\\
    [\tilde{\boldsymbol \eta}_c]_i =& \mathbb E[ (f_1-\mathbb E[f_1])g_i] \text { for $i\geq 1$}
\end{align*}
is a solution for the high order approximation:
\begin{equation*}
    R(g_i(\sum_{i=1}^M[\tilde{\boldsymbol\eta}_c]_ig_i-f_1), \sum_{i=1}^M[\tilde{\boldsymbol\eta}_c]_ig_i) = 0
\end{equation*}

To verify, first notice that 
\begin{equation*}
    \mathbb E[f_1|\mathcal G] = \sum_{i=1}^M[\tilde{\boldsymbol\eta}_c]_ig_i
\end{equation*}
where $\mathcal G$ is the sigma-algebra generated by $g_1,g_2,...,g_M$.
This can be easily checked by notice that $\langle\sum_{i=1}^M[\tilde{\boldsymbol\eta}_c]_ig_i-f_1,g_i\rangle = 0$ (projection) for all $i$. Thus,
\begin{equation*}
    \mathbb E[g_i(\sum_{i=1}^M[\tilde{\boldsymbol\eta}_c]_ig_i-f_1)] =\mathbb E[g_i\mathbb E[\sum_{i=1}^M[\tilde{\boldsymbol\eta}_c]_ig_i-f_1|\mathcal G]]=0
\end{equation*}
Then, 
\begin{align*}
    \text{Cov}(g_i(\sum_{i=1}^M[\tilde{\boldsymbol\eta}_c]_ig_i-f_1), \sum_{i=1}^M[\tilde{\boldsymbol\eta}_c]_ig_i) =& \mathbb E[g_i(\sum_{i=1}^M[\tilde{\boldsymbol\eta}_c]_ig_i-f_1)( \sum_{i=1}^M[\tilde{\boldsymbol\eta}_c]_ig_i)]\nonumber\\
    =&\mathbb E[g_i( \sum_{i=1}^M[\tilde{\boldsymbol\eta}_c]_ig_i)\mathbb E[\sum_{i=1}^M[\tilde{\boldsymbol\eta}_c]_ig_i-f_1|\mathcal G]]\nonumber\\
    =&0
\end{align*}
Then, we use Lemma \ref{ratioe} and \eqref{horder0}, \eqref{horder1},\eqref{horder2} to check
\begin{equation*}
     \mathbb E\bigg[g_i\bigg(1-\frac{f_1}{\sum_{i=1}^M [\tilde{\boldsymbol\eta}_c]_ig_i}\bigg)\bigg ] = R(g_i(\sum_{i=1}^M[\tilde{\boldsymbol\eta}_c]_ig_i-f_1),\sum_{i=1}^M[\tilde{\boldsymbol\eta}_c]_ig_i) + o(r^2) = o(r^2)
\end{equation*}

\begin{equation*}
     -\mathbb E\bigg[\int_0^T \frac{\Big(\sum_{i=1}^M [\tilde{\boldsymbol\eta}_c]_ig_i(t)-f_1(t)\Big) f_2(t)}{\sum_{i=1}^M [\tilde{\boldsymbol\eta}_c]_ig_i(t)+c}dt\bigg]\sigma_h =R((\sum_{i=1}^M[\tilde{\boldsymbol\eta}_c]_ig_i-f_1)f_2,\sum_{i=1}^M[\tilde{\boldsymbol\eta}_c]_ig_i+c) + o(r^2) = o(r^2)
\end{equation*}

The rest follows from $o(r^2)+o(\sigma_h)=o(r^2+\sigma_h)$, Assumption \ref{ass} and the setting have $M=2$ and $g_1=1$ and $g_2 = \frac{ f_2-\mathbb E[f_2]}{\sqrt{\text{Var}(f_2)}}$.

\end{proof}

\section{Simulation study and empirical verification } \label{appendix:simulatio_empirical}
This section presents a detailed simulation study and empirical verification of our model.

\subsection{Sinuoid background } \label{subsec:sinusoid}
\paragraph{Experiment settings}
The sinusoid function provides a simple way to control the correlation between the background activity of $f_i$ and $f_j$.
Consider the dynamic background as follows,
\begin{align}
f_i(t) =& A \sin(2\pi (t - \phi_{\mathrm{rnd}})) \\
f_j(t) =& A \sin(2\pi (t -  \phi_{\mathrm{rnd}} - \phi_{\mathrm{lag}})) 
\end{align}
where $A$ is the amplitude,
the length of a trial is $T=5$ second,
$\phi_{\mathrm{rnd}}\sim \mathrm{Uniform(0, 1)}$ varies from trial to trial so the background is not repeatedly observed,
$\phi_{\mathrm{lag}}$ controls the correlation between background signals, which is measured as the normalized dot product
\begin{equation}
\langle f_i, f_j \rangle
:= \frac{1}{T A^2} \int_0^T f_i(s) f_j(s) \mathrm{d}s
\end{equation}
When $\phi_{\mathrm{lag}}= 0, 0.5$, it achieves the largest positive or negative correlation respectively;
when $\phi_{\mathrm{lag}}= 0.25$, the correlation is zero.
$A=5$ spikes/second, $\alpha_i=\alpha_j=30$ spikes/second. 
Each simulation includes 200 trials.
$h_{i\to j}(t) = 2 \mathbb{I}_{[0, \sigma_h]}(t)$, $\sigma_h = 30$ ms.
$h_{j\to i}(t) = 0$.
The error band in the main Figure \ref{fig:bias_comparison} is obtained by repeating the simulation 100 times.

\paragraph{Neural Hawkes baseline model }
The Neural Hawkes is a representative deep learning framework for Hawkes processes \citep{mei2017neural}. The key components are
\begin{gather}
\boldsymbol{c}(t) = 
\bar{\boldsymbol{c}}_{i+1} + 
(\boldsymbol{c}_{i+1} - \bar{\boldsymbol{c}}_{i+1}) 
\exp(-\boldsymbol{\delta}_{i+1}(t-t_i^{\mathrm{source}})) \\
\boldsymbol{h}(t) = \boldsymbol{o}_i \odot \mathrm{tanh}(\boldsymbol{c}(t)) \\
\lambda_{\mathrm{target}}
= \left(\boldsymbol{W}_{\mathrm{target}}^T \boldsymbol{h} \right)_+
\end{gather}
To make a fair comparison with other models, we replace the decay function with 
\begin{equation}
\boldsymbol{c}(t) = 
\bar{\boldsymbol{c}}_{i+1} + 
(\boldsymbol{c}_{i+1} - \bar{\boldsymbol{c}}_{i+1}) 
\mathbb{I}_{[0,\sigma_h]}(t-t_i^{\mathrm{source}})
\end{equation}
The model takes the superimposed time points as the input.
At every step, the interval $t^{\mathrm{target}} - t^{\mathrm{source}}$  is passed to the model with labels of the source node and the target node.
In the standard Hawkes model, the impact of the history points to the intensity of the target is $\sum_{t_m \in N_{\mathrm{source}}, t_m < t} h_{\mathrm{source} \to \mathrm{target}}(t - t_m)$.
In contrast, Neural Hawkes considers
\begin{equation}
\boldsymbol{W}_{\mathrm{target}}^T
\Big[\boldsymbol{o}_i \odot \mathrm{tanh}
\Big(\bar{\boldsymbol{c}}_{i+1} + 
(\boldsymbol{c}_{i+1} - \bar{\boldsymbol{c}}_{i+1}) h_{\mathrm{source} \to \mathrm{target}}(t - t_m) \Big)\Big]  
\end{equation}
Although the model only takes one last time point $t_m$ at each step, it can still incorporate the history (in a non-linear way) using intermediate variables 
$\boldsymbol{o}_i, \boldsymbol{c}_{i+1}, \bar{\boldsymbol{c}}_{i+1}$ carried by recurrent neural network (RNN). 
The process label of $t_m$ is passed to the RNN through embedding, so the model can distinguish the source process.
$\boldsymbol{W}_{\mathrm{target}}$ varies from target to target, so the same source process can have different influences on different target processes.

After replacing the coupling window function from the default exponential function with the square window function,
the amplitude of the coupling effect between a source node and a target node is assessed as
$\boldsymbol{W}_{\mathrm{target}}^T
\Big[\boldsymbol{o}_i \odot \mathrm{tanh}
(\boldsymbol{c}_{i+1} - \bar{\boldsymbol{c}}_{i+1}) \Big]$
which contributes to the \textit{increment} of the intensity.
Since the output function is ReLU or softmax, which is close to the identity function on $\mathbb{R}_+$, the amplitude does not need to be rescaled.
Our implemented baseline model is in the public repository.

\subsection{Second-order stationary background } \label{appendix:sim_linear_cox}
In this section, we further study the behavior of the model in  \eqref{eq:target_likelihood} through simulations.
The simulation setup in this section is the foundation of the following sections.
In this special setup, we can approximate the bias, variance, risk, and likelihood of the estimator and they match the numerical results very well.
Details and all derivations will be shown in Appendix \ref{appendix:theoretical_derivations}.
Assume the fluctuating background activity $f_{i}$ and $f_{j}$ are second-order stationary stochastic processes,
meaning $\mathbb{E}[f_{i}(t) f_{i}(t +u)]$ only depends on $u$ but not $t$.
A special case of the second-order stationary process is the \textit{cluster point process} or \textit{linear Cox process}, which is widely used in point process study \cite{diggle1985kernel, bartlett1964spectral} and \cite[sec. 6.3]{daley2003introduction}.
We add the second-order stationary condition only to make theoretical derivations easier.
More variant simulation scenarios will be shown later.

We first generate random background $f_{i}, f_{j}$, then generate point processes. 
Let $\phi_{\sigma_I}(\tau) = \frac{1}{\sqrt{2\pi \sigma_I^2}}\exp(-\frac{\tau^2}{2\sigma_I^2})$ be a Gaussian window function with scale $\sigma_I$. $t^c_i$ are the time points of the center process determining the positions of Gaussian windows, which is generated by a homogeneous Poisson process with intensity $\rho$.
\begin{equation} \label{eq:linear_cox}
f_{i}(t) = f_{j}(t) = \sum_{i} \phi_{\sigma_I} \left(t - t^c_i\right)
\end{equation}
For simplicity, we first consider the impact function in the form
\begin{equation}
h_{i\to j}(t) = \alpha_{i\to j} \cdot \mathbb{I}_{[0,\sigma_h]}(t)
\label{eq:square_coupling_filter}
\end{equation}
where $\alpha_{i\to j}$ is the amplitude, and the filter length is $\sigma_h$.
$\sigma_I$ controls the timescale of the background activity. If $\sigma_I$ is smaller, then $f_{i}$ changes faster.
$\sigma_h$ controls the timescale of the point-to-point coupling effect. If $\sigma_h$ is smaller, then neuron $i$ influences neuron $j$ in a shorter time range.
The impact function estimator has form 
$\hat h_{i\to j} = \hat\alpha_{i\to j} \cdot \mathbb{I}_{[0,\sigma_h]}(t)$ with just one parameter $\hat\alpha_{i\to j}$ and the timescale $\sigma_h$ is known.
We use the thinning method to generate continuous-time point processes \cite{ogata1988statistical}.

\begin{figure}[H]
\centering
\includegraphics[width=0.92\linewidth]{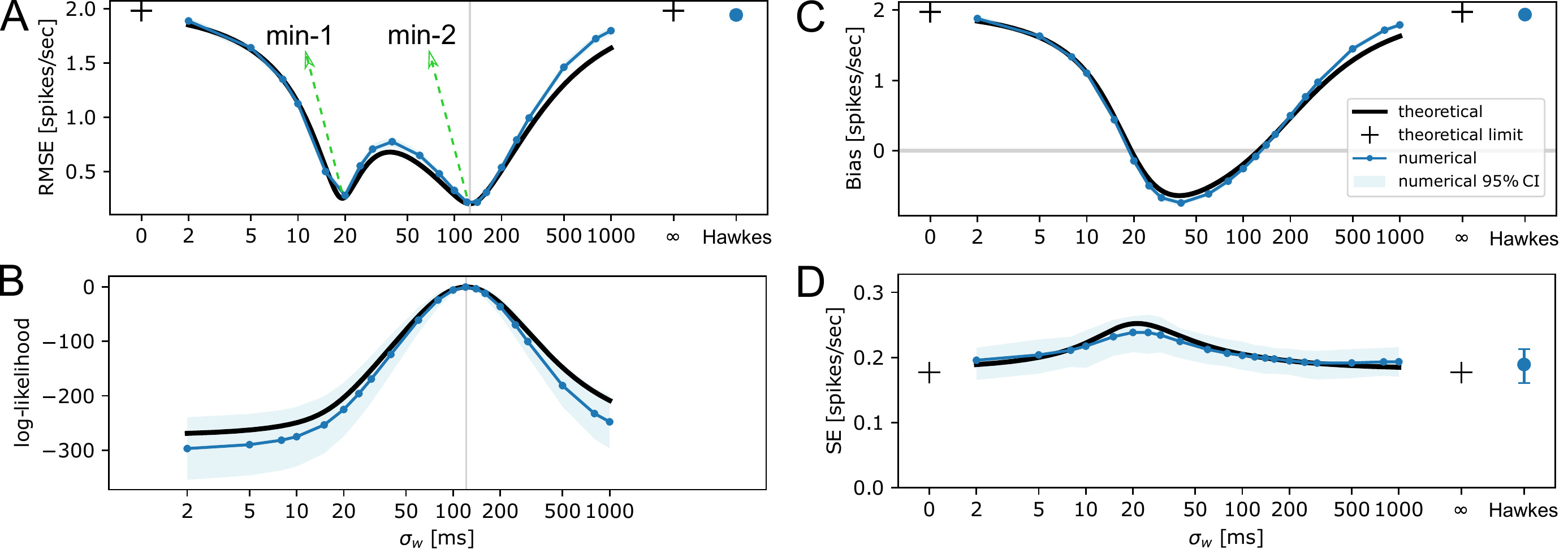}
\caption{Numerical and analytical properties of the estimator $\hat\alpha_{i\to j}$.
We show the properties of $\hat\alpha_{i\to j}$ as a function of smoothing kernel scale $\sigma_w$ of $W$. 
For numerical cases, we evaluate the properties at different $\sigma_w$ indicated by the blue dots.
The x-axis is in logarithmic scale.
The numerical (blue curves) and theoretical results (dark curves) are very close.
The pointwise confidence interval for RMSE and SE is calculated using bootstrap (bootstrap the replicated estimators, not the point processes).
The pointwise confidence interval for bias is calculated based on standard deviation. 
The blue band for the likelihood is 1.96$\times$standard deviation.
\textbf{A} The estimated risk root mean square error (RMSE) of the estimator $\hat\alpha_{i\to j}$.
The two local minimums are labeled by "min-1" and "min-2". Our method prefers to select "min-2" indicated by the vertical line.
The RMSE can be decomposed into bias (shown in \textbf{C}) and standard error (shown in \textbf{D}).
\textbf{B} The maximum log-likelihood as function of $\sigma_w$. Since the likelihood functions may have different offsets, we align them by the peak (the maximum value across $\sigma_w$) to zero, then calculate the mean and pointwise standard deviation.
The vertical line indicates the peak (numerical and theoretical peaks overlap), which matches the position of "min-2" in A.
The theoretical extreme cases "0" and "$\infty$" mean the scale $\sigma_w$ of smoothing window $W$ goes to limit $0$ or $\infty$.
The numerical case "no nuisance" represents the model without including the nuisance regressor $\overline{\textbf{s}_i}$, which becomes a typical Hawkes process model ignoring the fluctuating background activity.
}
\label{fig:estimator_properties}
\end{figure}

Figure \ref{fig:estimator_properties} shows the numerical and analytical approximation of the properties of the estimator.
The analytical formula is derived based on the second-order stationarity condition of the background activity, through which we would hope to provide insights into how the timescale of the activity is linked to the behaviors of the estimator.
The activity $f_{i}$ in the true model is set as a cluster process in \eqref{eq:linear_cox} with $\sigma_I = 100$ ms. 
The square window filter width is $\sigma_h= 30$ ms and $\alpha_{i\to j} = 2$ spikes/sec.
The firing rate of the center process $\rho=30$ spikes/sec.
The baselines are $\alpha_j = \alpha_i = 10$ spikes/sec.
One simulation case has 200 trials and the length of the trial is 5 sec. Each trial is assigned with an independently generated $f_{i}$.
Results in Figure \ref{fig:estimator_properties} are obtained through 100 repetitions.

If the estimation only considers a constant baseline, as known as the standard Hawkes model, without considering the fluctuating background signal, the estimated impact function will be positively biased (Figure \ref{fig:estimator_properties}C), as 
the model struggles to distinguish the effect of mutual interaction from that of the correlated input between neurons.
This explains why the estimated filter in the main Figure \ref{fig:sim_demo} is larger than the true filter.

Our model uses a smoothing kernel to eliminate the background artifacts as in main \eqref{eq:mean_coarsen_regressor}.
When the background smoothing kernel width is too wide $\sigma_w\to\infty$ or too narrow $\sigma_w\to 0$,
the nuisance variable is not able to capture any background activity, and the performance is as bad as the standard Hawkes with large positive bias (Figure \ref{fig:estimator_properties} with labels "$\infty$" points and "0" points).
The bias becomes negative between $\sigma_w = 20$ ms and $\sigma_w = 125$ ms.
The standard error in Figure \ref{fig:estimator_properties}D does not change too much as $\sigma_w$ changes.

The estimated risk of the estimator has two local minimums, labeled "min-1" and "min-2" in the figure.
The MLE points at "min-2", which is indicated by the vertical line in Figure \ref{fig:estimator_properties}A, B. 
The slope of the risk curve around "min-2" is smaller than the slope near "min-1" (the x-axis of the figure is in logarithmic scale), 
so the model is relatively less sensitive to the estimation or selection of $\sigma_w$ near "min-2".
As will be shown shortly, the position of "min-2" is related to the timescale $\sigma_I$ in $f_{i}$ and it is almost invariant of the impact function scale $\sigma_h$ or amplitude $\alpha_{i\to j}$.
The nuisance variable, the coarsened spike train $\overline{\textbf{s}_i}$ (as in main \eqref{eq:mean_coarsen_regressor}), can be interpreted as an approximation of the background activity, and $\sigma_w$ reflects the timescale of the background.

\subsection{Influences of the timescales of coupling effect and background effect }
\label{appendix:timescales}

\begin{figure}[H]
\centering
\includegraphics[width=0.99\linewidth]{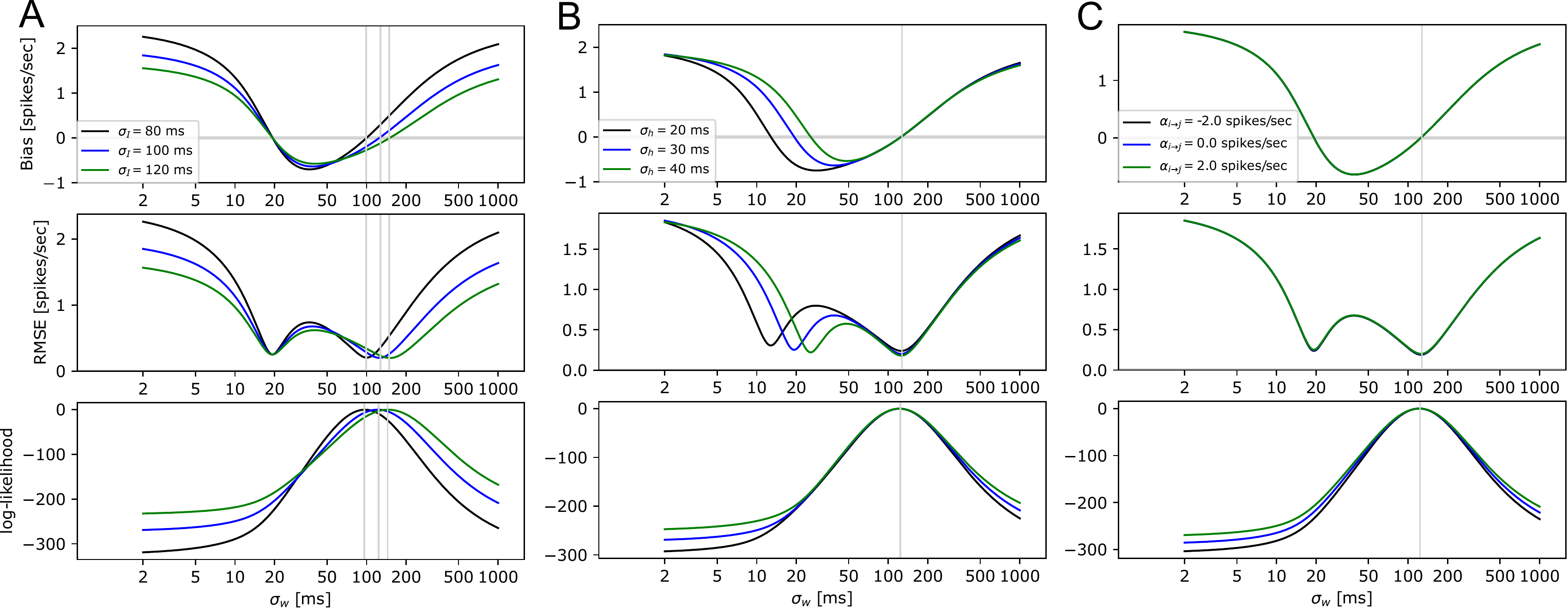}
\caption{Relations between the estimator's properties and background activity timescale, impact function timescale, and impact function amplitude.
We show the RMSE and log-likelihood curves as in Figure \ref{fig:estimator_properties}.
The settings are the same as Figure \ref{fig:estimator_properties}.
$\sigma_I$ is tuned in A, $\sigma_h$ is tuned in B, and $\alpha_{i\to j}$ is tuned in C.
This figure only shows the analytical results. Numerical results match the analytical formula very well, so data is not shown.
The log-likelihood functions may have different offsets, we align them by the peak to zero (maximum value across $\sigma_w$).
The vertical line indicates the MLE.
\textbf{A} $\sigma_I=80,100,120$ ms. 
$\sigma_h=30$ ms and $\alpha_{i\to j}=2$ spikes/sec are fixed.
\textbf{B} $\sigma_h=20,30,40$ ms. 
$\sigma_I=100$ ms and $\alpha_{i\to j}=2$ spikes/sec are fixed.
\textbf{C} $\alpha_{i\to j}=-2,0,2$ spikes/sec.
$\sigma_I=100$ ms and $\sigma_h=30$ ms are fixed.
}
\label{fig:property_tuning}
\end{figure}

As pointed out in main \eqref{eq:bias_formula} and derived in Appendix \ref{appendix:theoretical_derivations}, the properties of the estimator, such as bias and standard error, are related to the timescale of the background and the timescale of the coupling effect, but not the amplitude of the coupling effect.
These connections are presented in Figure \ref{fig:property_tuning}.
The figure shows the relations between the estimator's properties and the timescale of the background activity ($\sigma_I$ of $f_{i}, f_{j}$ in \eqref{eq:linear_cox}), the timescale of the point-to-point coupling activity ($\sigma_h$ of the impact function in \eqref{eq:square_coupling_filter}), and the amplitude of the impact function ($\alpha_{i\to j}$ in \eqref{eq:square_coupling_filter}). 
In Figure \ref{fig:varying_sigma_I_estimator_properties}A, the scale $\sigma_I$ of the background activity $f_{i}$ is related to the estimated smoothing kernel width $\sigma_w$.
If $\sigma_I$ is larger, the optimal $\sigma_w$ also becomes larger.
In Figure \ref{fig:varying_sigma_I_estimator_properties}B, the scale $\sigma_h$ of the impact function does not affect MLE too much, but it is related to the left root of the bias or the left local minimum of the risk.
In Figure \ref{fig:varying_sigma_I_estimator_properties}C, the amplitude of the impact function, whether it is positive or negative, does not change the bias or the RMSE, or the estimated $\sigma_w$.
These properties suggest a simpler and heuristic way of estimating $\sigma_w$. 
Unlike the jitter-based conditional inference method \citep{amarasingham2012conditional}, our method does not rely on the assumption restricting the timescale of the background being larger than the timescale of the coupling effect.
In Appendix \ref{appendix:small_sigma_I}, we will show a scenario with $\sigma_I < \sigma_h$, which violates the assumption of the condition inference, but our method still works well.

The optimal smoothing kernel width $\sigma_w$ is insensitive to the impact function's amplitude or timescale, which suggests a heuristic approximation for $\sigma_w$, meaning the range of the optimal kernel width $\sigma_w$ can be determined before estimating the impact function.
The variant of the model below without the impact function can be used for this purpose.
\begin{gather}
\min_{\beta_j, \beta_w, \sigma_w }
\left\{
-\sum_{s \in N_j } \log \tilde{\lambda}_{j}(s)
+ \int_0^T \tilde{\lambda}_{j}(s) \mathrm{d} s  \right\} \\
\tilde{\lambda}_{j}(t) :=
\beta_j
+ \beta_w \; \overline{\textbf{s}_{i}}(t) \\
\overline{\textbf{s}_{i}}(t)
= \int_0^T W(t-s) \mathrm{d} N_i(s)
\end{gather}

\subsection{Cross-connections and self-connections} \label{appendix:full_connection}
As a test of a more general scenario, this simulation considers full connections cross processes and self-connection within processes.
Simulation data is generated according to,
\begin{equation}
\begin{aligned}
\lambda_{j}(t)
&= \alpha_j
+ f_{j}(t) 
+ \int_0^t h_{i\to j}(t-s) \mathrm{d}N_{i}(s)
+ \int_0^t h_{j\to j}(t-s) \mathrm{d}N_{j}(s) \\
\lambda_{i}(t) 
&= \alpha_i
+ f_{i}(t)
+ \int_0^t h_{j\to i}(t-s) \mathrm{d} N_{j}(s)
+ \int_0^t h_{i\to i }(t-s) \mathrm{d} N_{i}(s)
\end{aligned}
\end{equation}

The fluctuating background follows the linear Cox process with the same settings as in Appendix \ref{appendix:sim_linear_cox}.
The only difference is that this scenario includes 4 impact functions $h_{i\to j}=-2, h_{j\to i}=-2, h_{i\to i}=1, h_{j\to j}=1$ spikes/sec.
The number of samples in each simulation is 200 5-second trials, and the number of repetitions is 100, the same as Appendix \ref{appendix:sim_linear_cox}.

\subsection{Varying-timescale background } \label{appendix:varying_sigma_I}
This section considers a variant of the scenario in Appendix \ref{appendix:sim_linear_cox} with not only fluctuating background but also with 
time-varying timescale $\sigma_I$.
The background activity $f_{i}$ in \eqref{eq:linear_cox} is composed of a sequence of Gaussian windows with fixed scale $\sigma_I$.
The locations of the windows are randomly determined through a homogeneous Poisson process with intensity $\rho$.
$\sigma_I$ controls how fast the background changes; if $\sigma_I$ is smaller, the activity will change faster.
$f_i$ is second-order stationary and some properties can be derived in closed-form formula.

Consider a similar process but the scale of the window $\sigma_I$ is no longer fixed,
\begin{equation} \label{eq:linear_cox_vary_sigma_I}
f_{i} = \sum_{i} \phi_{\sigma_{I,i}} \left(t - t^c_i\right)
\end{equation}
where $\sigma_{I,i}$ changes randomly; every time point of the center process $t^c_i$ is assigned with a different scale $\sigma_{I,i} \sim \mathrm{Uniform}(80, 140)$ ms.
The process $f_{i}$ changes faster at smaller $\sigma_{I,i}$, and changes slower at larger $\sigma_{I,i}$.
The rest of the experiment settings is the same as Appendix \ref{appendix:sim_linear_cox}.
The true impact function is a square window 
$h_{i\to j}(t) = \alpha_{i\to j} \cdot \mathbb{I}_{[0, \sigma_h]}(t)$,
where the timescale is $\sigma_h=30$ ms, the amplitude is $\alpha_{i\to j}=2$ spikes/sec.

\begin{figure}[H]
\centering
\begin{subfigure}{1\textwidth}
\centering
\includegraphics[width=1\linewidth]{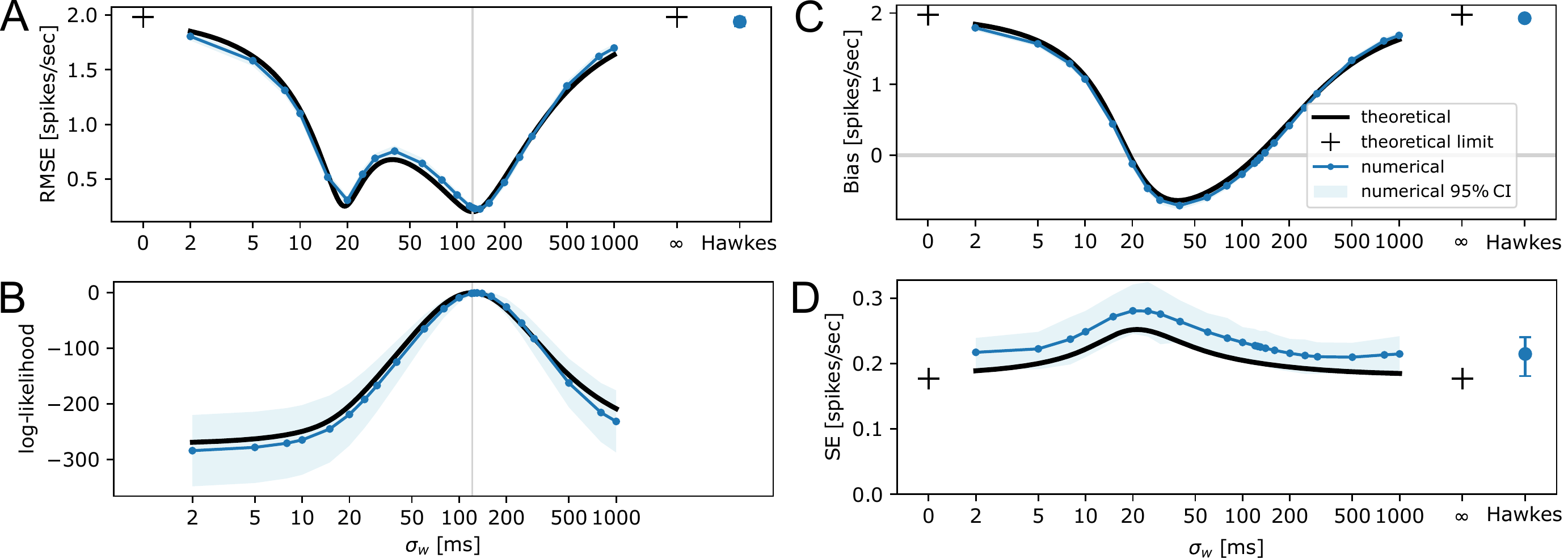}
\end{subfigure}
\caption{Simulation results of the impact function estimator with varying background activity timescale.
The figure is presented in the same ways as Figure \ref{fig:estimator_properties}. The simulation details are in the text.
The background activity $f_{i}$ in \eqref{eq:linear_cox} is replaced with \eqref{eq:linear_cox_vary_sigma_I} with varying timescale.
The results are similar to Figure \ref{fig:estimator_properties}.
The dark curves show the equivalent theoretical approximation using the model in main \ref{eq:target_likelihood} with fixed timescale $\sigma_I=100$ ms, which is manually tuned to match the numerical results.
}
\label{fig:varying_sigma_I_estimator_properties}
\end{figure}
As shown in Figure \ref{fig:varying_sigma_I_estimator_properties}, the selected kernel width $\sigma_w$ will balance the varying timescale, and it can still the select estimator with small risk and low bias, indicated by the vertical lines in Figure \ref{fig:varying_sigma_I_estimator_properties} A and B.
Similar to Figure \ref{fig:estimator_properties}D, the SE does not change too much as the smoothing kernel width $\sigma_w$ changes.
The model can balance the bias, which can be explained by its properties in Figure \ref{fig:estimator_properties}C.
If the timescale of the background $\sigma_I$ is fixed and consider the bias of the estimator near the right root. If $\sigma_w$ is larger than the right root, the bias will be positive, if $\sigma_w$ is a little smaller than the right root, the bias will become negative.
In this scenario, the timescale of the background $\sigma_I$ varies.
The optimal $\sigma_w$ is relatively large for the activity with small $\sigma_I$, so the bias is positive for the fast-changing part.
The optimal $\sigma_w$ is relatively small for the sessions with large $\sigma_I$, so the bias is negative for the slow-changing part of the activity.
With proper selection of $\sigma_w$, the estimator will balance the overall bias between fast- and slow-changing activities, and it can still achieve zero bias. Together with the SE, the risk properties remain similar in Figure \ref{fig:varying_sigma_I_estimator_properties} A.

To verify the reasoning, we compare the numerical results with the equivalent theoretical approximation shown in the dark curves in Figure \ref{fig:varying_sigma_I_estimator_properties}.
The theoretical method is for the model in main \ref{eq:target_likelihood} with fixed timescale for the background by manually tune the timescale as $\sigma_I=100$ ms to match the numerical curves.
The behavior of the estimator for the varying-timescale background activity is almost equivalent to the case with fixed-timescale background activity. The SE of the numerical results is slightly larger though.


\subsection{Fast-changing background } \label{appendix:small_sigma_I}
In extreme cases, the background activity $f_{i}$ can have fast-changing activities.
In this situation, conditional inference-based method can be limited by its formalization of the null hypothesis:
\begin{quote}
\textit{samples from the null distribution are generated by jittering the time points by a random amount, small enough to maintain the fluctuating background intensity, but big enough to break the time association pattern. }
\end{quote}
which implicitly assumes the timescale of the coupling effect is much smaller than the timescale of the background.
This simulation scenario is similar to the setup in section \ref{subsec:simulation_study}, except that $\sigma_I$ is set as a small value, which is comparable to or much smaller than the point-to-point interaction timescale $\sigma_h$.
We will show in this section, this is not a necessary assumption or constraint for our method.
Even the background changes faster than the coupling effect, our model can still have small error.

This simulation scenario is the same as Appendix \ref{appendix:sim_linear_cox} except that the timescale in \ref{eq:linear_cox} is set to $\sigma_I=20$ ms (Figure \ref{fig:small_sigma_I} A,B,C,D), and $\sigma_I=8$ ms (Figure \ref{fig:small_sigma_I} E,F,G,H).
The rest settings of the simulation scenarios are the same as the basic scenario in section \ref{appendix:sim_linear_cox},
where the true impact function is a square window 
$h_{i\to j}(t) = \alpha_{i\to j} \cdot \mathbb{I}_{[0, \sigma_h]}(t)$,
the timescale is $\sigma_h=30$ ms, the amplitude is $\alpha_{i\to j}=2$ spikes/sec.

\begin{figure}[H]
\centering
\begin{subfigure}{1\textwidth}
\centering
\includegraphics[width=1\linewidth]{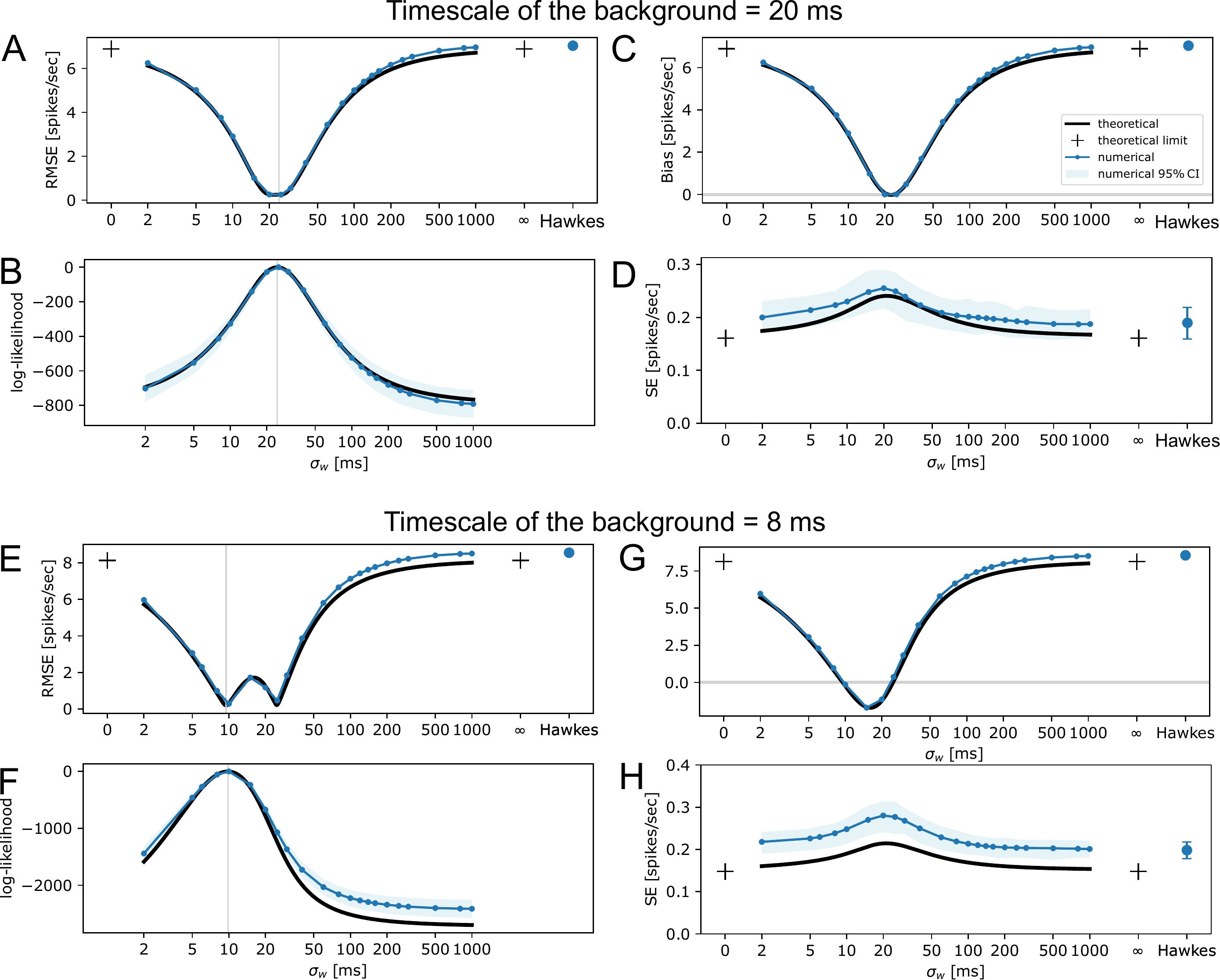}
\end{subfigure}
\caption{Fast-changing background with small $\sigma_I$.
The experiment is similar to Figure \ref{fig:estimator_properties} except that in
\textbf{A, B, C, D} $\sigma_I=20$ ms,
in \textbf{E, F, G, H} $\sigma_I=8$ ms.
}
\label{fig:small_sigma_I}
\end{figure}

As discussed in main section \ref{subsec:background_kernel_smoothing}, Appendix \ref{appendix:timescales}, Figure \ref{fig:estimator_properties},
and Figure \ref{fig:property_tuning}, the right root of the bias curve (or the right local minimum of the risk curve) is associated with background timescale $\sigma_I$ when $\sigma_I > \sigma_h$: if $\sigma_I$ decreases, the right root will shift toward the left.
Figure \ref{fig:small_sigma_I}A shows a special case if $\sigma_I$ keeps decreasing, two local minimums of the risk will overlap.
In Figure \ref{fig:small_sigma_I}C, two roots of the bias will merge to one. 
The property of SE does not change too much, see Figure \ref{fig:small_sigma_I}D and \ref{fig:estimator_properties}D.
If $\sigma_I$ keeps decreasing when $\sigma_I < 20$ ms, the local minimum of the risk or the root of the bias corresponding to the MLE will move on the left side, see Figure \ref{fig:small_sigma_I}E and F.
We also notice that our analytical approximation of the standard error in Figure \ref{fig:small_sigma_I}H begins to have a large error when $\sigma_I$ is very small.

\begin{figure}[H]
\centering
\begin{subfigure}{1\textwidth}
\centering
\includegraphics[width=1\linewidth]{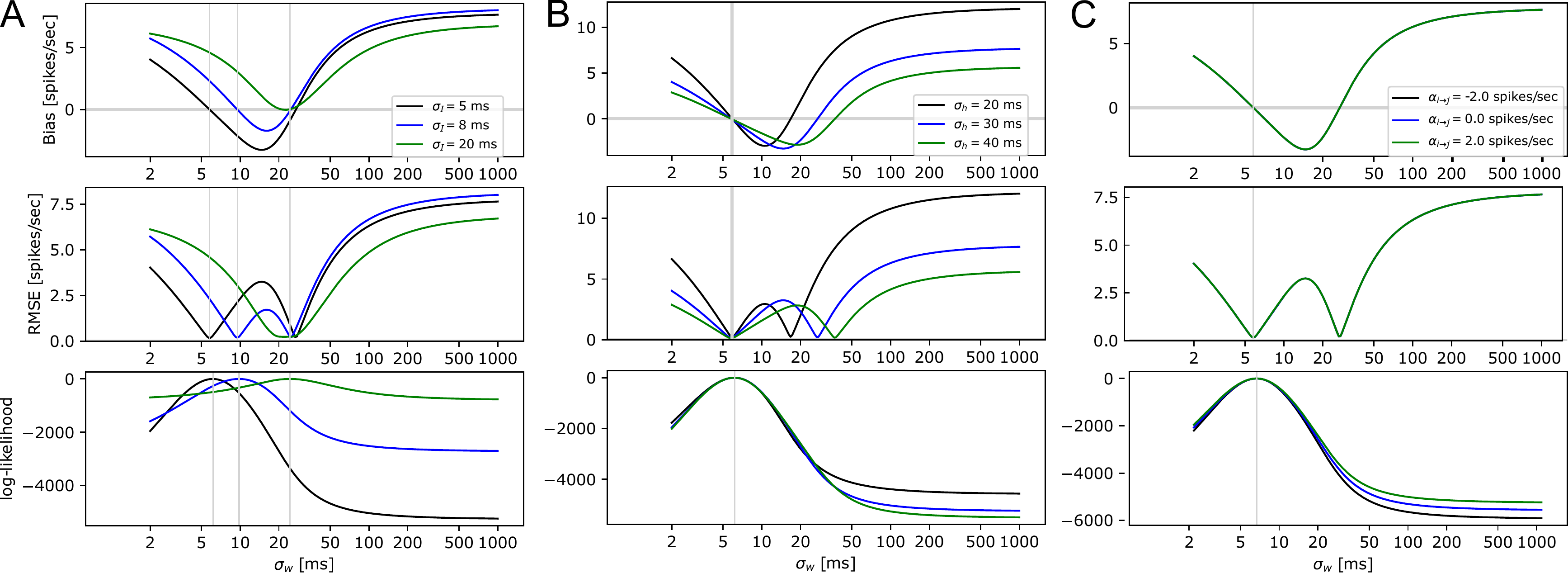}
\end{subfigure}
\caption{Properties of the estimator with fast-changing background.
This figure is analog to Figure \ref{fig:property_tuning} but the time scale $\sigma_I$ of the background activity $f_{i}$ is very small.
The settings are the same as Figure \ref{fig:estimator_properties}, \ref{fig:property_tuning}, and \ref{fig:small_sigma_I_tune_curves} except for different tuning parameters.
Only the analytical RMSE and log-likelihood curves are shown, which match the numerical results very well. Some numerical results have already been shown in Figure \ref{fig:small_sigma_I}.
The log-likelihood functions may have different offsets, we align them by the peak to zero (maximum value across $\sigma_w$).
The vertical lines indicate the MLE.
\textbf{A} Background timescale $\sigma_I$ as the tuning variable.
\textbf{B} Coupling effect timescale $\sigma_h$ as the tuning variable.
\textbf{C} Coupling effect amplitude $\alpha_{i\to j}$ as the tuning variable.
}
\label{fig:small_sigma_I_tune_curves}
\end{figure}

Similar to Figure \ref{fig:property_tuning}, next, we explore how the timescale of the background activity $\sigma_I$, the timescale of coupling effect $\sigma_h$, and the amplitude of impact function $\alpha_{i\to j}$ are related to the above properties when $\sigma_I$ is very small.
When $\sigma_I$ is tuned, it will change the optimal $\sigma_w$.
If $\sigma_I$ is around 20 ms, two local minimum values of the risk curve may merge to one, which agrees with the numerical result in Figure \ref{fig:small_sigma_I}A.
If $\sigma_I < 20$ ms, the root of the bias or the local minimum of the risk corresponding to the MLE will be on the left side.
When $\sigma_h$ is tuned, it will be the right local minimum of the risk or the right root of the bias that will be associated with the impact function timescale, that is opposite to the conclusion in Figure \ref{fig:property_tuning}B.
Similar to Appendix \ref{appendix:timescales}, the timescale $\sigma_h$ or the amplitude $\alpha_{i\to j}$ of the impact function do not affect $\sigma_w$ of the MLE.

\begin{figure}[H]
\centering
\begin{subfigure}{1\textwidth}
\centering
\includegraphics[width=1\linewidth]{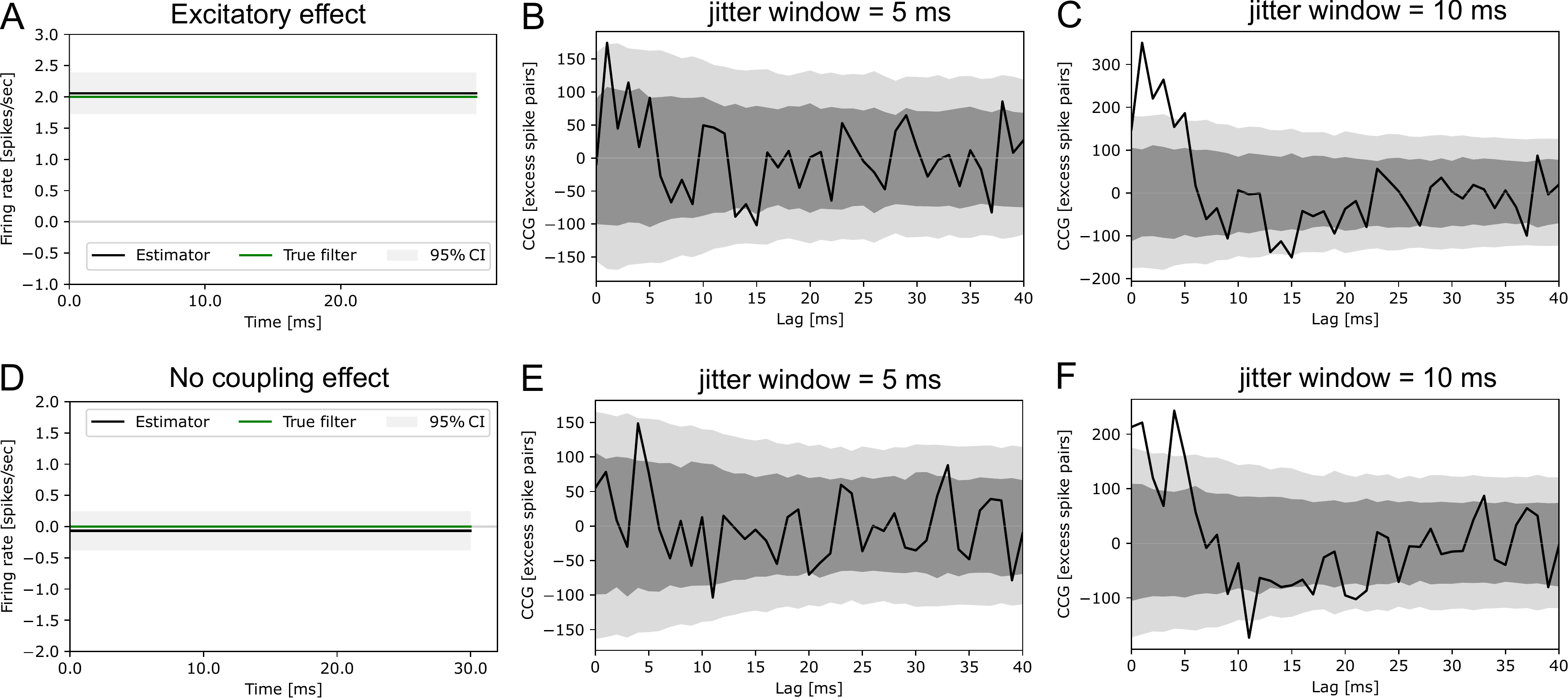}
\end{subfigure}
\caption{A comparison between the estimations of the coupling effect with fast-changing background.
The figure compares the performance of the estimators of the coupling effect. 
The simulation settings are the same as Figure \ref{fig:estimator_properties} except that the timescale of the background $\sigma_I=5$ ms is very small. The timescale of the impact function as in Eq \eqref{eq:square_coupling_filter} is $\sigma_h=30$ ms.
In A, B, and C, the amplitude of the true impact function is $\alpha_{i\to j}=2$ spikes/sec.
In D, E, and F, the amplitude of the true impact function is $\alpha_{i\to j}=0$ spikes/sec.
\textbf{A, D} The estimator of the point process regression. It can accurately estimate the true filter, which is supported by the analysis in Figure \ref{fig:small_sigma_I} and \ref{fig:small_sigma_I_tune_curves}.
\textbf{B, E} Jitter-based CCG. The time bin for the spike train is 1 ms. The jitter window width is set as 5 ms, close to the timescale of the background activity. The dark grey band is pointwise 95\% acceptance band, and the light grey band is simultaneous 95\% acceptance band. The result is acquired from 1000 surrogate jitter samples.
The CCG method detects a small excitatory effect before lag = 5 ms no matter whether the neurons have true coupling effect. 
\textbf{C, F} Similar to B except that the jitter window width is 10 ms.
In both B and C, the jitter-based CCG method can only detect a small effect before 5 ms lag or 7 ms lag. A large part of the coupling effect between 0 to 30 ms is buried under the CI band.
However, such an effect is due to the fast-changing background, but not the neuron-to-neuron coupling effect.
}
\label{fig:small_sigma_I_reg_vs_jitter}
\end{figure}
A significant advantage of our model over the jitter-based model is that the proposed model does not assume the background activity changes slower than the coupling effect, and the model can automatically find the optimal timescale.
The jitter-based method can not avoid such an assumption due to its nature of conditional inference.
The null hypothesis states that the coupling effects do not change faster than the jitter window width. Thus the samples under the null distribution are obtained by randomly jittering the points within the jitter window. If the background changes as fast as the coupling effect, such bootstrapping method can not maintain the temporal structure of the background activity, so it can not split the background artifacts and the coupling effects.
In other words, if the jitter window is set a little larger than the coupling effects, it can not tell whether the detected effect belongs to the background or the point-to-point interaction.
Some other bootstrapping methods have the same issue for exactly the same reason \cite{cowling1996bootstrap}.
Figure \ref{fig:small_sigma_I_reg_vs_jitter} compares the point process regression method and jitter-based CCG method.
The simulation scenario is the same as the basic model in Figure \ref{fig:estimator_properties} except that the timescale of the background activity is very small $\sigma_I=5$ ms. The numerical properties of the estimator have been shown in Figure \ref{fig:small_sigma_I}E-H. 
The true impact function is a square window 
$h_{i\to j}(t) = \alpha_{i\to j} \cdot \mathbb{I}_{[0, \sigma_h]}(t)$.
The timescale of the coupling effect is $\sigma_h=30$ ms.
The true amplitude of the impact function is $\alpha_{i\to j} = 2$ spikes/sec in Figure \ref{fig:small_sigma_I_reg_vs_jitter}A,B,C, and $\alpha_{i\to j} = 0$ spikes/sec in Figure \ref{fig:small_sigma_I_reg_vs_jitter}E,E,F.
In both cases, the regression method can accurately estimate the true estimator, which agrees with the numerical and theoretical results in Figure \ref{fig:small_sigma_I}.
In Figure \ref{fig:small_sigma_I_reg_vs_jitter} B and C, the CCG method with jitter window width = 5 or 10 ms can detect some excitatory effect in a lag range smaller than 5 ms or 10 ms. Nevertheless, it misses the excitatory effect between lag=10 to 30 ms. It is unreasonable to use a larger jitter window, as it will not match the background timescale.
The CCG results are similar to another example in Figure \ref{fig:small_sigma_I_reg_vs_jitter}E, F, where there is no coupling effect. The detected significant data points are totally due to the fast-changing background.
So for the results in Figure \ref{fig:small_sigma_I_reg_vs_jitter}B and C, we can not conclude that the jitter method has removed the background artifacts and the significant effect is caused by the coupling effect.

\subsection{Asymptotic Normality of the estimator } \label{appendix:normality}
In this section, we perform simulations to verify the asymptotic normality property of the estimator empirically.
The dataset is the same as Figure \ref{fig:estimator_properties}.
The true impact function is a square window
$h_{i\to j}(t) = \alpha_{i\to j} \cdot \mathbb{I}_{[0, \sigma_h]}(t)$,
where the amplitude is $\alpha_{i\to j}=2$ spikes/sec, and the timescale is $\sigma_h=30$ ms.
The estimator for the coupling effect $\hat\alpha_{i\to j}$ is \ref{eq:target_likelihood}.
We compare the empirical distribution of the estimators with the theoretical distribution.
As shown by Figure \ref{fig:normality}, the estimator has normal distribution at the optimal selection of $\sigma_w=120$ ms.
If $\sigma_w$ is too small or too large, the empirical distributions will have visible deviations.

\begin{figure}[H]
\centering
\begin{subfigure}{0.7\textwidth}
\centering
\includegraphics[width=1\linewidth]{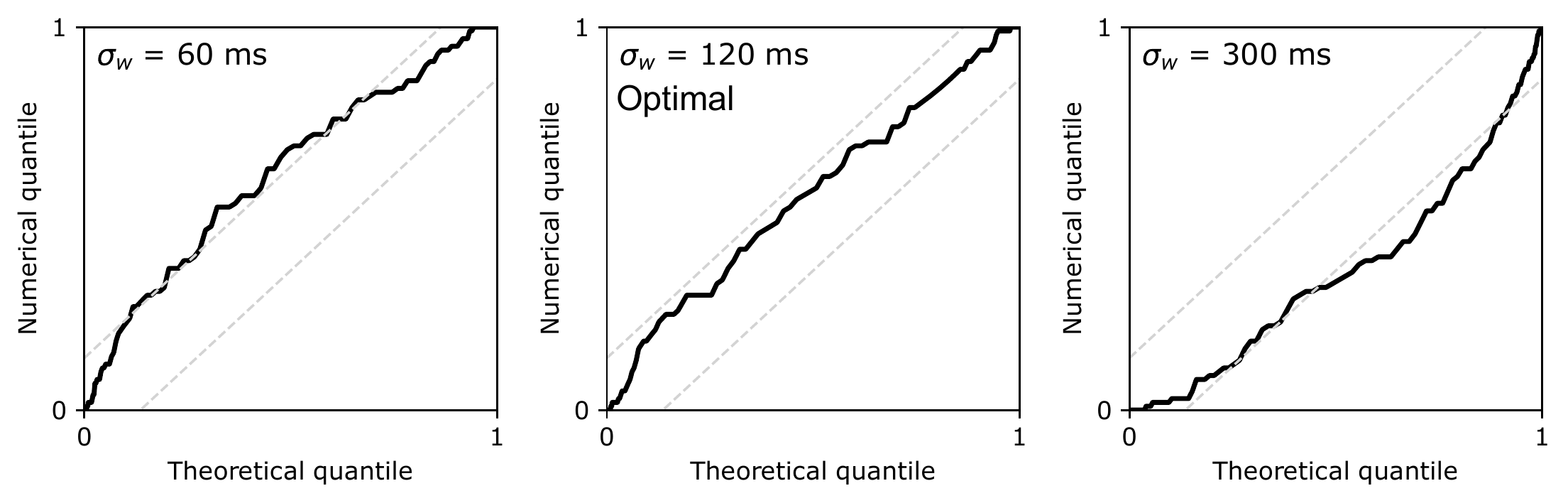}
\end{subfigure}
\caption{Normality of the estimator's distribution.
The dataset is the same as Figure \ref{fig:estimator_properties} including 100 repetitions.
The figure shows the Q-Q plots of the empirical distribution of the estimator against the theoretical normal distribution, shown in the dark curves. 
The straight dashed grey lines are 95\% uniform CI.
In the first row, the empirical distribution matches the theoretical distribution very well at the optimal model ($\sigma_w=120$ ms).
We also evaluate the model at other different smoothing kernel widths. 
If $\sigma_w$ is too small or too large $\sigma_w=60, 300$ ms), the empirical distribution will have a visible deviation from the theoretical distribution.
}
\label{fig:normality}
\end{figure}

\subsection{Selection of impact function length. }  \label{appendix:unmatch_filter_length}
The simulation scenario in the main text in Figure \ref{fig:estimator_properties} and many scenarios in the supplementary sections simplify the impact function estimation using a square window and assume the timescale of the coupling effect $\sigma_h$ in Eq \eqref{eq:square_coupling_filter} is known.
In this section, we show the consequences of unmatched impact function timescale. Because in practice, the timescale of the coupling effect is usually unknown.

We used the same dataset in Figure \ref{fig:estimator_properties}, where the true impact function is a square window. The amplitude $\alpha_{i\to j} = 2$ spikes/sec and the window width is $\sigma_h=30$ ms.
We applied two versions of the regression model to the dataset. Both versions used a square window as the impact function estimator, but one with shorter timescale $\sigma_{h,1} = 20$ ms, the other with longer timescale $\sigma_{h,2}=40$ ms. We present the results in Figure \ref{fig:unmatched_filter_length} in the same way as Figure \ref{fig:estimator_properties}.
\begin{figure}[H]
\centering
\begin{subfigure}{1\textwidth}
\centering
\includegraphics[width=1\linewidth]{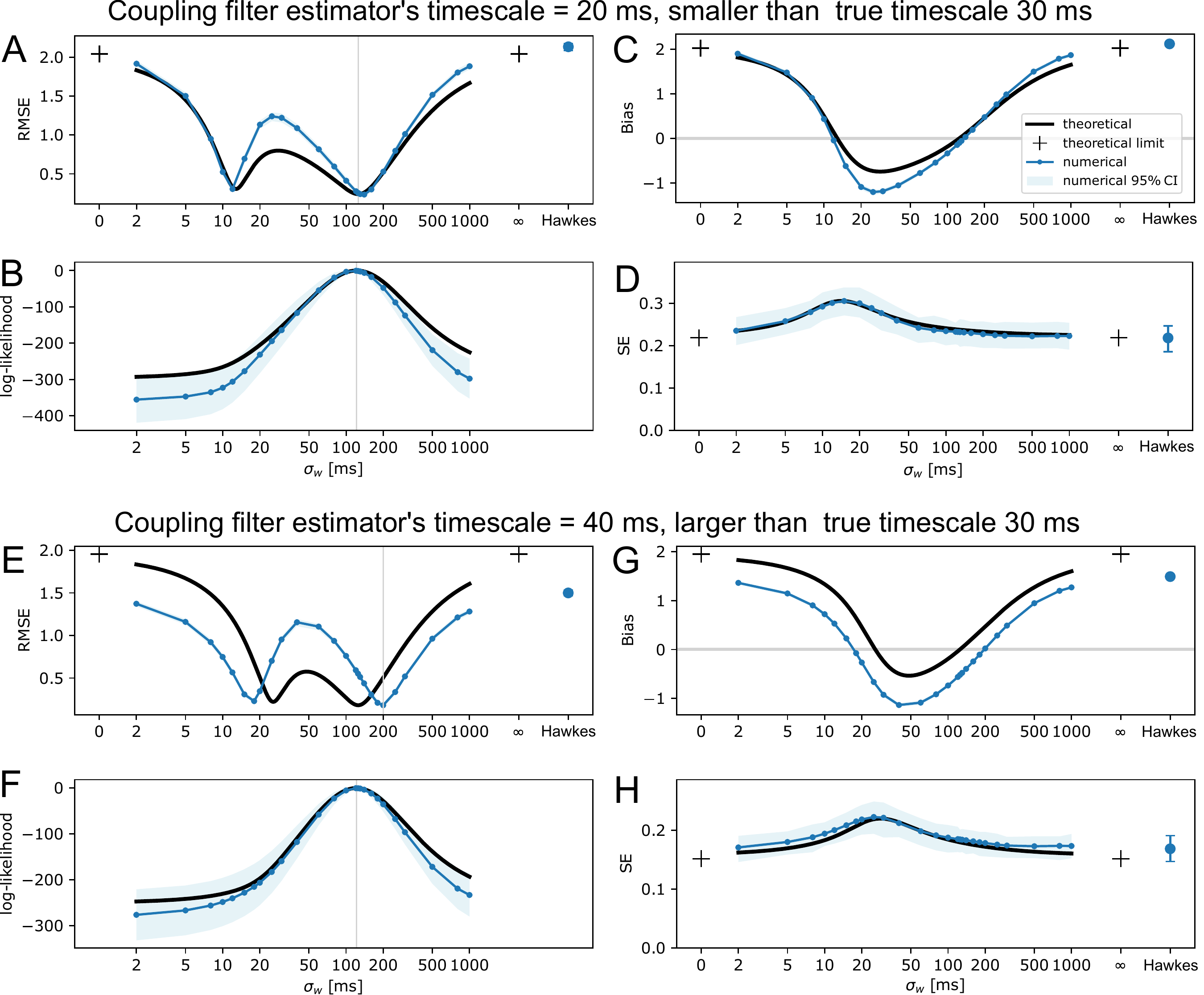}
\end{subfigure}
\caption{Consequences of unmatched impact function timescale.
The dataset is the same as Figure \ref{fig:estimator_properties}, where the true impact function is a square window. The amplitude is $\alpha_{i\to j} = 2$ spikes/sec and the window width is $\sigma_h=30$ ms. We tested the regression model with unmatched impact function width. 
The model in \textbf{A, B, C, D} estimates the impact function using a shorter timescale $\sigma_{h,1} = 20$ ms.
The model in \textbf{E, F, G, H} estimates the impact function using a longer timescale $\sigma_{h,1} = 40$ ms.
As a reference, the dark curves show the theoretical approximation using the basic regression model by setting the impact function timescale as 20 ms in A-D, and 40 ms in E-H.
The rest settings are the same as the simulation.
}
\label{fig:unmatched_filter_length}
\end{figure}
If the impact function is estimated using a shorter timescale ($\sigma_{h,1}=20$ ms) as shown in Figure \ref{fig:unmatched_filter_length}A,B,C,D, the selected model still has the minimum risk indicated by the vertical line.
The dark curves in Figure \ref{fig:unmatched_filter_length}A,B,C,D, show the theoretical approximation of the properties using the basic regression model in \eqref{eq:target_likelihood} by setting the impact function timescale as 20 ms instead, which can be seen as the expected properties of the model.
The absolute values of the bias are larger than expected if $\sigma_w$ is between 10 ms and 120 ms or larger than 200 ms. But the roots of the bias still match the expected position. The SE is not affected by the unmatched timescale. So the optimal selection of $\sigma_w$ does not change.
As a contrast, if the impact function timescale of the estimator ($\sigma_{h,2}=40$ ms) is longer than the truth (30 ms), the consequence is more severe. As shown in Figure \ref{fig:unmatched_filter_length} G, the actual bias is uniformly lower than the expected bias. The SE is not affected. So the consequence is that the selected smoothing kernel width $\sigma_w$ (Figure \ref{fig:unmatched_filter_length} F vertical line) does not match the actual risk minimum (Figure \ref{fig:unmatched_filter_length} E vertical line). 

By combining the results of the two cases, we recommend users select shorter impact function timescale if they are not confident about the impact function timescale, or using non-parametric fitting as in Supplementary \ref{appendix:nonparametric}.

\subsection{Multivariate regression and partial relation } \label{appendix:multivariate_regression_sim}

\begin{figure}[H]
\centering
\begin{subfigure}{0.35\textwidth}
\centering
\includegraphics[width=1\linewidth]{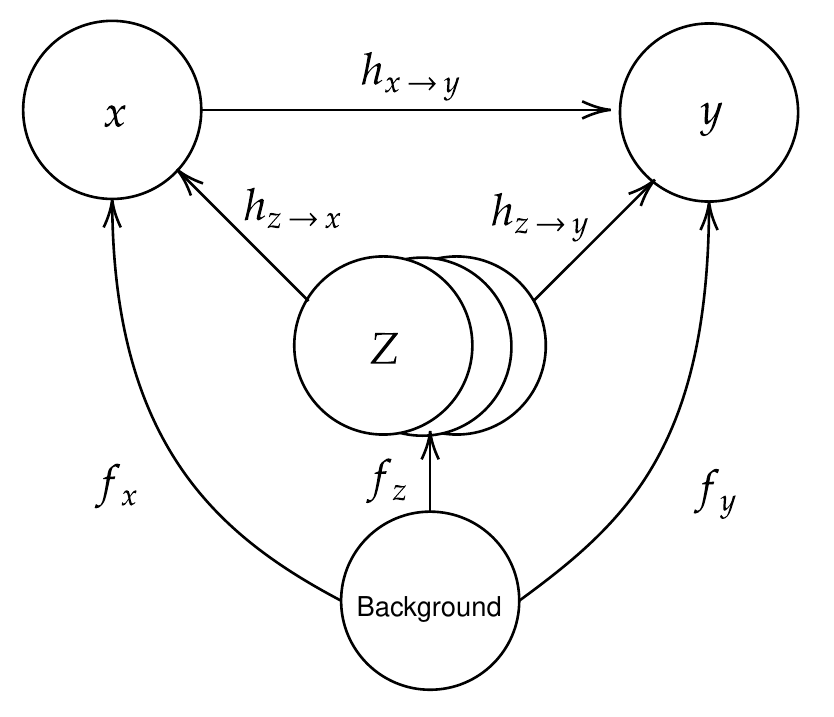}
\end{subfigure}
\caption{Diagram of multivariate point process network driven by background activity.
}
\label{fig:multi_regression}
\end{figure}

Multivariate regression is a natural extension of the basic regression model introduced in the main text.
The diagram is shown in Figure \ref{fig:multi_regression}.
\begin{equation}
\lambda_{y}(t | \mathcal{H}_{t})
= \alpha_y
+ \int_0^t h_{x\to y}(t-\tau) N_{x}(\mathrm{d} \tau)
\underbrace{+ f_{y}(t) 
+ \sum_{z \in Z} \int_0^t h_{z\to y}(t-\tau) N_{z}(\mathrm{d} \tau) 
}_{\tilde{f}_y(t) }
\end{equation}
where besides the interaction between node $x$ and $y$, they are both connected to other nodes denoted by $Z$.
The whole network is also driven by unobserved background activity.
This scenario is motivated by the challenge in practice: only part of the network can be observed with a limited number of nodes; 
Besides interaction across nodes, the network is also driven by other factors, usually not directly observed, see Figure \ref{fig:multi_regression}).
The input of node $y$ includes the coupling effect $h_{x\to y}$ from $x$, or $h_{z\to y}$ from other nodes $z\in Z$, and background influence $f_y$.
The sum input of $h_{z\to y}$ and $f_y$ can be seen as $\tilde{f}_y$.
Similarly, the total input for $x$ is denoted by $\tilde{f}_x$.
Our goal is to estimate $h_{x\to y}$ as the target relation conditioning on both background and $Z$, by properly handling the correlation between $\tilde{f}_x$ and $\tilde{f}_y$.
The multivariate regression problem is reduced to the pairwise bivariate regression problem as the main text.
This case is also inspired by \citep{chen2017nearly}, where authors proposed that the coupling effect in multivariate point process regression problems can be approximated by pairwise cross-correlation very well.
But their method assumes constant baselines and only positive impact function functions.
Next, we demonstrate using simulations to show our model is promising to overcome these limitations.

The simulation scenarios follow the diagrams in Figure \ref{fig:multi_regression}.
The process in $Z$ and $x, y$ are all driven by fluctuating background $f_x=f_y=f_z$ set as the linear Cox process as \ref{eq:linear_cox} in section \ref{appendix:sim_linear_cox},
where the intensity of the center process is $\rho=20$ spikes/sec, and the window function is Gaussian with scale $\sigma_I=100$ ms. 
The constant baseline of all processes is 10 spikes/sec.
The network includes 6 nodes, coupled with square window function $\alpha_{i\to j}\mathbb{I}_{[0,\sigma_h]}(t)$, $\sigma_h=30$ ms as known. The amplitude $\alpha_{i\to j}$ are positive, negative, or zero.
Each simulation has 200 trials with a 5-second duration.
The performance in the main Table \ref{tab:multivariate} is obtained from 100 repetitions.

\section{Some use cases of the model} \label{appendix:use_cases}

\subsection{Non-parametric fitting for the impact function } \label{appendix:nonparametric}
For simplicity, the models presented in the main text and many sections in the appendix use a square window for the impact function. 
This section considers non-parametric fitting for the impact function through splines.
The linear form of the intensity function in main \eqref{eq:target_likelihood} can be easily extended for this purpose, also see Appendix \ref{appendix:jitter_optimization_algo}.
The impact function now is estimated as a linear combination of spline bases as follows,

\begin{equation*}
h_{i\to j}(s) = \beta_{h,1} B_1(s) + ...+ \beta_{h,k} B_k(s)
\end{equation*}
where $B_1,...,B_k$ are spline bases. Define the covariates in the regression,
\begin{equation*}
\phi_{h,1}(t) := \int B_1(t-s) N_i(\mathrm{d}s),...,\;
\phi_{h,k}(t) := \int B_k(t-s) N_i(\mathrm{d}s)
\end{equation*}
The intensity function becomes,
\begin{equation*}
\tilde{\lambda}_{j}(t) =
\beta_j
+ \beta_w \overline{\textbf{s}_i}(t)
+ \beta_{h,1} \phi_{h,1}(t)
+...+ \beta_{h,k} \phi_{h,k}(t)
\end{equation*}
$\overline{\textbf{s}_i}(t)$ is same as the coarsened spike train in main \eqref{eq:mean_coarsen_regressor}.
The coefficients of the impact function $\beta_{h,1},...,\beta_{h,k}$ can still be estimated using the model in \ref{eq:target_likelihood}. The optimization algorithm is in Appendix \ref{appendix:jitter_optimization_algo}. We applied the non-parametric fitting to the dataset in Figure \ref{fig:estimator_properties}. 
The true impact function is a square window 
$h_{i\to j}(t) = \alpha_{i\to j} \cdot \mathbb{I}_{[0, \sigma_h]}(t)$
where the amplitude of the impact function is $\alpha_{i\to j}=2$ spikes/sec, the timescale of the square window is $\sigma_h=30$ ms. 
The impact function is estimated in a lag window between 0 and 50 ms using B-splines with 9 equal-distance knots. 
We evaluate the risk using root-mean-integral-square error (RMISE). The RMISE between the true impact function $h(t)$ and the estimator $\hat h(t)$ is defined as follows. $L_h = 50$ ms is the length of the impact function. 
\begin{equation*}
\mathrm{RMISE}(h, \hat h) := \sqrt{\frac{1}{L_h} \int_0^{L_h} 
    \big(\hat h(t) - h(t)\big)^2 \mathrm{d} t }
\end{equation*}
We evaluate the bias and the standard error of the filter at lag 5, 15, 25 ms. The result is shown in Figure \ref{fig:bspline_sim_results} below.

\begin{figure}[H]
\centering
\begin{subfigure}{1\textwidth}
\centering
\includegraphics[width=1\linewidth]{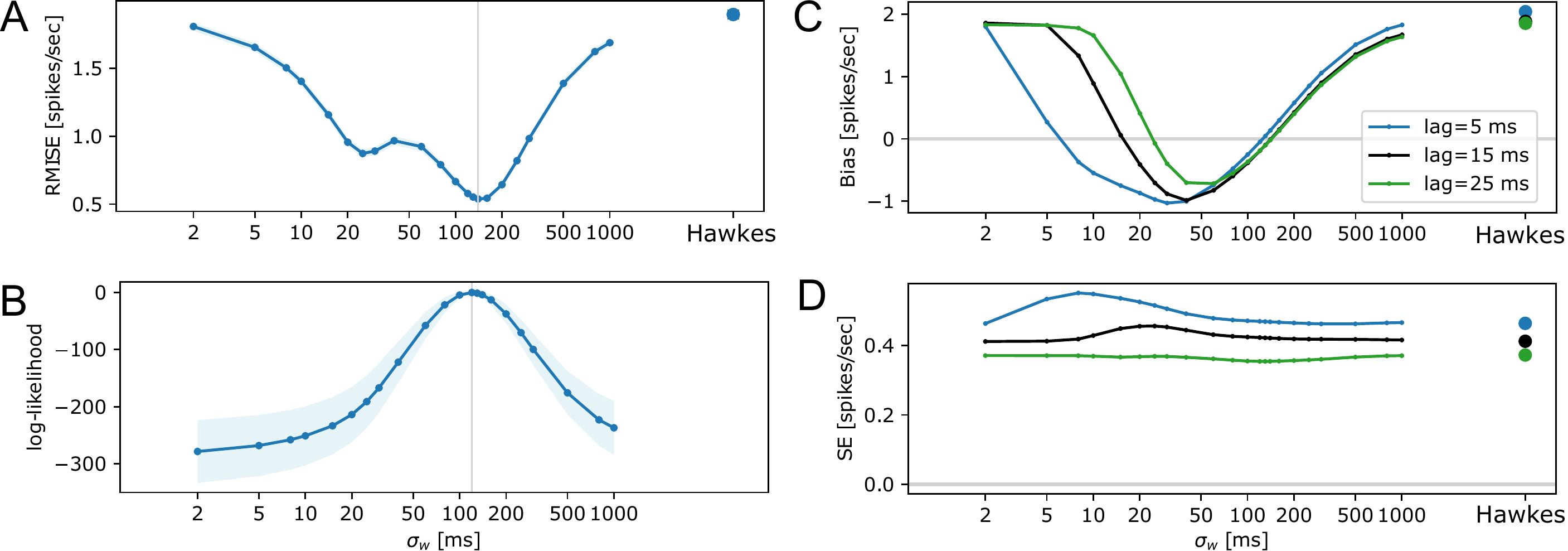}
\end{subfigure}
\caption{Non-parametric fitting for the impact function.
The dataset and the non-parametric estimator are described in the test.
The results are presented in the same way as Figure \ref{fig:estimator_properties}.
\textbf{A} RMISE of the estimated impact function as a function of smoothing kernel width $\sigma_w$ of $W$  in main \eqref{eq:mean_coarsen_regressor}. The vertical line indicates the minimum risk.
\textbf{B} The maximum log-likelihood as function of $\sigma_w$. Since the likelihood
functions may have different offsets, we align them by the peak (the maximum value across $\sigma_w$) to zero, then calculate the mean and pointwise standard deviation. The vertical line indicates the peak of the mean log-likelihood.
\textbf{C} The bias of the estimator is evaluated at lag = 5, 15, 25 ms.
\textbf{D} The standard error of the estimator is evaluated at $h(\mathrm{lag})$, lag = 5, 15, 25 ms.
}
\label{fig:bspline_sim_results}
\end{figure}
The risk curve and the log-likelihood curve are similar to the result in Figure \ref{fig:estimator_properties}. 
The optimal model with minimum risk can be selected by maximizing the likelihood, which is the same as the basic regression scenario in Figure \ref{fig:estimator_properties}.
The difference is that, in the non-parametric fitting, the left local minimum risk has a higher value than the right local minimum. While in the basic fitting case, two local minimum values of the risk curve are close in Figure \ref{fig:estimator_properties}).
This can be explained by decomposing the risk into bias and SE shown in Figure \ref{fig:bspline_sim_results}C and D.
If the smoothing kernel width $\sigma_w$ is around 130 ms, the bias values at different lags of the impact function are nearly the same.
But if $\sigma_w$ is around 20 ms, the bias values at different lags have large divergence: the beginning part of the estimator at lag=5 ms has a negative bias, the middle part at lag=15 ms has around zero bias, and the end part of the estimator at lag=25 ms has a positive bias.
The SE of the estimator at different lags does not change a lot as $\sigma_w$ varies.
So overall, the RMISE has a  much larger value near lag=20 ms than at lag=130 ms.
These properties are further demonstrated in Figure \ref{fig:bspline_coupling_filter}.

\begin{figure}[H]
\centering
\begin{subfigure}{1\textwidth}
\centering
\includegraphics[width=1\linewidth]{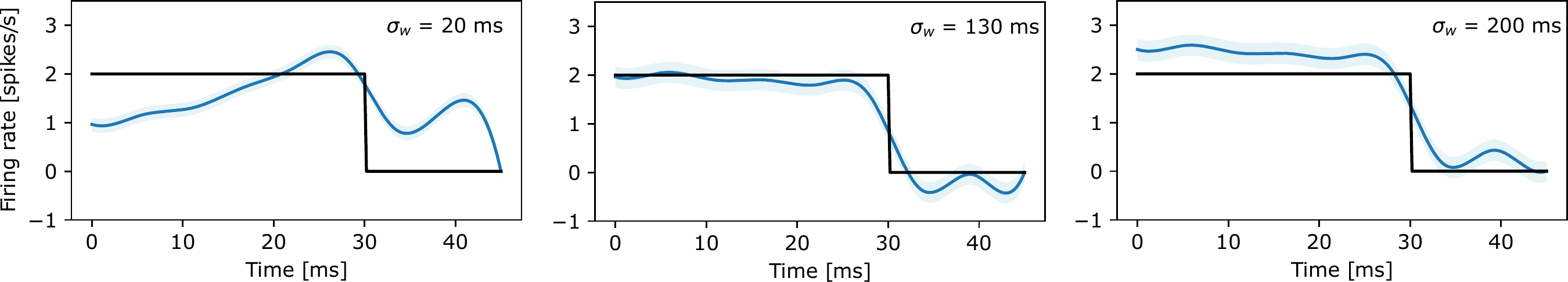}
\end{subfigure}
\caption{Non-parametric fitting for the impact function.
The figure compares the true impact function (dark) and the estimator (blue). 
The light blue band is pointwise 95\% CI.
The impact functions were fitted in the same way as described in Figure \ref{fig:bspline_sim_results}.
This figure picks out some fitted estimators with different smoothing kernel widths $\sigma_w=20, 130, 200$ ms. 
If the smoothing kernel width is too small ($\sigma_w=20$ ms), the bias values of the estimator at different lags have large differences. This matches the bias curves shown in Figure \ref{fig:bspline_sim_results}C. At the beginning part of the estimated impact function around lag=5 ms, the bias is negative, and at the end part around lag=25 ms, the bias is positive. 
If the smoothing kernel width is selected optimally ($\sigma_w=130$ ms), the fitted impact function matches the true filter very well.
If the smoothing kernel width is too wide ($\sigma_w=200$ ms), the whole estimated impact function has uniform positive bias at different lags. This agrees with Figure \ref{fig:bspline_sim_results}C that multiple bias curves with different lags beyond $\sigma_w=130$ ms are very close. 
}
\label{fig:bspline_coupling_filter}
\end{figure}

\subsection{Hypothesis testing example } \label{appendix:hypothesis_testing}
The regression model can be adopted for hypothesis testing problems.
The simulation scenario in this section is similar to the case in Figure \ref{fig:estimator_properties}.
The background activity is a cluster point process in \eqref{eq:linear_cox}, and the impact function is a square window
$h_{i\to j}(t) = \alpha_{i\to j}\cdot\mathbb{I}_{[0, \sigma_h]}(t)$ in \eqref{eq:square_coupling_filter}.
Each simulation dataset only has 10 trials. We reduce the sample size to make the tasks more difficult, and the performances of different estimators will be more distinguishable.
The length of the trial is 5 seconds, the time scale of the background activity is $\sigma_I = 100$ ms. The intensity of the center process is $\rho=30$ spikes/sec.
The baselines of two neurons are $\alpha_i = \alpha_j = 10$ spikes/sec.
The impact function is estimated using a square window with known timescale $\sigma_h = 30$ ms. 
The amplitude of the impact function is $\alpha_{i\to j} = 0$ spikes/sec in the null cases without coupling effects. 
We include three true positive scenarios with impact function amplitudes $\alpha_{i\to j} = 2, -2, 1$ spikes/sec respectively. 
The dataset generating and the model fitting procedure was repeated for 100 times.

Consider the null hypothesis 
\begin{equation*}
H_0:\; \hat\alpha_{i\to j} = 0    
\end{equation*}
$\hat\alpha_{i\to j}$ is the estimator for $\alpha_{i\to j}$.
The inference method is a direct application of the properties of the estimator.
The smooth kernel is $\sigma_w = 125$ ms, which is chosen by maximizing the likelihood.
$\hat\alpha_{i\to j}$ has asymptotic normal distribution (see details in Appendix \ref{appendix:normality} and Appendix \ref{appendix:theoretical_derivations}), so the p-value can be easily calculated accordingly.
The alternative method is jitter-based CCG, where the time bin width is 2 ms, the jitter window width is 100 ms. The CCG with shorter or longer jitter window width, for example 60 ms or 140 ms, gives similar results, so the figures are not shown. 
The p-value of the method is obtained by considering the multiple testing across all time lags between 0 and 30 ms, which is the same as the true impact function length. 
The calculation detail is in \cite{amarasingham2012conditional} supplementary document.

\begin{figure}[H]
\centering
\begin{subfigure}{0.8\textwidth}
\centering
\includegraphics[width=1\linewidth]{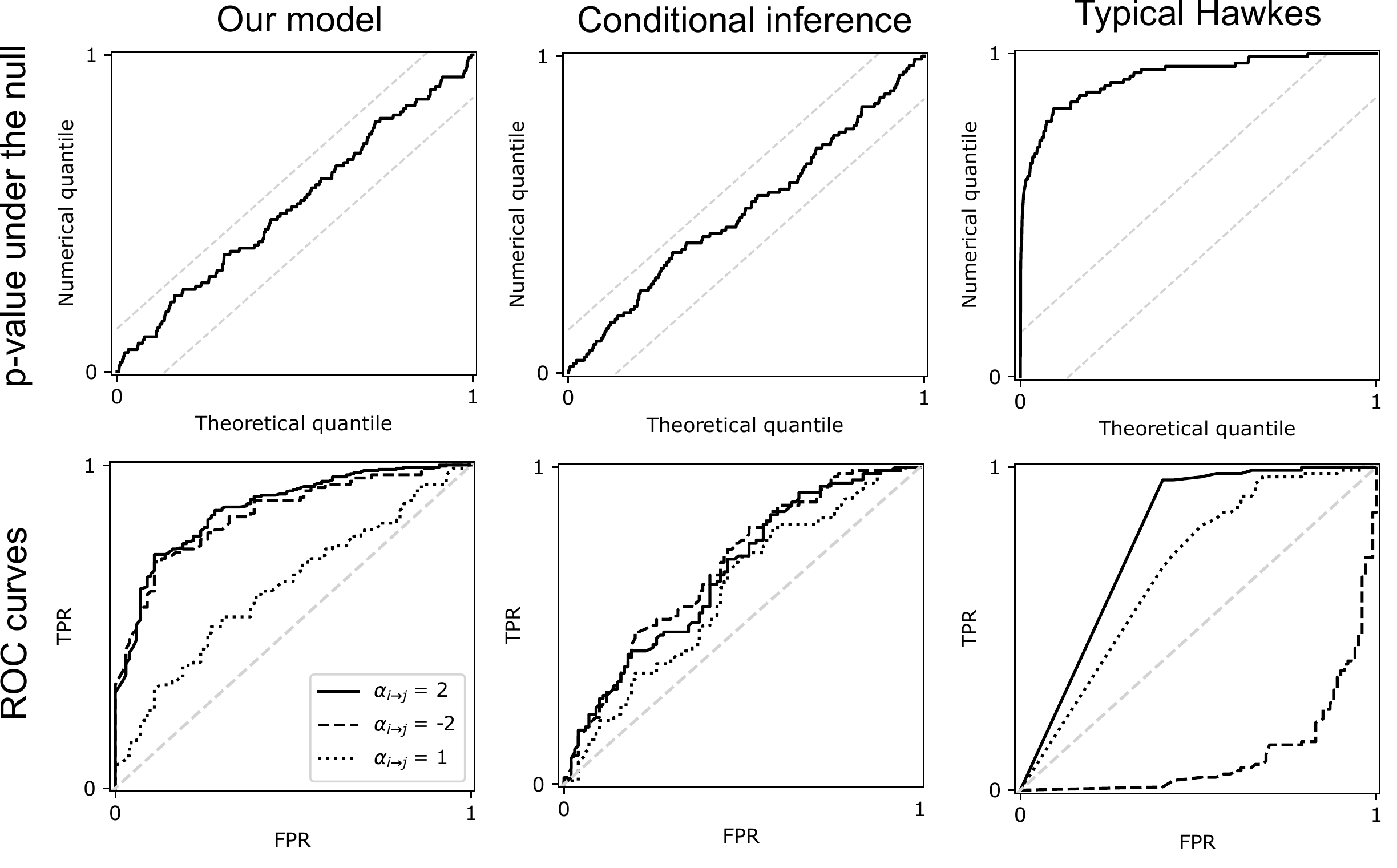}
\end{subfigure}
\caption{Hypothesis testing examples.
3 models are compared: our point process regression model (first column), jitter-based cross-correlation (CCG) method (second column), and the typical Hawkes model (third column).
The simulation details are in the text.
The first row shows the Q-Q plot between p-value distribution under the null (numerical quantile along the y-axis) and the uniform distribution (theoretical quantile along the x-axis). 
The dashed line is the 95\% CI. 
The second row shows the results of ROC analysis with the false positive rate (FPR) along the x-axis, and the true positive rate (TPR) along the y-axis. 
The score of an outcome is the p-value of the hypothesis test.
}
\label{fig:hypothesis_testing}
\end{figure}

Figure \ref{fig:hypothesis_testing} first row verifies if the p-value distribution under the null is uniform. 
Both our method and conditional inference yield proper p-value distributions.
As typical Hawkes model suffers large error due to background artifacts, the p-value under the null is ill.
Figure \ref{fig:hypothesis_testing} second row presents the ROC analysis. 
The score of a test outcome is the p-value.
The our method has better performance when the amplitude of the coupling effect is $\alpha_{i\to j} = 2, -2$ spikes/sec. Neither methods has satisfactory performance if the coupling effect is as weak as $\alpha_{i\to j}=1$ spikes/sec.
Usually, the jitter-based method focuses on pointwise statistic at a specific time lag, which can ignore the connection between adjacent time lags.
We think the power of the CCG method can be improved by considering the time lag dependency and designing the multiple hypothesis test more carefully, but it is not the main interest of this paper.
Typical Hawkes model has positive bias in this setting. 
When the coupling effect is inhibitory $\alpha_{i\to j}=-2$, some of the outcome will be detected as excitatory instead of inhibitory, so the ROC curve is under the diagonal.
When the coupling effect is excitatory $\alpha_{i\to j}=1, 2$, the model will be over confident, so large p-values in the outcome are missing, for example, when $\alpha_{i\to j}=2$ the smallest FPR observed is around 0.4.

\subsection{Simple Bayesian model } \label{appendix:bayesian}
The regression model in main \eqref{eq:target_likelihood} is probabilistic, so it can be easily adapted for Bayesian inference. In this section, we present some simple Bayesian models where the scale $\sigma_w$ of the smoothing kernel $W$ in \eqref{eq:mean_coarsen_regressor} can be treated as a random variable.
We want to investigate how incorporating the uncertainty of the smoothing kernel width affects the estimation of the impact function, and how the variance of the background timescale affects the uncertainty of the smoothing kernel width.
The posterior of the impact function coefficients obtained using the sampling-based method can also verify the Normality property in the regression method when the sample size dominates the prior.

We consider two Bayesian models below and the basic point process regression model.
The likelihood of the model is the same as the main \eqref{eq:target_likelihood}. 
The impact function is estimated using a square window, same as \eqref{eq:square_coupling_filter}.
We choose non-informative flat priors for all the variables.
As the sample size is large, the posterior does not heavily rely on the prior. 
Model 2 is similar to the regression model, where the kernel width $\sigma_w$ is selected using the same way as the regression model and held as fixed. 
Model 2 and the regression model are expected to have similar results.
In Model 1, $\sigma_w$ is a random variable.
We performed the estimation on two datasets: 
The first dataset is the same as the example in Figure \ref{fig:estimator_properties} (details are in the main text);
The second one is the same as the scenario in Appendix \ref{appendix:varying_sigma_I}, where the timescale of the background activity $\sigma_I$ randomly changes in a continuous range between 80 ms and 140 ms.
The true impact function is a square window 
$h_{i\to j}(t) = \alpha_{i\to j} \cdot \mathbb{I}_{[0, \sigma_h]}(t)$
where the amplitude of the impact function is $\alpha_{i\to j}=2$ spikes/sec. The timescale of the square window is $\sigma_h=30$ ms. 
We used the Hastings-Metropolis method for the model inference, which was a Monte Carlo Markov Chain (MCMC) sampler. The posterior was acquired by drawing 1000 samples.
The model was initialized using the basic point process regression method main \eqref{eq:target_likelihood}.
The basic regression model approximates the estimator's distribution using Normal distribution; the mean is the MLE $\hat\alpha_{i\to j}$, and the standard error is from the Fishier information. 

\textbf{Model 1:}
\begin{align*}
\beta_j, \beta_w, \alpha_{i\to j}, \sigma_w \propto&\; 1 \\
p(\beta_j, \beta_w, \alpha_{i\to j}, \sigma_w | \textbf{s}_j, \textbf{s}_i)
\propto& p(\textbf{s}_j | \textbf{s}_i, \beta_j, \beta_w, \alpha_{i\to j}, \sigma_w)
\end{align*}
where $p(\textbf{s}_j | \textbf{s}_i, \beta_j, \beta_w, \alpha_{i\to j}, \sigma_w)$ is the likelihood function of the point process similar to main \eqref{eq:target_likelihood}.
$\sigma_w$ is a variable of the model.

\textbf{Model 2:}
$\sigma_w$ is fixed and the parameter selection follows the regression method.
\begin{align*}
\beta_j, \beta_w, \alpha_{i\to j} \propto&\; 1 \\
p(\beta_j, \beta_w, \alpha_{i\to j} | \textbf{s}_j, \textbf{s}_i)
\propto& p(\textbf{s}_j | \textbf{s}_i, \beta_j, \beta_w, \alpha_{i\to j}, \sigma_w)
\end{align*}

\begin{figure}[H]
\centering
\begin{subfigure}{1\textwidth}
\centering
\includegraphics[width=1\linewidth]{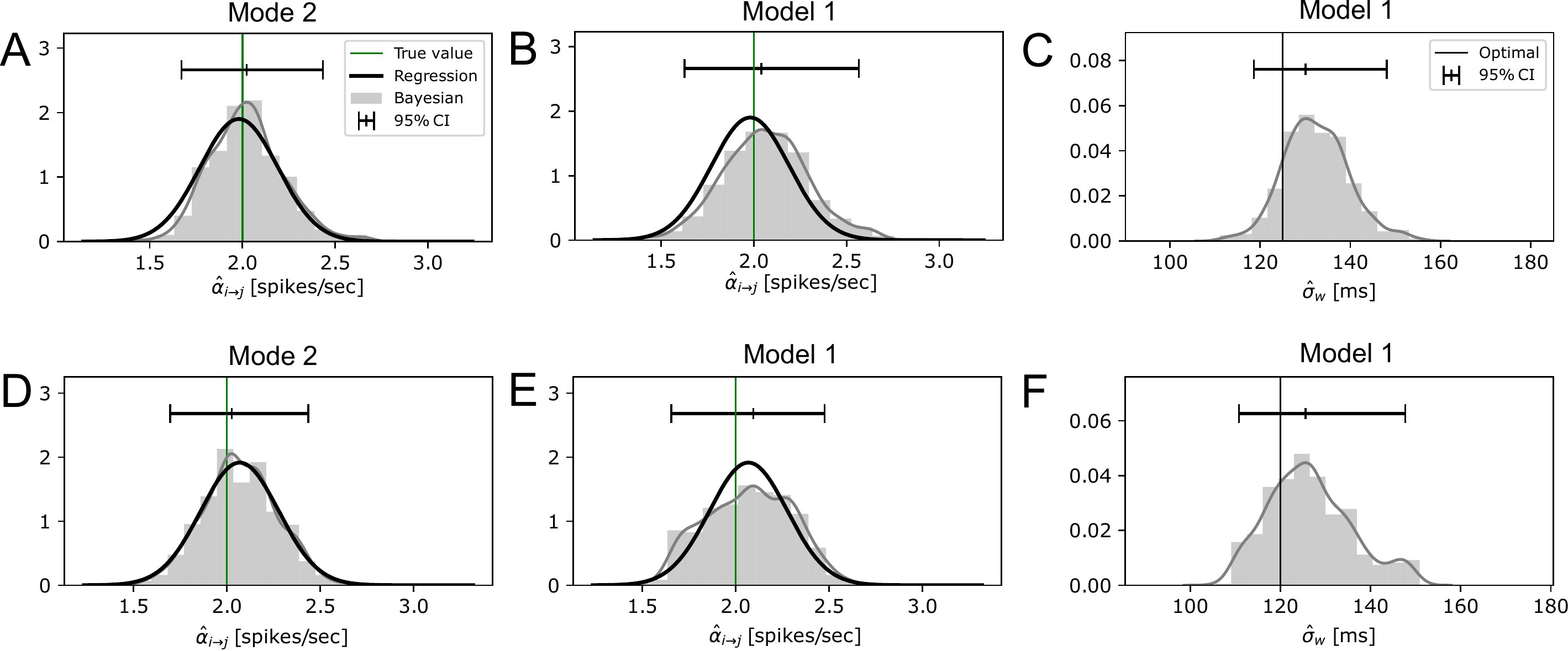}
\end{subfigure}
\caption{Applications of Bayesian model 1 and model 2 to two datasets.
The figure shows the posterior of the estimated impact function amplitude $\hat\alpha_{i\to j}$ of Bayesian models 1 and 2 (grey histograms in A,B,D,E) and the posterior of the smoothing kernel scale $\hat\sigma_w$ of model 1 (grey histograms in C, F).
The solid dark curves are the Normal distributions of $\hat\alpha_{i\to j}$ obtained using the point process regression in main \eqref{eq:target_likelihood}.
\textbf{A, B, C}
Applications of Bayesian models 1, 2, and the basic regression model to the dataset in Figure \eqref{fig:estimator_properties}. 
The timescale of the background activity $f_{i}, f_{j}$ is fixed at $\sigma_I=100$ ms. Details of the dataset description is in the main text.
In C, the mode of the kernel scale is 130 ms, and the 95\% CI is [119, 148] ms.
The optimal kernel scale $\sigma_w$ selected by the regression model is 125 ms.
\textbf{D, E, F}
Applications of Bayesian models 1, 2, and the basic regression model to the dataset in section \ref{appendix:varying_sigma_I}, where the timescale of the shared activity $\sigma_I$ varies from 80 ms to 140 ms.
In F, the mode of the kernel scale is 126 ms, and the 95\% CI is [111, 148] ms.
The optimal kernel scale $\sigma_w$ selected by the regression model is 120 ms.
}
\label{fig:bayesian}
\end{figure}

Figure \ref{fig:bayesian} presents the estimated impact function coefficient of Bayesian models 1, 2, and the basic regression model using two simulation datasets.
One dataset has fixed shared activity timescale $\sigma_I=100$ ms, shown in plot A,B,C;
the other dataset has a time-varying timescale in a continuous range between 80 ms and 140 ms, shown in plots D,E,F.
In A and D, the posterior distributions of $\hat\alpha_{i\to j}$ (grey histogram) and the estimated distribution of the regression model (solid curves) are very close. This can be a side proof of the result in Appendix \ref{appendix:normality}, that $\hat\alpha_{i\to j}$ has asymptotic Normal distribution. 
In both datasets (first row and second row), by comparing the results between model 1 and model 2, incorporating the uncertainty of the smoothing kernel scale $\sigma_w$ does not change the posterior of $\hat\alpha_{i\to j}$ too much.
As shown in Figure \ref{fig:property_tuning}, the selected smoothing kernel scale $\sigma_w$ is related to the timescale of the shared activity $\sigma_I$.
If $\sigma_I$ increases, the corresponding selected $\sigma_w$ will increase by around the same amount.
By comparing Figure \ref{fig:bayesian}C and F, the CI width does not change a lot (from 29 ms in C to 37 ms in F) when the timescale of the background switched from a fixed value $\sigma_I=100$ ms to a randomly varying value in $[80, 140]$ ms.
So the uncertainty of the $\sigma_w$ does not directly reflect the variance of the shared activity timescale.

\section{Derivations related to main equation \ref{eq:bias_formula}} \label{appendix:theoretical_derivations}


In this section, we provide derivations of the estimator properties, including bias, standard error, risk.
All of these are based on the second-order stationary background similar to \eqref{eq:linear_cox} in main \ref{subsec:simulation_study}.

\begin{definition} \label{lemma:second_order_def}
Let $\xi$ be a second-order stationary random measure on $\mathcal{X}$. It satisfies two properties \cite{daley2003introduction}:
\begin{enumerate}
\item The first-moment measure is
$M_{\xi,1}(A) := \mathbb{E} \xi(A)$,
where $A$ is a set in the Borel $\sigma$-field of $\mathcal{X}$, satisfies,
\begin{equation}
M_{\xi,1}(\mathrm{d}x) = \bar\lambda \mathrm{d}x
\end{equation}
where $\bar\lambda$ is a constant, which is called the mean density.

\item The second-moment measure is
$M_{\xi,2}(A \times B) := \mathbb{E} \xi(A) \xi(B)$.
$A, B$ are sets in the Borel $\sigma$-field of $\mathcal{X}$.
The second-moment can be expressed as the product of a Lebesgue component $\mathrm{d}x$ and a reduced measure, say $\breve{M}_{\xi,2}$. $\breve{m}_{\xi,2}$ is the density of the reduced measure 
$\breve{M}_{\xi,2}(\mathrm{d}u) = \breve{m}_{\xi,2}(u) \mathrm{d}u$. The following equation holds,
\begin{equation}
\int_\mathcal{X} \int_\mathcal{X} f(s,t) M_{\xi,2}(\mathrm{d}s \times \mathrm{d}t)
= \int_\mathcal{X} \int_\mathcal{X} f(x, x + u) \mathrm{d}x \cdot \breve{m}_{\xi,2}(u) \mathrm{d}u
\end{equation}
\end{enumerate}
\end{definition}
The reduced second-moment measure $\breve{M}_{\xi,2}$ is symmetric, positive, positive-definite and translation-bounded. Details can be found in \citep[proposition 8.1.I, 8.1.II]{daley2003introduction}.
The mean corrected process is $\tilde\xi(A):= \xi(A) - \bar\lambda \ell(A)$.
Similarly, the reduced covariance measure and its density can be defined as,
\begin{gather} \label{eq:covariance_intensity_def}
\breve{C}_{\xi}(\mathrm{d}u) 
:= \breve{M}_{\tilde\xi,2}(\mathrm{d}u)
= \breve{M}_{\xi,2}(\mathrm{d}u) - \bar\lambda^2\mathrm{d}u \\
\breve{c}_{\xi}(u)
= \breve{m}_{\xi,2}(u) - \bar\lambda^2
\end{gather}

\begin{lemma} \label{lemma:bias}
Assuming $f_{i} = f_{j}$ is second-order stationary.
The intensities of two coupling processes are
$\lambda_{j}(t)
= \alpha_j
+ f_{j}(t) 
+ \int_0^t h_{i\to j}(t-\tau) N_{i}(\mathrm{d} \tau) $
and 
$\lambda_{i}(t) 
= \alpha_i + f_{i}(t)$.
The impact function has format 
$h_{i\to j}(\tau) = \alpha_{i\to j} h(\tau)$ where only the amplitude needs to be fitted,
the bias of the estimator $\hat\alpha_{i\to j}$ using model \eqref{eq:target_likelihood} is approximated as,
\begin{equation}
\mathrm{bias}(\hat\alpha_{i\to j}) \approx 
\frac{ 
\langle W, W \rangle_{\breve{c}_{N}} 
    \langle h, \mathbf{1} \rangle_{\breve{c}_{\Lambda}} 
- \langle h, W \rangle_{\breve{c}_{N}} 
    \langle W, \mathbf{1} \rangle_{\breve{c}_{\Lambda}}
}{
\langle W, W \rangle_{\breve{c}_{N}} 
\langle h, h^- \rangle_{\breve{c}_{N}}
- \langle W, \mathbf{1} \rangle_{\breve{c}_{\Lambda}}^2
}
\end{equation}
where $\breve{c}_{N}$ is the reduced second-order moment measure intensity of spike count measure $N_i(\cdot)$;
$\breve{c}_{\Lambda}$ is the reduced second-order moment measure intensity of the intensity measure $\Lambda_i(\cdot)$ as described in  \ref{lemma:second_order_def}. 
$\ast$ denotes the convolution, $\mathbf{1}$ is a constant function,
$h^-(\tau)=h(-\tau)$.
The operator between two functions $g_1, g_2$ is defined as
\begin{equation}
\langle g_1, g_2 \rangle_{\breve{c}} 
:= \int [g_1\ast g_2](s) \breve{c}(\mathrm{d}s) \mathrm{d}s
\end{equation}

Additionally, if the background activity $f_{i}$ follows the cluster process in \eqref{eq:linear_cox} with parameters $\sigma_I, \rho$, and the impact function has form in \eqref{eq:square_coupling_filter} with parameters $\sigma_h$,
then we have 
$\mathrm{bias}(\hat\alpha_{i\to j}) \approx \mathrm{Numerator} / \mathrm{Denominator}$ as follows,
\begin{equation} \label{eq:bias_linear_cox}
\begin{aligned}
\mathrm{Numerator} 
=& \left( \frac{\rho}{2 \sqrt{\pi}\sqrt{\sigma_w^2 + \sigma_I^2}}
        + \frac{\bar\lambda_i}{2\sqrt{\pi}\sigma_w} \right) \cdot 
    \left(\frac{\rho }{2 }
        \mathrm{erf}\left(\frac{\sigma_h }{2\sigma_I}\right)
        \right) \\
&-\left(\frac{\rho }{2}
    \mathrm{erf}\left(\frac{\sigma_h }{2\sqrt{\sigma_w^2/2 + \sigma_I^2}}\right)
    + \frac{\bar\lambda_i}{2}
    \mathrm{erf}\left(\frac{\sigma_h }{\sqrt{2} \sigma_w}\right)
    \right) \cdot 
    \left(\frac{\rho}{2\sqrt{\pi} \sqrt{\sigma_w^2/2 + \sigma_I^2} }
    \right) \\
\mathrm{Denominator} 
=& \left( \frac{\rho}{2 \sqrt{\pi}\sqrt{\sigma_w^2 + \sigma_I^2}}
        + \frac{\bar\lambda_i}{2\sqrt{\pi}\sigma_w} \right) \cdot \\
&\qquad  \left(\rho \left[\sigma_h
        \mathrm{erf}\left(\frac{\sigma_h}{2\sigma_I}\right) 
        -\frac{2\sigma_I}{\sqrt{\pi}}\left(1-e^{-\frac{\sigma_h^2}{4\sigma_I^2}}\right)  \right]
        + \bar\lambda_i \sigma_h\right) \\
&- \left(\frac{\rho}{2}
    \mathrm{erf}\left(\frac{\sigma_h }{2\sqrt{\sigma_w^2/2 + \sigma_I^2}}\right)
    + \frac{\bar\lambda_i}{2}
    \mathrm{erf}\left(\frac{\sigma_h }{\sqrt{2} \sigma_w}\right)
    \right)^2 
\end{aligned}
\end{equation}
$\bar\lambda_i = \mathbb{E} [N_i(\mathrm{d}t)/\mathrm{d}t] = \alpha_i+\rho$.
$\mathrm{erf}(x)=\frac{2}{\sqrt{\pi}}\int_0^x e^{-t^2}\mathrm{d}t$.
\end{lemma}
\begin{proof}
The bases of the regression model \eqref{eq:regression_intensity} include a constant, the nuisance variable $\overline{\textbf{s}_i} = W\ast \textbf{s}_i$, and the impact function term $h_{i\to j}\ast \textbf{s}_i$.
The target is,
\begin{equation}
\tilde{\ell} :=
\underbrace{
- \int_0^T \log \tilde{\lambda}_{j}(s) \mathrm{d}N_j
+ \int_0^T \tilde{\lambda}_{j}(s) \mathrm{d} s
    }_{\tilde{\ell}_j }
\underbrace{
- \int_0^T \log \tilde{\lambda}_{i}(s) \mathrm{d}N_i
+ \int_0^T \tilde{\lambda}_{i}(s) \mathrm{d} s
    }_{\tilde{\ell}_i }
\end{equation}
where $\tilde\lambda_j$ is rewritten as,
\begin{equation*}
\tilde{\lambda}_{j}(s) = 
\beta_j
+ \beta_w \varphi_w(s)
+ \alpha_{i\to j} \varphi_h(s)
\end{equation*}
$\varphi_w, \varphi_h$ are mean-subtracted bases,
\begin{align*}
\varphi_w(s) :=& \int W(s-t) (N_i(\mathrm{d}t)- N_i(T)/T \mathrm{d}t) \\
\varphi_h(s) :=& \int h_{i\to j}(s-t) (N_i(\mathrm{d}t)- N_i(T)/T \mathrm{d}t)
\end{align*}
where $N_i(T)/T \to \mathbb{E} [N_i(\mathrm{d}t)/\mathrm{d}t]$ as $T\to \infty$.
When estimating the filter $h_{i\to j}$, minimizing the total negative log-likelihood $\tilde{\ell}$ is equivalent to minimizing the negative log-likelihood $\tilde{\ell}_j$, which can be approximated using the Laplace method with coefficients $H$ and $\boldsymbol{b}$,
\begin{equation}
\begin{aligned}
\tilde{\ell}_j (\boldsymbol\beta )
- \tilde{\ell}_j (\boldsymbol\beta_{\mathrm{MLE}} )
=& \frac{1}{2} (\boldsymbol\beta - \boldsymbol\beta_{\mathrm{MLE}})^T 
\frac{\partial^2 \tilde{\ell}_j}{\partial \beta_{\mathrm{MLE}} \partial \beta_{\mathrm{MLE}}^T} (\boldsymbol\beta - \boldsymbol\beta_{\mathrm{MLE}})
    + \left(\frac{\partial \tilde{\ell}_j}{\partial \beta_{\mathrm{MLE}}}\right)^T
    (\boldsymbol\beta - \boldsymbol\beta_{\mathrm{MLE}})
    + o(\|\boldsymbol\beta - \boldsymbol\beta_{\mathrm{MLE}} \|^2) \\
=& \boldsymbol\beta^T H \boldsymbol\beta
    + \boldsymbol{b}^T \boldsymbol\beta + \text{const}
\label{eq:log_likelihood_quad_approx}
\end{aligned}
\end{equation}
$H$ is the Hessian matrix can be obtained from second-order derivative, analytical form of $\boldsymbol{b}$ needs further approximation.
\begin{gather}
\frac{\partial \tilde{\ell}_j }{\partial\boldsymbol\beta}  
\approx H \boldsymbol\beta + \boldsymbol{b}, \quad
H = \frac{\partial^2 \tilde{\ell}_j }{\partial\boldsymbol\beta \partial\boldsymbol\beta^T}
\end{gather}
The MLE thus can be expressed as 
$\boldsymbol\beta_{\mathrm{MLE}} \approx - H^{-1} \boldsymbol{b}$.

Define the following shorthands,
\begin{gather} \label{eq:Sxx_def}
S_{ww} := \langle \varphi_w, \varphi_w \rangle, \quad
S_{hh} := \langle \varphi_h, \varphi_h \rangle, \quad
S_{hw} = S_{wh} := \langle \varphi_w, \varphi_h \rangle
S_{w\lambda} := \langle \varphi_w, \lambda_i \rangle, \quad
S_{h\lambda} := \langle \varphi_h, \lambda_i \rangle, \quad
\end{gather}
$\langle \cdot,\cdot\rangle$ denotes the inner product between two functions on interval $[0, T]$. Lemma \ref{lemma:Sxx} will show the analytical forms of these inner products in the special case with $f_i$ being the linear Cox process as in \eqref{eq:linear_cox}.

\begin{align*}
H =&
\frac{\partial^2 \tilde{\ell}_j}{\partial \beta \partial \beta^T}
= \int_0^T \frac{\Psi(s) \Psi(s)^T}{\tilde\lambda_j(s)^2} N_j(\mathrm{d} s)
\approx \mathbb{E}_{N_j} \left[
\int_0^T \frac{\Psi(s) \Psi(s)^T}{\tilde\lambda_j(s)^2} N_j(\mathrm{d} s)
\middle| N_i \right] \\
=& \int_0^T \frac{\Psi(s) \Psi(s)^T}{\tilde\lambda_j(s)^2} \lambda_j(s) \mathrm{d}s
\approx \frac{1}{\bar\lambda_j} \int_0^T \Psi(s) \Psi(s)^T \mathrm{d}s
\end{align*}
$\Psi(s) = (\mathbf{1}, \varphi_w, \varphi_h)^T$ is the vector of two bases.
$1/T \langle \varphi_h, \mathbf{1} \rangle \to 0$ as $T\to \infty$.
The parameter $\textbf{b}$ in \eqref{eq:log_likelihood_quad_approx} can be solved using two special (suboptimal) solutions with $\hat{\boldsymbol\beta}^B$ at the conditions
\begin{gather*} 
\hat\alpha_{i\to j}^B = 0, \quad
\frac{\partial \tilde{\ell}_j} {\partial \beta_w} = 0, \quad
\frac{\partial \tilde{\ell}_j} {\partial \beta_j^B} = 0
\end{gather*}
and $\hat{\boldsymbol\beta}^C$ at conditions
\begin{gather*}
\hat\beta^C_{w} = 0, \quad
\frac{\partial \tilde{\ell}_j } {\partial \alpha_{i\to j}^C} = 0, \quad
\frac{\partial \tilde{\ell}_j} {\partial \beta_j^C} = 0
\end{gather*}
Solution $\hat\beta_w^B$ corresponds to the model
$\tilde\lambda_j = \beta_j + \beta_w\varphi_w$ without the impact function term at the condition $\hat\alpha_{i\to j}^B=0$.
\begin{align*}
&0= \frac{\partial \tilde{\ell}_j }{\partial \beta_w }
=-\int_0^T \frac{\varphi_w(s)}{\tilde\lambda_{j}(s)}\mathrm{d}N_j(s)
+ \int_0^T \varphi_w(s) \mathrm{d}s \\
=& -\int_0^T\varphi_w(s) \frac{1}{\bar\lambda_{j} + \hat\beta_w^B \varphi_w(s)} 
    \mathrm{d}N_j(s)
= - \frac{1}{\bar\lambda_j}
    \int_0^T\varphi_w(s) \frac{1}{1 + \frac{\hat\beta_w^B}{\bar\lambda_{j}}\varphi_w(s)} 
    \mathrm{d}N_j(s) \\
=& - \frac{1}{\bar\lambda_j}
    \int_0^T\varphi_w(s) \left(1 - \frac{\hat\beta_w^B}{\bar\lambda_j} \varphi_w(s) 
    \right) \mathrm{d}N_j(s)
    + o\left( 
    \frac{\hat\beta_w^B}{\bar\lambda_j^2}
    \int_0^T\varphi_w(s) \varphi_w(s) \mathrm{d}N_j(s) \right) \\
\approx& \mathbb{E}\left[ 
    - \frac{1}{\bar\lambda_j}
    \int_0^T\varphi_w(s) \left(1 - \frac{\hat\beta_w^B}{\bar\lambda_j} \varphi_w(s) 
    \right) \mathrm{d}N_j(s)
    \middle | N_i
\right] \\
=& - \frac{1}{\bar\lambda_j}
    \int_0^T\varphi_w(s) \left(1 - \frac{\hat\beta_w^B}{\bar\lambda_j} \varphi_w(s) 
    \right)\lambda_j(s)  \mathrm{d}s
\end{align*}
Then we can derive the $\hat\beta_w^B$,
\begin{align*}
\hat\beta_w^B \approx& \bar\lambda_j \frac{\langle \varphi_w, \lambda_j \rangle }
    {\langle \varphi_w^2, \lambda_j \rangle }
\approx \bar\lambda_j \frac{\langle \varphi_w, \lambda_j \rangle }
    {\langle \varphi_w^2, \bar\lambda_j \rangle }
= \frac{\langle \varphi_w, \lambda_j \rangle }
    {\langle \varphi_w, \varphi_w \rangle }
\end{align*}
Similarly, we have
\begin{equation*}
\hat\alpha_{i\to j}^C \approx \bar\lambda_j \frac{\langle \varphi_h, \lambda_j \rangle }
    {\langle \varphi_h^2, \lambda_j \rangle }
\approx \frac{\langle \varphi_h, \lambda_j \rangle }
    {\langle \varphi_h, \varphi_h \rangle }
\end{equation*}
The MLE then is,
\begin{equation}
\hat\beta \approx - H^{-1} \boldsymbol{b}
\end{equation}
\begin{equation*}
H \approx \frac{1}{\bar\lambda_j} \left( \begin{array}{ccc}
        S_{ww} & S_{wh} & 0\\
        S_{hw} & S_{hh} & 0 \\
        0 & 0& \hat \beta_j^2 T
    \end{array}\right), \quad
\boldsymbol{b} \approx - \frac{1}{\bar\lambda_j}
\left(\begin{array}{cc}
        \langle \varphi_w, \lambda_j \rangle \\
        \langle \varphi_h, \lambda_j \rangle \\
        T \bar\lambda_j^3
    \end{array} \right)
\end{equation*}
where $\bar\lambda_j = \mathbb{E}\lambda_j 
= \bar\lambda_i + \alpha_{i\to j} \sigma_h$.
So we have the estimator $\hat\alpha_{i\to j}$,
\begin{align*}
\hat\alpha_{i\to j} 
\approx& \frac{S_{ww}\cdot \langle \varphi_h, \lambda_j \rangle 
    - S_{hw}\cdot \langle \varphi_w, \lambda_j \rangle }
{S_{ww}S_{hh} - S_{wh}^2} \\
=& \frac{S_{ww}\cdot \langle \varphi_h, \alpha_j + f_{i} + \alpha_{i\to j} \varphi_h \rangle 
    - S_{hw}\cdot \langle \varphi_w, \alpha_j + f_{i} + \alpha_{i\to j} \varphi_h \rangle }
{S_{ww}S_{hh} - S_{wh}^2} \\
=& \frac{S_{ww}\cdot \langle \varphi_h, f_{i} \rangle 
    - S_{hw}\cdot \langle \varphi_w, f_{i} \rangle }
{S_{ww}S_{hh} - S_{wh}^2}
+ \alpha_{i\to j} \cdot
\frac{S_{ww}\cdot \langle \varphi_h, \varphi_h \rangle 
    - S_{hw}\cdot \langle \varphi_w, \varphi_h \rangle }
{S_{ww}S_{hh} - S_{wh}^2} \\
=& \frac{S_{ww}\cdot \langle \varphi_h, \alpha_i + f_{i} \rangle 
    - S_{hw}\cdot \langle \varphi_w, \alpha_i + f_{i} \rangle }
{S_{ww}S_{hh} - S_{wh}^2}
+ \alpha_{i\to j} \\
\approx& \alpha_{i\to j}
+ \frac{ S_{ww}\langle \varphi_h, \lambda_i \rangle
    -S_{hw}\langle \varphi_w, \lambda_i \rangle
}{S_{ww}S_{hh} - S_{hw}^2}
\end{align*}
So the bias of the estimator is approximately,
\begin{equation}
\mathrm{bias}(\hat\alpha_{i\to j}) 
\approx 
\frac{ S_{ww} S_{h\lambda}
    -S_{hw}S_{w\lambda}
}{S_{ww}S_{hh} - S_{hw}^2}
\end{equation}
Lemma \ref{lemma:Sxx} shows the derivation of the inner products
$S_{ww}, S_{hh},S_{hw},S_{w\lambda},S_{h\lambda}$, which lead to the equations for the linear Cox background in \eqref{eq:bias_linear_cox}.
\end{proof}

\begin{corollary} \label{corollary:bias_hawkes}
If the regression model \eqref{eq:regression_intensity} does not include the nuisance variable, which becomes a typical Hawkes process, then the bias of the estimator is the following with similar derivation.
\begin{equation}
\mathrm{bias}(\hat\alpha_{i\to j}) 
\approx 
\frac{S_{h\lambda}}{S_{hh}}
\approx
\frac{
    \frac{\rho }{2 }
    \mathrm{erf}\left(\frac{\sigma_h }{2\sigma_I}\right)
}{\rho \left[\sigma_h
    \mathrm{erf}\left(\frac{\sigma_h}{2\sigma_I}\right) 
    -\frac{2\sigma_I}{\sqrt{\pi}}\left(1-e^{-\frac{\sigma_h^2}{4\sigma_I^2}}\right)  \right]
    + \bar\lambda_i \sigma_h
}
\end{equation}
\end{corollary}

\begin{corollary} \label{corollary:bias_sharp_W}
When the smoothing kernel becomes infinitely narrow, the bias in \eqref{eq:bias_linear_cox} satisfies
\begin{equation}
\lim_{\sigma_w\to 0} \mathrm{bias}(\hat\alpha_{i\to j}) 
\to 
\frac{
    \frac{\rho }{2 }
    \mathrm{erf}\left(\frac{\sigma_h }{2\sigma_I}\right)
}{\rho \left[\sigma_h
    \mathrm{erf}\left(\frac{\sigma_h}{2\sigma_I}\right) 
    -\frac{2\sigma_I}{\sqrt{\pi}}\left(1-e^{-\frac{\sigma_h^2}{4\sigma_I^2}}\right)  \right]
    + \bar\lambda_i \sigma_h
}
\end{equation}
\end{corollary}

\begin{corollary} \label{corollary:bias_wide_W}
When the smoothing kernel becomes infinitely wide, the bias in \eqref{eq:bias_linear_cox} satisfies
\begin{equation}
\lim_{\sigma_w\to \infty} \mathrm{bias}(\hat\alpha_{i\to j}) 
\to 
\frac{
    \frac{\rho }{2 }
    \mathrm{erf}\left(\frac{\sigma_h }{2\sigma_I}\right)
}{\rho \left[\sigma_h
    \mathrm{erf}\left(\frac{\sigma_h}{2\sigma_I}\right) 
    -\frac{2\sigma_I}{\sqrt{\pi}}\left(1-e^{-\frac{\sigma_h^2}{4\sigma_I^2}}\right)  \right]
    + \bar\lambda_i \sigma_h
}
\end{equation}
\end{corollary}

The approximation is similar to Lemma \ref{lemma:bias}. Note that the three corollaries have the same results.

\begin{lemma} \label{lemma:var}
Same as the settings in Lemma \ref{lemma:bias},
the variance of the estimator is approximated as,
\begin{equation}
\mathrm{Var}(\hat\alpha_{i\to j}) \approx 
\frac{\bar\lambda_j}{T}
\frac{ \langle W, W \rangle_{\breve{c}_{N}} 
}{
\langle W, W \rangle_{\breve{c}_{N}} 
\langle h, h^- \rangle_{\breve{c}_{N}}
- \langle W, \mathbf{1} \rangle_{\breve{c}_{\Lambda}}^2
}
\end{equation}
If $f_{i}$ follows the cluster process in \eqref{eq:linear_cox}, then
$\mathrm{Var}(\hat\alpha_{i\to j}) 
\approx \mathrm{Numerator}/\mathrm{Denominator}$
\begin{equation} \label{eq:var_linear_cox}
\begin{aligned}
\mathrm{Numerator}
=& \frac{\bar\lambda_j}{T} \left( \frac{\rho}{2 \sqrt{\pi}\sqrt{\sigma_w^2 + \sigma_I^2}}
        + \frac{\bar\lambda_i}{2\sqrt{\pi}\sigma_w} \right) \\
\mathrm{Denominator}
=& \left( \frac{\rho}{2 \sqrt{\pi}\sqrt{\sigma_w^2 + \sigma_I^2}}
        + \frac{\bar\lambda_i}{2\sqrt{\pi}\sigma_w} \right) \cdot \\
&\qquad  \left(\rho \left[\sigma_h
        \mathrm{erf}\left(\frac{\sigma_h}{2\sigma_I}\right) 
        -\frac{2\sigma_I}{\sqrt{\pi}}\left(1-e^{-\frac{\sigma_h^2}{4\sigma_I^2}}\right)  \right]
        + \bar\lambda_i \sigma_h\right) \\
&- \left(\frac{\rho}{2}
    \mathrm{erf}\left(\frac{\sigma_h }{2\sqrt{\sigma_w^2/2 + \sigma_I^2}}\right)
    + \frac{\bar\lambda_i}{2}
    \mathrm{erf}\left(\frac{\sigma_h }{\sqrt{2} \sigma_w}\right)
    \right)^2 
\end{aligned}
\end{equation}
\end{lemma}
The proof is similar to Lemma \ref{lemma:bias} using the Fisher information.

\begin{corollary} \label{corollary:var_hawkes}
If the regression model \eqref{eq:regression_intensity} does not include the nuisance variable, which becomes a typical Hawkes process, then the variance of the estimator is
\begin{align*}
\mathrm{Var}(\hat\alpha_{i\to j}) \approx 
\frac{\bar\lambda_j}{S_{hh}}
\approx \frac{\bar\lambda_j}{T}
\left(\rho \left[\sigma_h
    \mathrm{erf}\left(\frac{\sigma_h}{2\sigma_I}\right) 
    -\frac{2\sigma_I}{\sqrt{\pi}}\left(1-e^{-\frac{\sigma_h^2}{4\sigma_I^2}}\right)  \right]
    + \bar\lambda_i \sigma_h
\right)^{-1}
\end{align*}
\end{corollary}
The proof is similar to Lemma \ref{lemma:bias}. The three corollaries above have the same results.

\begin{corollary} \label{corollary:var_sharp_W}
If $\sigma_w \to 0$ of the variance in \eqref{eq:var_linear_cox} will converge
\begin{equation}
\begin{aligned}
\lim_{\sigma_w\to 0} \mathrm{Var}(\hat\alpha_{i\to j})
\to
\frac{\bar\lambda_j}{T}
\left(\rho \left[\sigma_h
    \mathrm{erf}\left(\frac{\sigma_h}{2\sigma_I}\right) 
    -\frac{2\sigma_I}{\sqrt{\pi}}\left(1-e^{-\frac{\sigma_h^2}{4\sigma_I^2}}\right)  \right]
    + \bar\lambda_i \sigma_h
\right)^{-1}
\end{aligned}
\end{equation}
\end{corollary}

\begin{corollary} \label{corollary:var_wide_W}
If $\sigma_w \to \infty$ of the variance in \eqref{eq:var_linear_cox} will converge
\begin{equation}
\lim_{\sigma_w\to \infty} \mathrm{Var}(\hat\alpha_{i\to j})
\to
\frac{\bar\lambda_j}{T}
\left(\rho \left[\sigma_h
    \mathrm{erf}\left(\frac{\sigma_h}{2\sigma_I}\right) 
    -\frac{2\sigma_I}{\sqrt{\pi}}\left(1-e^{-\frac{\sigma_h^2}{4\sigma_I^2}}\right)  \right]
    + \bar\lambda_i \sigma_h
\right)^{-1}
\end{equation}
\end{corollary}

\begin{lemma} \label{lemma:Sxx}
If the linear Cox model in \eqref{eq:linear_cox} and \eqref{eq:square_coupling_filter}, the inner products defined in  \eqref{eq:Sxx_def} can be derived in analytical forms as follows.
\end{lemma}
\begin{proof}
Apply \ref{lemma:second_order_def}, \ref{lemma:inner_prod}, and \ref{lemma:linear_cox_reduced_moment},
\begin{align*}
\frac{1}{T} S_{ww} 
\approx& \int_\mathbb{R} [W\ast W](s) \breve{c}_{N}(\mathrm{d}s) \mathrm{d}s \\
=& \int_{\mathbb{R}} \frac{1}{2\sigma_w \sqrt{\pi}} \exp\left\{ - \frac{s^2}{4\sigma_w^2} \right\}
    \left( 
        \frac{\rho}{2 \sigma_I \sqrt{\pi}} 
            \exp\left\{ -\frac{s^2}{4\sigma_I^2} \right\}
        + \bar\lambda_i \delta(s) 
    \right) \mathrm{d}s \\
=& \frac{\rho}{2 \sqrt{\pi}\sqrt{\sigma_w^2 + \sigma_I^2}}
        + \frac{\bar\lambda_i}{2\sqrt{\pi}\sigma_w} \\
\frac{1}{T} S_{hh} 
\approx& \int_\mathbb{R} [h \ast h^-](s) \breve{c}_{N}(\mathrm{d}s) \mathrm{d}s \\
=& \int_{\mathbb{R}} 
    \left[\mathrm{rect}\left(\frac{u}{\sigma_h} - \frac{1}{2}\right) \ast
          \mathrm{rect}\left(-\frac{u}{\sigma_h} - \frac{1}{2}\right) \right](s)
    \left(\frac{\rho}{2\sigma_I \sqrt{\pi}} 
            \exp\left\{ -\frac{s^2}{4\sigma_I^2} \right\}
        + \bar\lambda_i \delta(s)
    \right) \mathrm{d}s \\
=& \rho \left[
        \sigma_h \mathrm{erf}\left(\frac{\sigma_h}{2\sigma_I}\right) 
        -\frac{2\sigma_I}{\sqrt{\pi}}\left(1- \exp\left\{ \frac{\sigma_h^2}{4\sigma_I^2} \right\} \right) 
        \right]
        + \bar\lambda_i \sigma_h \\
\frac{1}{T} S_{hw}
\approx& \int_\mathbb{R} [h \ast W](s) \breve{c}_{N}(\mathrm{d}s) \mathrm{d}s \\
=& \int_{\mathbb{R}} 
    \left[\mathrm{rect}\left(\frac{u}{\sigma_h} - \frac{1}{2}\right) \ast
          \phi_{\sigma_W}(u) \right](s)
    \left(
        \frac{\rho}{2 \sigma_I \sqrt{\pi}} 
            \exp\left\{ -\frac{s^2}{4\sigma_I^2} \right\}
        + \bar\lambda_i \delta(s) 
    \right) \mathrm{d}s \\
=& \frac{\rho}{2}
    \mathrm{erf}\left(\frac{\sigma_h }{\sqrt{2\sigma_w^2 + 4\sigma_I^2}}\right)
    + \frac{\bar\lambda_i}{2}
    \mathrm{erf}\left(\frac{\sigma_h }{\sqrt{2} \sigma_w}\right) \\
\frac{1}{T} S_{w\lambda}
\approx& \int_\mathbb{R} W(s) \breve{c}_{\Lambda}(s)\mathrm{s} \\
=& \int_{\mathbb{R}} \frac{1}{\sigma_w \sqrt{2\pi}} \exp\left\{ - \frac{s^2}{2\sigma_w^2} \right\}
    \left( 
        \frac{\rho}{2 \sigma_I \sqrt{\pi}} 
            \exp\left\{ -\frac{s^2}{4\sigma_I^2} \right\}
    \right) \mathrm{d}s \\
=& \frac{\rho}{\sqrt{2\sigma_w^2 + 4\sigma_I^2}\cdot \sqrt{\pi}} \\
\frac{1}{T} S_{h\lambda}
\approx& \int_\mathbb{R} h(s) \breve{c}_{\Lambda}(s)\mathrm{s} \\
=& \int_{\mathbb{R}} \mathbb{I}_{[0, \sigma_h]}(s)
    \left(
        \frac{\rho}{2 \sigma_I \sqrt{\pi}} 
            \exp\left\{ -\frac{s^2}{4\sigma_I^2} \right\}
    \right) \mathrm{d}s \\
=& \frac{\rho}{2}\mathrm{erf}\left(\frac{\sigma_h}{2\sigma_I} \right)
\end{align*}
\end{proof}

\begin{lemma} \label{lemma:inner_prod}
Assume the point process $N_i(\cdot)$ is second-order stationary (definition  \ref{lemma:second_order_def}) with reduced second-order moment measure intensity $\breve{c}_{N}$, intensity function $\lambda_i$, and mean intensity $\bar\lambda_i$. Define two mean subtracted processes,
\begin{align*}
\varphi_1(s) :=& \int W_1(s-t) (N_i(\mathrm{d}t)-\bar\lambda_i\mathrm{d}t) \\
\varphi_2(s) :=& \int W_2(s-t) (N_i(\mathrm{d}t)-\bar\lambda_i\mathrm{d}t)
\end{align*}
Then the inner product on interval $[0, T]$ between the processes is
\begin{equation} \label{eq:inner_prod}
\langle \varphi_1, \varphi_2 \rangle
\approx T \int_\mathbb{R} [W_1\ast W_2^-](r)
    \breve{c}_{N}(r)\mathrm{d}r
\end{equation}
$W_2^-(x) := W_2(-x)$.
\begin{equation}
\langle \varphi_w, \lambda_i \rangle 
\approx \int_\mathbb{R} W(r) \breve{c}_{\Lambda} (r) \mathrm{d}s \mathrm{d}t
\end{equation}
$\breve{c}_{\Lambda}$ and $\breve{c}_{N}$ are reduced second-order moment measure intensity corresponding to $\Lambda_i(\cdot)$ and  $N_i(\cdot)$, see Lemma \ref{lemma:linear_cox_reduced_moment}. 
$\Lambda_i(A) := \int_A \lambda_i(t)\mathrm{d}t$ is the intensity measure.
\end{lemma}
\begin{proof}
\begin{align*}
\langle \varphi_1, \varphi_2 \rangle
=& \int_0^T\int_0^T\int_0^T 
    \underbrace{W_1(t-u) }_{s := t - u}
    \underbrace{W_2(t-v)}_{W_2^-(x) := W_2(-x)}
    \left(N_i(\mathrm{d}u) -\bar\lambda_i \mathrm{d}u\right)
    \left(N_i(\mathrm{d}v) -\bar\lambda_i \mathrm{d}v\right)
    \mathrm{d}t \\
=& \int_0^T\int_0^T
    \int_{-u}^{T-u} W_1(s) 
        \underbrace{ W_2^-((v-u)-s) }_{u:=u, r:v-u}
    \mathrm{d}s 
    \left(N_i(\mathrm{d}u) -\bar\lambda_i \mathrm{d}u\right)
    \left(N_i(\mathrm{d}v) -\bar\lambda_i \mathrm{d}v\right) \\
=& \int_0^T \int_{-u}^{T-u}
    \int_{-u}^{T-u} W_1(s) W_2^-(r-s) \mathrm{d}s \cdot
    \breve{c}_{N}(r)\mathrm{d}r \cdot
    \mathrm{d}u \\
\approx& \int_0^T \int_{-u}^{T-u} 
    [W_1 \ast W_2^-](r) \breve{c}_{N}(r)\mathrm{d}r \cdot \mathrm{d}u
\approx T \int_\mathbb{R} [W_1\ast W_2^-](r)
    \breve{c}_{N}(r)\mathrm{d}r
\end{align*}
The approximation error comes from the boundary effect. If the kernels $W_1, W_2$ decays fast, the error can be ignored.
\begin{align*}
\langle \varphi_w, \lambda_i \rangle
=& \int_0^T \int_0^T W(t-u) 
    \left(N_i(\mathrm{d}u)) - \bar\lambda_i \mathrm{d}u\right)
    \lambda_i(t) \mathrm{d} t \\
=& \int_0^T \int_0^T W(t-u) 
    \left(N_i(\mathrm{d}u)) - \bar\lambda_i \mathrm{d}u\right)
    \left(\lambda_i(t)\mathrm{d}t - \bar\lambda_i \mathrm{d}t\right) \\
\approx& T \int_\mathbb{R} W(r) \breve{c}_{\Lambda} (r) \mathrm{d}s \mathrm{d}t
\end{align*}
\end{proof}

\textbf{Remark} 
Many works that study the second-order stationary point process is in the frequency-domain \cite{bartlett1963statistical, bremaud2005power, brillinger1972spectral, brillinger1974cross, hawkes1971spectra, lewis1970remarks, mugglestone1996practical} and \cite[ch. 8]{daley2003introduction}. All of our analysis is in time-domain. If we apply the Parseval's theorem  to \eqref{eq:inner_prod}, it equivalently shifts almost all results into frequency-domain.
\begin{equation*}
\int_\mathbb{R} [W_1\ast W_2^-](r) \breve{c}_{N}(r)\mathrm{d}r
= \int_\mathbb{R} \widehat W_1(f) \cdot \widehat W_2^-(f)\cdot \Gamma_{N}(\mathrm{d}f)
\end{equation*}
where $\widehat W_1, \widehat W_2^-$ are the spectrum of kernels, and $\Gamma_{N}$ is called \textit{Bartlett spectrum} for point process or \textit{Bochner spectrum} for wide-sense process (see \citep[ch. 8]{daley2003introduction} and \citep{bremaud2005power}). This can shift the time-domain analysis into the frequency-domain.
This work does not include any frequency properties of the estimator, but it is promising to interpret some steps using the Bartlett spectrum measure in the future work.

\begin{lemma}\label{lemma:linear_cox_reduced_moment}
Consider the cluster process in \eqref{eq:linear_cox}.
Let $\phi(\cdot)_{\sigma_I}$ be a window function with scale $\sigma_I$, $t^c_i$ be the points of the center process which is generated by homogeneous Poisson process with intensity $\rho$. $\alpha_i$ is the baseline. The intensity function has form,
\begin{equation}
\lambda_i(t) = \alpha_i + \sum_{i}
    \phi_{\sigma_I} \left(t - t^c_i\right)
\end{equation}
$N_i(\cdot)$ is the corresponding count measure.
Assume $\phi_{\sigma_I}$ is a Normal window with mean zero and standard deviation $\sigma_I$.
The reduced covariance measure intensity of $\Lambda_i(t)$ is,
\begin{equation}
\breve{c}_{\Lambda}(u) 
= \rho\cdot [\phi_{\sigma_I} \ast \phi_{\sigma_I}](u)
= \frac{\rho}{\sqrt{4\pi \sigma_I^2}} 
    \exp\left\{ -\frac{u^2}{4\sigma_I^2} \right\}
\end{equation}
Similarly, the reduced covariance measure intensity the point process $N_i(t)$ is,
\begin{equation}
\breve{c}_{N}(u) 
= \rho\cdot [\phi_{\sigma_I} \ast \phi_{\sigma_I}](u) + \bar\lambda_i \delta(u)
= \frac{\rho}{\sqrt{4\pi \sigma_I^2}} 
    \exp\left\{ -\frac{u^2}{4\sigma_I^2} \right\}
    + \bar\lambda_i \delta(u)
\end{equation}
\end{lemma}
\begin{proof}
The first-moment property of the intensity is the following.
\begin{equation*}
\bar\lambda_i 
=\mathbb{E}[\lambda(t)]
= \mathbb{E}\left[\int_0^\infty
    \phi_{\sigma_I} \left(t - s \right) N(\mathrm{d}s) \right]
= \int \phi_{\sigma_I} \left(t - s\right) (\alpha_i + \rho) \mathrm{d} s
= \alpha_i + \rho
\end{equation*}
The \textit{reduced covariance} for the second-moment stationary process is defined as,
\begin{align*}
\breve{c}_{\Lambda}(u) 
= \mathbb{E}[ \lambda_i(x) \lambda_i(x+u)] 
- \mathbb{E}[\lambda_i(x)] \mathbb{E}[\lambda_i(x+u)]
= \mathbb{E}[ \lambda_i(x) \lambda_i(x+u)] - \bar\lambda_i^2
\end{align*}
The second-moment measure of homogeneous Poisson process is \citep{hawkes1971spectra},
\begin{gather*}
\breve{M}^c_{N,2}(\mathrm{d} v)
=\bar\lambda_i \delta(v) \mathrm{d}v + \bar\lambda_i^2 \mathrm{d}v
\end{gather*}
The second equation holds due to the Campbell lemma \cite[Lemma 1.1]{kutoyants1998statistical}. $N^c(\cdot)$ is the count measure of the center process.
$\Lambda_i(\cdot) := \int_A \lambda_i(t)\mathrm{d}t$ is the intensity measure with respect to the intensity $\lambda_i$.
\begin{align*}
& \breve{m}_{\Lambda,2}( u) 
= \mathbb{E}\left[\frac{\Lambda_i(\mathrm{d} x) 
    \Lambda_i(x+\mathrm{d} u)}{\mathrm{d}x\mathrm{d}u}\right]
= \mathbb{E}[ \lambda_i(x) \lambda_i(x+u)] \\
=& \mathbb{E}[(\alpha_i +f_{i}(x)) (\alpha_i +f_{i}(x+u))]
= \mathbb{E}[f_{i}(x) f_{i}(x+u)] + 2\rho\alpha_i + \alpha_i^2 \\
=& \mathbb{E} \left[
\left( \int \phi_{\sigma_I} \left(x - s\right) 
    \mathrm{d}N^c(s) \right)
\left( \int \phi_{\sigma_I} \left(x+u - r\right) 
    \mathrm{d}N^c(r) \right)
\right] + 2\rho\alpha_i + \alpha_i^2 \\
=& \mathbb{E} \left[
    \iint
    \phi_{\sigma_I} \left(x - s\right) \phi_{\sigma_I} \left(x+u - r\right) 
    \mathrm{d}N^c(s) \mathrm{d}N^c(r)
\right] + 2\rho\alpha_i + \alpha_i^2 \\
=& \int \mathrm{d}s \int
    \phi_{\sigma_I} \left(x - s\right) \phi_{\sigma_I} \left(x+u - (s+v) \right) 
    \breve{M}^c_{N,2}(\mathrm{d}v) + 2\rho\alpha_i + \alpha_i^2 \\
=& \bar\lambda_i^2 + \rho \int
    \phi_{\sigma_I} \left(s\right) \phi_{\sigma_I} \left(u - s \right) 
    \mathrm{d}s
= \bar\lambda_i^2 + \rho [\phi_{\sigma_I} \ast \phi_{\sigma_I}](u)
\end{align*}

The reduced covariance measure intensity of the count measure can be derived as follows.
\begin{align*}
&M_{N,2}(\mathrm{d}t \times (t+\mathrm{d}u))
=\mathrm{d}t \cdot \breve{M}_{N}(\mathrm{d}u) \\
=& \mathbb{E}\left[ N_i(\mathrm{d} t) N_i(t+ \mathrm{d} u) \right]
= \mathbb{E}_{\Lambda} \left[
    \mathbb{E}_N \left[ N(\mathrm{d} t) N(t+ \mathrm{d} u) 
        \middle | \Lambda_i \right] \right] \\
=& \bar\lambda_i \delta(u)\mathrm{d}u 
    +\mathbb{E}_{\lambda} \left[\Lambda_i(\mathrm{d}t) \Lambda_i(t+ \mathrm{d} u) \right] 
= \bar\lambda_i \delta(u)\mathrm{d}u \mathrm{d}t
+ \breve{m}_{\Lambda}(u)\mathrm{d}u\mathrm{d}t
\end{align*}
So we have,
\begin{gather*}
\breve{m}_{N}(u) = \bar\lambda_i \delta(u) + \breve{m}_{\Lambda}(u) \\
\breve{c}_{N}(u) = \bar\lambda_i \delta(u) + \breve{c}_{\Lambda}(u)
\end{gather*}
Similarly, the reduced second-order covariance intensity is, 
\begin{align*}
& \mathbb{E}\left[ N(\mathrm{d} t) \Lambda(t+ \mathrm{d} u) \right]
= \mathbb{E}_{\Lambda} \left[
    \mathbb{E}_N \left[N(\mathrm{d} t) \Lambda(t+ \mathrm{d} u) \middle | \lambda \right]
    \right] \\
=& \mathbb{E}_{\Lambda} \left[\Lambda(\mathrm{d} t) \Lambda(t+ \mathrm{d} u) \right] 
= \breve{m}_{\Lambda}(u)\mathrm{d}u\mathrm{d}t
\end{align*}

\end{proof}

\section{Application to neuroscience dataset } \label{appendix:jitter_app}

\subsection{Materials }
We applied our method to the Allen Brain Observatory Visual Coding Neuropixels \cite{siegle2021survey}.
It used multiple high-density extracellular electrophysiology probes to simultaneously record spiking activity from many areas in the mouse brain, especially the visual cortex. 
The animals were passively presented with visual stimuli while the head was fixed.
The details of the experimental setup can be found in \cite{siegle2021survey}.
Our work used drifting gratings as the trials are long, and visual stimuli strongly elicit neural responses.
The drifting gratings have 8 different orientations 
(\ang{0}, \ang{45}, \ang{90}, \ang{135},\ang{180}, \ang{225}, \ang{270}, \ang{315}, clockwise from \ang{0} = right-to-left).
The temporal frequency is 8Hz.
The spatial frequency is 0.04 cycles/deg and the contrast is 80\% for all trials. 
The dataset assigns unique identities for all properties, such as conditions, trials, neurons, etc. In this paper, we refer to those identities directly.
We analyzed mouse session \texttt{798911424}.
The stimulus condition identities are:
\texttt{246},
\texttt{254},
\texttt{256},
\texttt{263},
\texttt{267},
\texttt{276},
\texttt{281},
\texttt{284}.
Each condition has 15 repeated trials, 120 trials in total.
A trial lasts for 3 sec with a 2-second stimulus and a 1-second blank screen.
5 brain areas were recorded by separate Neuropixels probes simultaneously.
The number of recorded neurons in each visual cortical area is roughly 100. 
We selected the top 50\% most active neurons (thresholded by the mean firing rate) including 
47 V1 neuron,
39 LM neurons,
23 RL neurons,
44 AL neurons,
39 AM neurons.

\subsection{Goodness-of-fit} \label{appendix:ks_test}
The goodness-of-fit test was assessed with the Kolmogorov-Smirnov (KS) test based on the time-rescaling theorem \cite{bowsher2007modelling, brown2002time, haslinger2010discrete}.
The theorem states that the transformed inter-spike intervals follow the unit exponential distribution. The KS test is used to compare the empirical distribution and the target distribution.
A good fit should have a straight curve along the diagonal in the Q-Q plot.
Figure \ref{fig:ks_test} shows the results between all pairs of brain areas.
The model does a good job with all curves staying along the diagonal without large deviations.

\begin{figure}[H]
\centering
\begin{subfigure}{0.9\textwidth}
\centering
\includegraphics[width=1\linewidth]{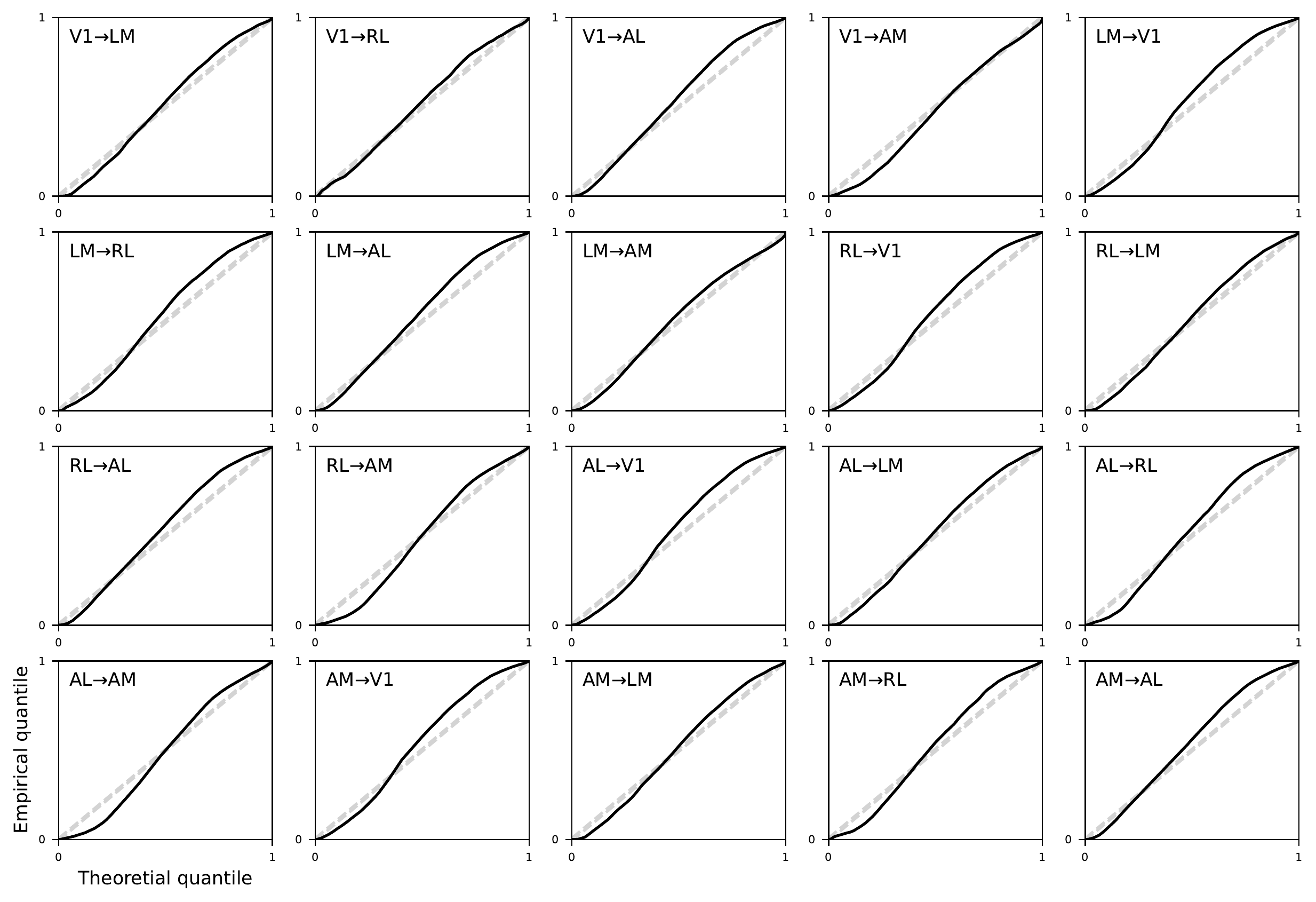}
\end{subfigure}
\caption{Goodness-of-fit tests between all pairs of brain areas. 
Each plot shows the results of an impact function of a pair of neurons. 
The connection direction is labeled at the corner.
The grey dashed lines are 99\% CI. 
A good fit should have a straight curve along the diagonal.
}
\label{fig:ks_test}
\end{figure}

Next, we show examples by comparing the fitted impact function and conditional inference-based CCG as another way of verification.
The CCG together with non-parametric fitting can be used to explore the timescale of the coupling effect.
Figure \ref{fig:ccg_filters} shows an example of excitatory, inhibitory, or neural coupling effects.
Similar to Figure \ref{fig:sim_demo} and \ref{fig:neural_demo}, the jitter-based CCG method may not be sensitive enough to detect the weak signals although it shows some clue of the coupling effect.
When the coupling effects are fitted using square windows in Figure \ref{fig:ccg_filters} second row, it will show a more significant excitatory or inhibitory coupling effect.
We also estimate those effects using non-parametric fitting as shown in Figure \ref{fig:ccg_filters} last row.
Our method allows us to aggregate all the information in a lag window to estimate the impact function with one parameter using a square window, or a few parameters using B-splines. The results will be more effective and significant than the method using the pointwise statistic.
We admit modeling using square windows loses some details of the coupling effect, but it is effective in capturing the general properties of the coupling effect with a limited dataset and large noise.

\begin{figure}[H]
\centering
\begin{subfigure}{0.9\textwidth}
\centering
\includegraphics[width=1\linewidth]{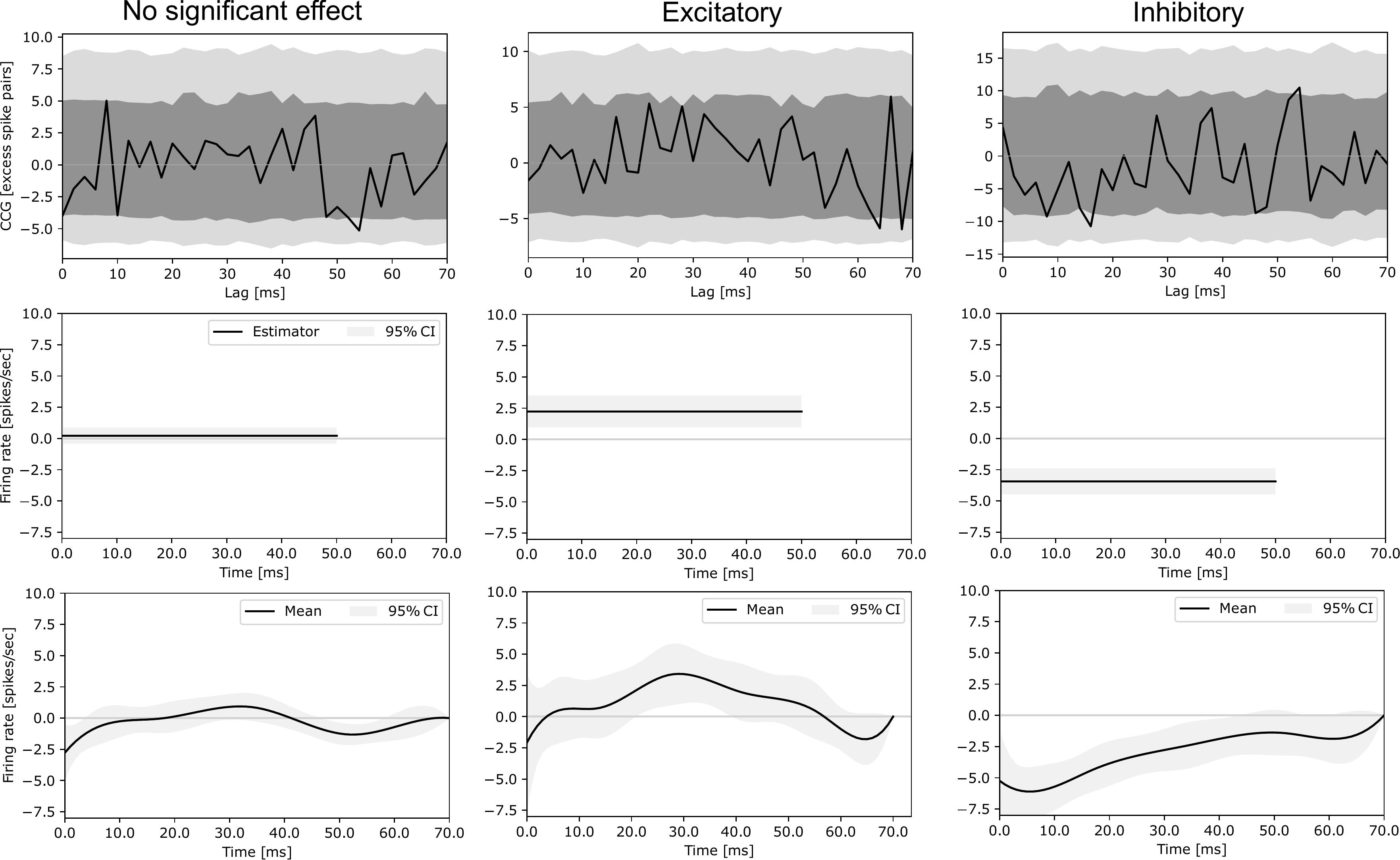}
\end{subfigure}
\caption{Comparison between impact functions and jitter-based CCG.
The calculation of the jitter-based CCG in the first row is the same as Figure \ref{fig:sim_demo}.
The second row shows the fitted impact functions using square windows. The third row shows the fitted impact functions using the non-parametric method.
}
\label{fig:ccg_filters}
\end{figure}


\end{document}